%% file: main.tex
\documentclass[journal]{IEEEtran}
\usepackage{cite}

\ifCLASSINFOpdf
   \usepackage[pdftex]{graphicx}
   \DeclareGraphicsExtensions{.pdf,.jpeg,.png}
\else
  \usepackage[dvips]{graphicx}
   \DeclareGraphicsExtensions{.eps}
\fi
\usepackage[cmex10]{amsmath}
\usepackage{amsfonts}
\usepackage{amssymb}
\usepackage{amsthm}

\usepackage{algorithm}
\usepackage{algpseudocode}
\usepackage{algorithmicx}
\newsavebox{\ieeealgbox}

 \usepackage{breqn}
 \usepackage{color}
 
\usepackage{rotating}
\usepackage{enumitem}
\usepackage{multirow}
\usepackage{array}
\newcolumntype{L}[1]{>{\raggedright\let\newline\\\arraybackslash\hspace{0pt}}m{#1}}
\newcolumntype{C}[1]{>{\centering\let\newline\\\arraybackslash\hspace{0pt}}m{#1}}
\newcolumntype{R}[1]{>{\raggedleft\let\newline\\\arraybackslash\hspace{0pt}}m{#1}}
\usepackage{arydshln}

\ifCLASSOPTIONcompsoc
  \usepackage[caption=false,font=normalsize,labelfont=sf,textfont=sf]{subfig}
\else
  \usepackage[caption=false,font=footnotesize]{subfig}
\fi
\usepackage{fixltx2e}
 \usepackage{amsthm}
 \theoremstyle{plain}
 
\newtheorem{theorem}{Proposition}
\theoremstyle{definition}
\newtheorem{definition}{Definition}

\usepackage[export]{adjustbox}

\newcommand{\Prob}{{\mathsf{p}}}   
\newcommand{\Probm}{{\mathsf{P}}} 
\newcommand{\Normal}{{\mathsf{N}}}
\newcommand{\Expect}{{\mathsf{E}}}

\DeclareMathOperator*{\minimize}{\mathrm{minimize}}

\DeclareMathOperator{\subjto}{\mathrm{subject \, to}}

 \DeclareMathOperator*{\argmin}{\mathrm{argmin}}

\newcommand{\tran}{^{\scriptscriptstyle{\sf T}}} 

\newcommand{\inv}{^{\scriptscriptstyle{-1}}}

\newcommand{\diag}{\mathrm{diag}}

\newcommand{\Graph}{\mathcal{G}}
\newcommand{\iset}{\mathcal{S}}

\newcommand{\nvert}{n}

\newcommand{\nsamp}{k}

\newcommand{\Adj}{{\mathbf{W}}} 
\newcommand{\Conn}{\mathbf{A}} 
\newcommand{\SLoop}{\mathbf{V}} 

\newcommand{\lap}{{\mathbf{L}}}

\newcommand{\eigM}{{\mathbf{\Lambda}}}

\newcommand{\eye}{{\mathbf{I}}} 
\newcommand{\Covar}{{\mathbf{\Sigma}}} 
\newcommand{\Precision}{{\mathbf{\Omega}}} 

\newcommand{\OD}{\boldsymbol{\Theta}} 
\newcommand{\vo}{\boldsymbol{\theta}}

\newcommand{\CS}{\mathbf{S}} 

\newcommand{\vx}{{\mathbf{x}}} 
\newcommand{\vy}{{\mathbf{y}}} 
\newcommand{\vr}{{\mathbf{r}}} 
\newcommand{\vp}{{\mathbf{p}}} 

\newcommand{\vones}{{\mathbf{1}}}
\newcommand{\vzeros}{{\mathbf{0}}}
 
\newcommand{\bigO}{O} 

\newcommand{\gbt}{{\mathbf{U}}} 

\begin{document}

%
\title{Graph-based Transforms for Video Coding}
%
%
%

\author{Hilmi~E.~Egilmez,~\IEEEmembership{Member,~IEEE}, Yung-Hsuan~Chao,~\IEEEmembership{Member,~IEEE}
        and~Antonio~Ortega,~\IEEEmembership{Fellow,~IEEE}
\thanks{H.E.~Egilmez and Y.-H.~Chao are with Qualcomm Technologies Inc., San Diego, CA, 92121 USA. A.~Ortega is with the Department of Electrical Computer Engineering,
University of Southern California, Los Angeles,
CA, 90084 USA. Contact author e-mail: hegilmez@qti.qualcomm.com.}
\thanks{This work has been done while H.E.~Egilmez and Y.-H.~Chao were with University of Southern California. Earlier versions of the work were presented at the IEEE International Conference on Image Processing in \cite{egilmez:15:gbt_inter,egilmez:16:gbst} and in the PhD dissertations \cite{hegilmez_phd_thesis,jessie_phd_thesis}.}
}

\maketitle

\begin{abstract}
In many state-of-the-art compression systems, signal transformation is an integral part of the encoding and decoding process, where transforms provide compact representations for the signals of interest. This paper introduces a class of transforms called graph-based transforms (GBTs) for video compression, and proposes two different techniques to design GBTs. In the first technique, we formulate an optimization problem to learn graphs from data and provide solutions for optimal separable and nonseparable GBT designs, called GL-GBTs. The optimality of the proposed GL-GBTs is also theoretically analyzed based on Gaussian-Markov random field (GMRF) models for intra and inter predicted block signals. The second technique develops edge-adaptive GBTs (EA-GBTs) in order to flexibly adapt transforms to block signals with image edges (discontinuities). The advantages of EA-GBTs are both theoretically and empirically demonstrated. Our experimental results show that the proposed transforms can significantly outperform the traditional Karhunen-Loeve transform (KLT).
\end{abstract}


%
\IEEEpeerreviewmaketitle

\let\oldemptyset\emptyset
\let\emptyset\varnothing

 \setlength{\parskip}{0mm plus0mm minus0mm}
\setitemize[0]{leftmargin=15pt,itemindent=0pt}

\addtolength{\textfloatsep}{-0.5cm}
\setlength{\abovecaptionskip}{2pt} 
\setlength\abovedisplayskip{1.5pt}
\setlength\belowdisplayskip{1.5pt}

\section{Introduction}
\label{sec:intro}
Predictive transform coding is a fundamental compression technique 
adopted in many block-based image and video compression systems, 
where block signals are initially predicted from a set of available (already coded) reference pixels, 
and then the resulting residual block signals are transformed (generally by a linear transformation)
to decorrelate residual pixel values for effective compression. After prediction and transformation steps, a typical image/video compression system applies quantization and entropy coding to convert transform coefficients into a stream of bits. 
Fig.~\ref{fig:encoder_decoder} illustrates a representative encoder-decoder architecture comprising three basic components, (i) prediction, (ii) transformation, (iii) quantization and entropy coding, which are implemented in state-of-the-art compression standards such as JPEG \cite{JPEG:1992:Pennebaker}, HEVC \cite{Sullivan:12:hevc}, VP9 \cite{Mukherjee:13:vp9}, {AV1 \cite{av1_standard} and VVC \cite{Bross:20:vvc10}.}
This paper focuses mainly on the transformation component of video coding and develops techniques to design orthogonal transforms, called graph-based transforms (GBTs), adapting diverse characteristics of video signals.  

\begin{figure}[!t]
        \centering \includegraphics[width=0.48\textwidth]{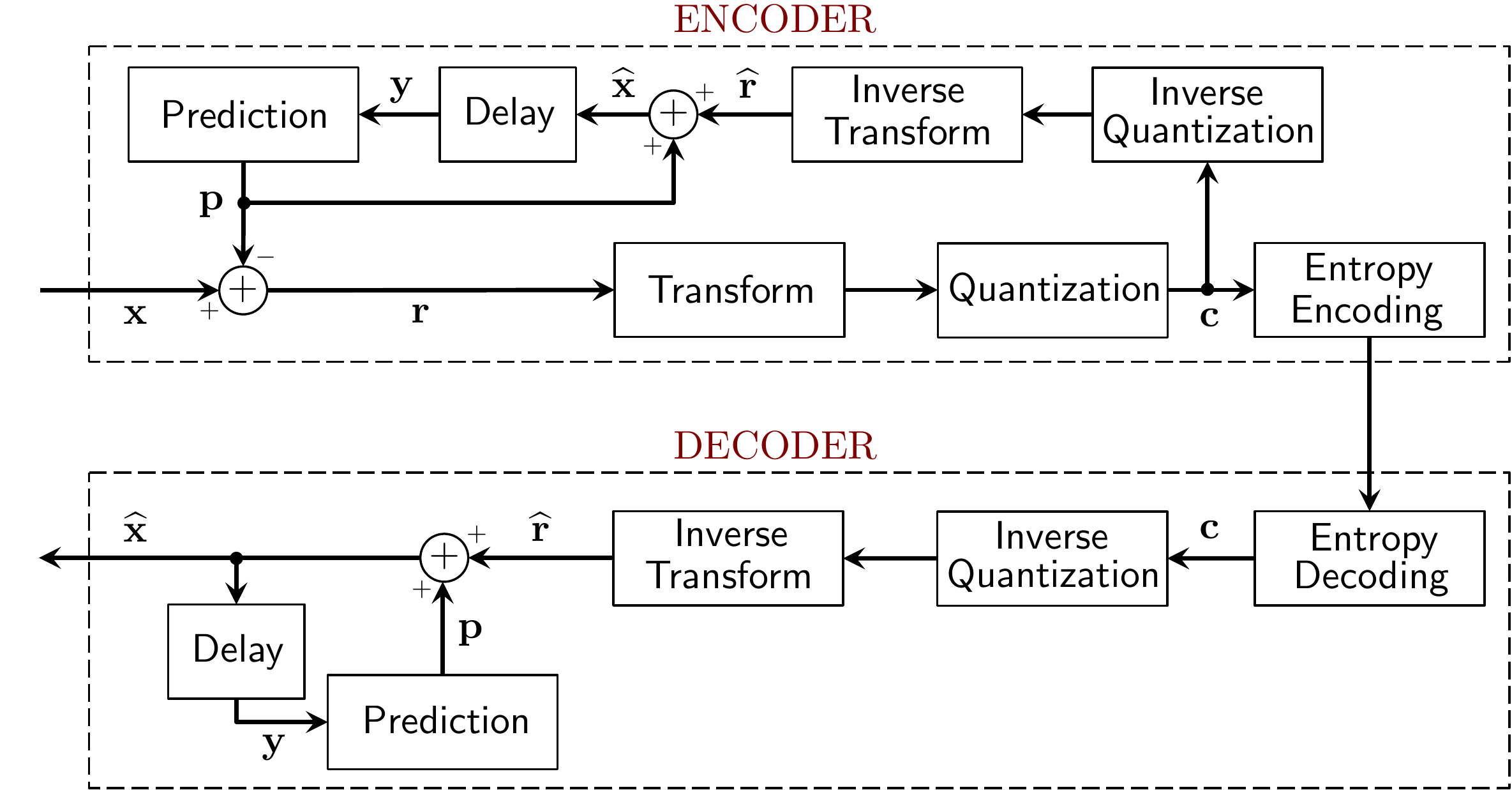}
                \caption{Building blocks of a typical video encoder and decoder consisting of three main steps, which are (i) prediction, (ii) transformation, (iii) quantization and entropy coding. } 
           	\label{fig:encoder_decoder}
\end{figure}

In predictive transform coding of video, the prediction is generally carried out by choosing one among multiple intra and inter prediction modes in order to exploit spatial and temporal redundancies between block signals. 
On the other hand, for the transformation, typically a single transform such as the discrete cosine transform (DCT-2) is applied in a separable manner to rows and columns of each residual block. 
The main problem of using fixed block transforms is the implicit assumption 
that all residual blocks share the same statistical properties. 
However, residual blocks can have very diverse statistical characteristics depending on the video content and the prediction mode (as will be demonstrated by some of the experiments in Section \ref{sec:results}). 
The HEVC \cite{Sullivan:12:hevc} standard, partially addresses this problem by  allowing the use of the asymmetric discrete sine transform (ADST or DST-7) in addition to the DCT-2 for small (i.e., $4\times4$) intra predicted blocks \cite{Han:12:hybrid}. 
Yet, it has been shown that better compression can be achieved by using data-driven transform designs that adapt to statistical properties of residual blocks \cite{Karczewicz:08:mddt,Takamura:13:intra_mddt,Arrufat:2014:intra_mddt,oscar_au:13:rdo_loyd-type,xin:2016:nnst,said:2016:hygt,Zhao:2018:joint_primary_secondary,Pierrick:2018:transform_inter,chao:16:eagbt_step_ramp,egilmez:2019:ce6_report}.

The majority of prior studies about transforms for video coding focus on developing transforms for intra predicted residuals. 
In \cite{Karczewicz:08:mddt}, a mode-dependent transform (MDT) scheme is proposed by designing a Karhunen-Loeve transform (KLT) for each intra prediction mode. In \cite{Takamura:13:intra_mddt,Arrufat:2014:intra_mddt}, the MDT scheme is similarly implemented for the HEVC standard, where a single KLT is trained for each intra prediction mode offered in HEVC. Moreover in \cite{xin:2016:emt,xin:2016:nnst,said:2016:hygt,Zhao:2018:joint_primary_secondary,Pierrick:2018:transform_inter}, the authors demonstrate considerable coding gains over the MDT method by using the rate-distortion optimized transformation (RDOT) scheme, which suggests designing multiple transforms for each prediction mode so that the encoder can select a transform (from the predetermined set of transforms) by minimizing a rate-distortion (RD) cost. {More recently, the VVC standard \cite{Bross:20:vvc10} has adopted simplified variants of the  RDOT-based methods in \cite{xin:2016:emt,xin:2016:nnst}.} 
Since the RDOT scheme allows the flexibility of selecting 
a transform on a per-block basis at the encoder side, it provides better adaptation to residual blocks with different statistical characteristics as compared to the MDT scheme. However, all of these methods rely on KLTs derived from sample covariance matrices, which may not be good estimators for the true covariances of models, especially when the number of data samples is small \cite{johnstone2009consistency,ravikumar:2011:bounds_inverse_cov_estimation}. Indeed, it has been shown that more accurate model estimates can be obtained using inverse covariance estimation methods \cite{friedman:2007:glasso,ravikumar:2011:bounds_inverse_cov_estimation} or graph learning methods, such as those introduced in our prior work \cite{egilmez:2017:gl_from_data_jstsp,egilmez:2019:gsi_tsipn}.

This paper proposes a novel graph-based modeling framework to design GBTs, where the models of interest are based on Gaussian-Markov random fields (GMRFs) whose inverse covariances are graph Laplacian matrices. 
The proposed framework consists of two distinct techniques to develop GBTs for video coding, called GL-GBTs and EA-GBTs:
\begin{itemize}
\item \emph{Graph learning for GBT {(GL-GBT)} design:} 
A graph learning problem with a maximum-likelihood (ML) criterion is formulated and solved to estimate a graph Laplacian matrix from training data. In order to construct separable and nonseparable GBTs, two instances of the proposed problem with different connectivity constraints are solved by applying the graph learning algorithm introduced 
in our prior work \cite{egilmez:2017:gl_from_data_jstsp}.
Then, the GBTs are constructed by the eigendecomposition of the graph Laplacians of the learned graphs. As the KLT, a {GL-GBT} is learned from a sample covariance, but in addition, it incorporates Laplacian and structural constraints reflecting the inherent model assumptions about the video signal. {From a statistical learning perspective, the main advantage of the proposed {GL-GBT} over the KLT is that the {GL-GBT} requires learning fewer model parameters from training data, and thus can lead to a more robust transform allowing better compression for the block signals outside of the training dataset. GL-GBTs can be adopted to improve existing MDT or RDOT schemes using KLTs.}
\item \emph{Edge-adaptive GBT {(EA-GBT)} design:} To adapt transforms for block signals with image edges\footnote{We use the term \emph{image edge} to distinguish edges in image/video signals with edges in graphs.}, 
we develop edge-adaptive GBTs (EA-GBTs) which are designed on a per-block basis. These lead to a block-adaptive transform scheme, where transforms are derived from graph Laplacians whose weights are modified based on image edges detected in a residual block. 
\end{itemize}
 
In the literature, there are a few studies on model-based transform designs for image and video coding. In \cite{Florencio:2013:PTC,microsoft:15:gsp_prob_framework}, the authors present a graph-based probabilistic framework for predictive video coding and use that to justify the optimality of DCT, yet optimal graph/transform design is out of their scope. In our previous work \cite{pavez:2015:gtt_paper}, we present a comparison of various instances of different graph learning problems for nonseparable image modeling. The present paper theoretically and empirically validates one of the conclusions in \cite{pavez:2015:gtt_paper}, which suggests the use of GBTs derived from graph Laplacian matrices for image compression. 
In \cite{Fracastoro:2018:apprx_graph}, a block-adaptive scheme is proposed for image compression with GBTs, where a graph is constructed for each block signal by minimizing a regularized Laplacian quadratic term used as the proxy for actual RD cost. 
On the other hand, in our present work, graphs are 
estimated from aggregate data statistics 
or constructed using image edge information. 
Then, the encoder selects the best transform based on exact RD measures.
Moreover, Shen \emph{et.~al.~}\cite{Shen:10:eat} propose edge adaptive transforms (EAT) specifically for depth-map compression, and Hu \emph{et.~al.~}\cite{hu:14:PWS} extend EATs for piecewise-smooth image compression. Although our proposed EA-GBTs adopt some basic concepts originally introduced in \cite{Shen:10:eat}, our graph construction method is different (in terms of image edge detection) and provides better compression for residual signals. 
 
The main contributions of this paper can be summarized as follows:
\begin{itemize}
\item We propose graph learning techniques   
for separable and nonseparable {GBTs} (i.e., GBSTs and GBNTs, respectively) and present a theoretical justification of their use for coding residual block signals modeled using GMRFs.
\item As an extension of the 1-D GMRF models used to design GBSTs for intra and inter predicted signals in our previous work \cite{egilmez:16:gbst}, we present a general 2-D GMRF model for GBNTs and analyze its optimality.
\item We apply EA-GBTs to intra and inter predicted blocks, 
while our prior work in \cite{egilmez:15:gbt_inter} focused only on inter predicted blocks. 
In addition to the experimental results, we further derive some theoretical results and discuss the cases in which EA-GBTs are useful. 
\item We present comprehensive experimental results comparing the compression performances obtained using KLTs, GL-GBTs and EA-GBTs. 
\end{itemize}

The rest of the paper is organized as follows. Section \ref{sec:notation_prelim} presents the basic notation and definitions used throughout the paper. Section \ref{sec:models} introduces 2-D GMRFs used for modeling the video signals and discusses graph-based interpretations. In Section \ref{sec:graph_learning_video}, the GBT design problem is formulated as a graph Laplacian estimation problem, and solutions for optimal GBT construction are proposed. Section \ref{sec:eagbt} presents EA-GBTs. Graph-based interpretations of residual block signal characteristics are discussed in Section \ref{sec:graph_res_charac} by empirically validating the theoretical observations. Experimental results are presented in Section \ref{sec:results}, and Section \ref{sec:conclusion} draws concluding remarks.

\section{Notation and Preliminaries}
\label{sec:notation_prelim}
Throughout the paper, lowercase normal (e.g., $a$ and $\theta$), lowercase bold (e.g., $\mathbf{a}$ and $\vo$) and uppercase bold (e.g., $\mathbf{A}$ and $\OD$) letters denote scalars, vectors and matrices, respectively. Unless otherwise stated, calligraphic capital letters (e.g., $\mathcal{E}$ and $\mathcal{S}$) represent sets. 
Notation is summarized in Table \ref{table:notation}.
\begin{table}[!ht]
\caption{List of Symbols and Their Meaning}
\label{table:notation}
\centering
\begin{tabular}{|c||c|}
 \hline 
   Symbols &  \multicolumn{1}{c|}{Meaning}  \\
 \hline  \hline
$\Graph$ $\mid$ $\lap$ &  weighted graph $\mid$ graph Laplacian matrix  \\ \hline
$\nvert$ $\mid$ $N$ &  number of vertices $\mid$ block size ($N\times N$)  \\ \hline
$\eye$ $\mid$ $\vones$ & identity matrix $\mid$ vector of all ones \\ \hline
$\OD\tran$ $\mid$ $\vo\tran$ &  transpose of $\OD$ $\mid$ transpose of $\vo$  \\ \hline
$(\OD)_{ij}$ & entry of $\OD$ at $i$-th row and $j$-th column \\ \hline
$(\vo)_{i}$ & $i$-th entry of $\vo$ \\ \hline 
$(\vo)_{\iset}$ & subvector of $\vo$ formed by selecting indexes in $\iset$ \\ \hline
$\geq $ ($\leq$) & element-wise greater (less) than or equal to operator  \\ \hline 
$\OD \succeq 0$ & $\OD$ is a positive semidefinite matrix  \\ \hline 
$\OD\inv$ $\mid$ $\mathrm{det}(\OD)$  &  inverse of $\OD$ $\mid$  determinant of $\OD$   \\ \hline
$\mathrm{Tr}$ $\mid$ $\mathrm{logdet}(\OD)$ & trace operator  $\mid$ natural logarithm of $\mathrm{det}(\OD)$   \\ \hline 
$\mathrm{diag}(\vo)$ & diagonal matrix formed by elements of $\vo$    \\ \hline 
$\mathrm{ddiag}(\OD)$ & diagonal matrix formed by diagonal elements of $\OD$    \\ \hline 
$\smash{\vx \sim \Normal(\vzeros,\Covar)}$ & zero-mean multivariate Gaussian with covariance $\Covar$ \\ \hline
\end{tabular}
\end{table}

In this paper, we focus on connected, undirected, weighted simple graphs with nonnegative edge weights \cite{Chung:1997:SGT}. We next present basic definitions related to graphs. 
\begin{definition}[Weighted Graph] 
The graph $\Graph \! = \! (\mathcal{V},\mathcal{E},f_w,f_v)$ is a weighted graph with $\nvert$ vertices in the set $\mathcal{V} \! = \! \{ v_1,\dotsc,v_\nvert \} $. The edge set $\mathcal{E}\! = \! \{ \, e  \, | \, f_w(e) \! \neq \! 0  , \, \forall \, e \! \in \! \mathcal{P}_u \}$ is a subset of $\mathcal{P}_u$, the set of 
all unordered pairs of vertices, where $f_w((v_i,v_j)) \! \geq \! 0 \text{ for } i \! \neq \! j$ is a real-valued edge weight function, and $f_v(v_i) \text{ for } i\!=\!1,\dotsc,\nvert$ is a real-valued vertex (self-loop) weight function. 
\label{def:weighted_graph}
\end{definition}
\begin{definition}[Algebraic representations of graphs]
\label{def:useful_matrices}
For a given weighted graph $\Graph\!=\!(\mathcal{V},\mathcal{E},f_w)$ with $\nvert$ vertices, $v_1,\dotsc,v_\nvert$: \\
The \emph{adjacency matrix} of $\Graph$ is an $\nvert  \times \nvert$ symmetric matrix, $\Adj$, such that $(\Adj)_{ij} = (\Adj)_{ji} = f_w((v_i,v_j))$ for $(v_i,v_j) \in \mathcal{P}_u$. \\
The \emph{degree matrix} of $\Graph$ is an $\nvert\times \nvert$ diagonal matrix, $\mathbf{D}$, with entries $(\mathbf{D})_{ii}= \sum_{j=1}^{\nvert} (\Adj)_{ij}$ and $(\mathbf{D})_{ij} = 0$ for $i \neq j$. 
\\
The \emph{self-loop matrix} of $\Graph$ is an $\nvert\times \nvert$ diagonal matrix, $\SLoop$, with entries $(\SLoop)_{ii}=f_v(v_i)$ for $i=1,\dotsc,\nvert$ and $(\SLoop)_{ij} = 0$ for $i \neq j$. If $\Graph$ is a simple weighted graph, then $\mathbf{V} = \mathbf{0} $. 
\\
The \emph{connectivity matrix} of $\Graph$ is an $\nvert \times \nvert$ matrix, $\Conn$, such that $(\Conn)_{ij}=1$ if $(\Adj)_{ij} \neq 0$, and  $(\Conn)_{ij}=0$ if $(\Adj)_{ij} = 0$ for $i,j=1,\dotsc,\nvert$, where  $\Adj$ is the adjacency matrix of $\Graph$.
\\
The \emph{combinatorial graph Laplacian (CGL)} of graph $\Graph$ is defined as $\lap = \mathbf{D}  - \Adj$.
\\
The \emph{generalized graph Laplacian (GGL)} of graph $\Graph$ is defined as $\lap = \mathbf{D}  - \Adj + \SLoop$, which reduces to the combinatorial graph Laplacian when there are no self-loops ($\SLoop \!=\! {\mathbf{0}}$). 
\end{definition}


\begin{definition}[Graph-based Transform (GBT)] 
Let $\lap$ be a graph Laplacian of a graph $\Graph$. The graph-based transform 
is the orthogonal matrix $\gbt$, satisfying $\smash{\gbt\tran\gbt=\eye}$, obtained 
by eigendecomposition of $\smash{\lap \! = \! \gbt \eigM \gbt\tran} $, 
where $\eigM$ is the diagonal matrix consisting of eigenvalues of $\lap$ (graph frequencies). 
\label{def:gbt_video} 
\end{definition}

\noindent As formally stated in the following proposition, GBTs are invariant under (i) constant scaling of graph weights and (ii) addition of a constant self-loop weight to all vertices.

\begin{theorem}
Let $\gbt$ be a GBT diagonalizing graph Laplacian $\lap$. The same $\gbt$ also diagonalizes Laplacians of the form $\widetilde{\lap} =c_1 \lap + c_2 \eye$, where $c_1$ and $c_2$ are  real-valued scalars. 
\end{theorem}
\begin{proof}
The proof is straightforward from Definition \ref{def:gbt_video}. \end{proof}

\section{Graph-based Models and Transforms for Video Block Signals}
\label{sec:models}
For modeling video block signals, we use Gaussian Markov random fields (GMRFs), which provide a probabilistic interpretation for our graph-based framework. Assuming that the random vector of interest $\mathbf{x} \in \mathbb{R}^n$ has zero mean\footnote{The zero mean assumption is made to simplify the notation. The models can be trivially extended to GMRFs with nonzero mean.}, a GMRF model for $\mathbf{x}$ is defined based on a precision matrix $\Precision$, so that $\mathbf{x}$ has a multivariate Gaussian distribution, ${\mathbf{x}} \sim \Normal(\mathbf{0}, \Precision\inv)$,
\begin{equation}
\label{eqn:GMRF_general_video}
\Prob({\mathbf{x}}| \Precision) = \frac{1}{{(2\pi)}^{n/2} {\mathrm{det}({\Precision})}^{-1/2} } \mathrm{exp} \hspace{-0.1cm} \left( {- \frac{1}{2} {\mathbf{x}}\tran {\Precision}{\mathbf{x}}} \right).
\end{equation} 
with covariance matrix $\smash{\Covar \! = \! \Precision\inv}$. 
The entries of the precision matrix $\Precision$ can be interpreted in terms of conditional dependence relations among variables,
\begin{align}
\Expect \left[ { x_i \! \mid \! (\vx)_{\mathcal{S}\backslash\{ i \} } } \right] & = - \frac{1}{(\Precision)_{ii}} \sum_{j\in{\mathcal{S}\backslash\{ i \} }} (\Precision)_{ij} x_j  \label{eqn:prec_expectation}  \\
\; \mathsf{Prec} \left[ { x_i \! \mid \! (\vx)_{\mathcal{S}\backslash\{ i \} } } \right] & = (\Precision)_{ii}  \label{eqn:prec_element}  \\
 \mathsf{Corr} \left[ { x_i x_j \! \mid \! (\vx)_{\mathcal{S}\backslash\{ i,j \}} } \right]   & =  - \frac{(\Precision)_{ij}}{\sqrt{(\Precision)_{ii} (\Precision)_{jj} }} \quad i \neq j, \label{eqn:partial_corr}
\end{align}
where $\mathcal{S} \! = \! \{1,\dotsc,\nvert\}$ is the index set for $\vx = {[ x_1,\dotsc,x_\nvert ]}\tran$. The conditional expectation in (\ref{eqn:prec_expectation}) gives the best minimum mean square error (MMSE) estimate of $x_i$ using all other random variables. The relation in (\ref{eqn:prec_element}) corresponds to the \emph{precision} of $x_i$ and (\ref{eqn:partial_corr}) to  the \emph{partial correlation} between $x_i$ and $x_j$ (i.e., correlation between random variables $x_i$ and $x_j$ given all other variables in $\vx$). For example, if $x_i$ and $x_j$ are conditionally independent (i.e., $\smash{(\Precision)_{ij} \! = \!0}$), there is no edge between corresponding vertices $v_i$ and $v_j$. 
If all partial correlations are nonnegative (i.e., off-diagonal elements of ${\Precision}$ are nonpositive), then the model in (\ref{eqn:GMRF_general_video}) is classified as \emph{attractive GMRF}  \cite{Koller:2009:PGM,rue:2005:GMRF,egilmez:2017:gl_from_data_jstsp}, whose precision matrix satisfies the following proposition \cite{egilmez:2017:gl_from_data_jstsp,egilmez:16:gbst}. 
\begin{theorem}
A GMRF model is attractive if and only if its precision matrix $\Precision$ is a generalized graph Laplacian matrix. 
\label{proposition:attractive}
\end{theorem}
\begin{proof}
The proof is straightforward by the definitions. 
\end{proof}

\begin{figure}[!t]
\centering
    \subfloat[ 
    \label{fig:intra_model_1d}]{%
      \hspace{-0.9cm}\includegraphics[width=5cm]{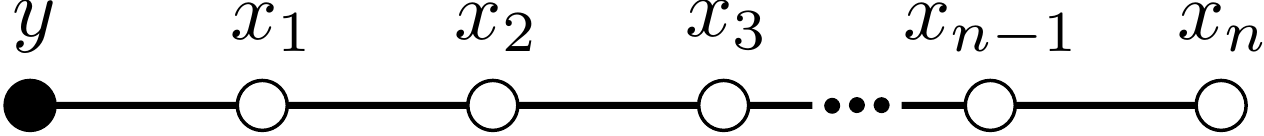}}
      \\
    \subfloat[ 
    \label{fig:inter_model_1d}]{%
     \includegraphics[width=4.1cm]{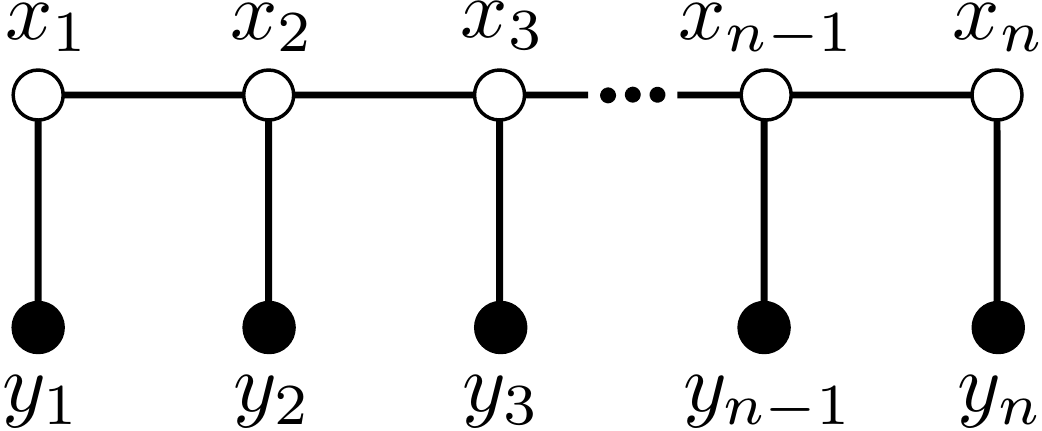}}
\caption{1-D GMRF models for (a) intra and (b) inter predicted signals.
Black filled vertices represent the reference pixels and unfilled vertices denote pixels to be predicted and then transform coded.}
\vspace{-0.25cm} 
\label{fig:intra_inter_models_1D}
\centering
    \subfloat[ \label{fig:intra_model_2d}]{%
    \includegraphics[height=3.8cm]{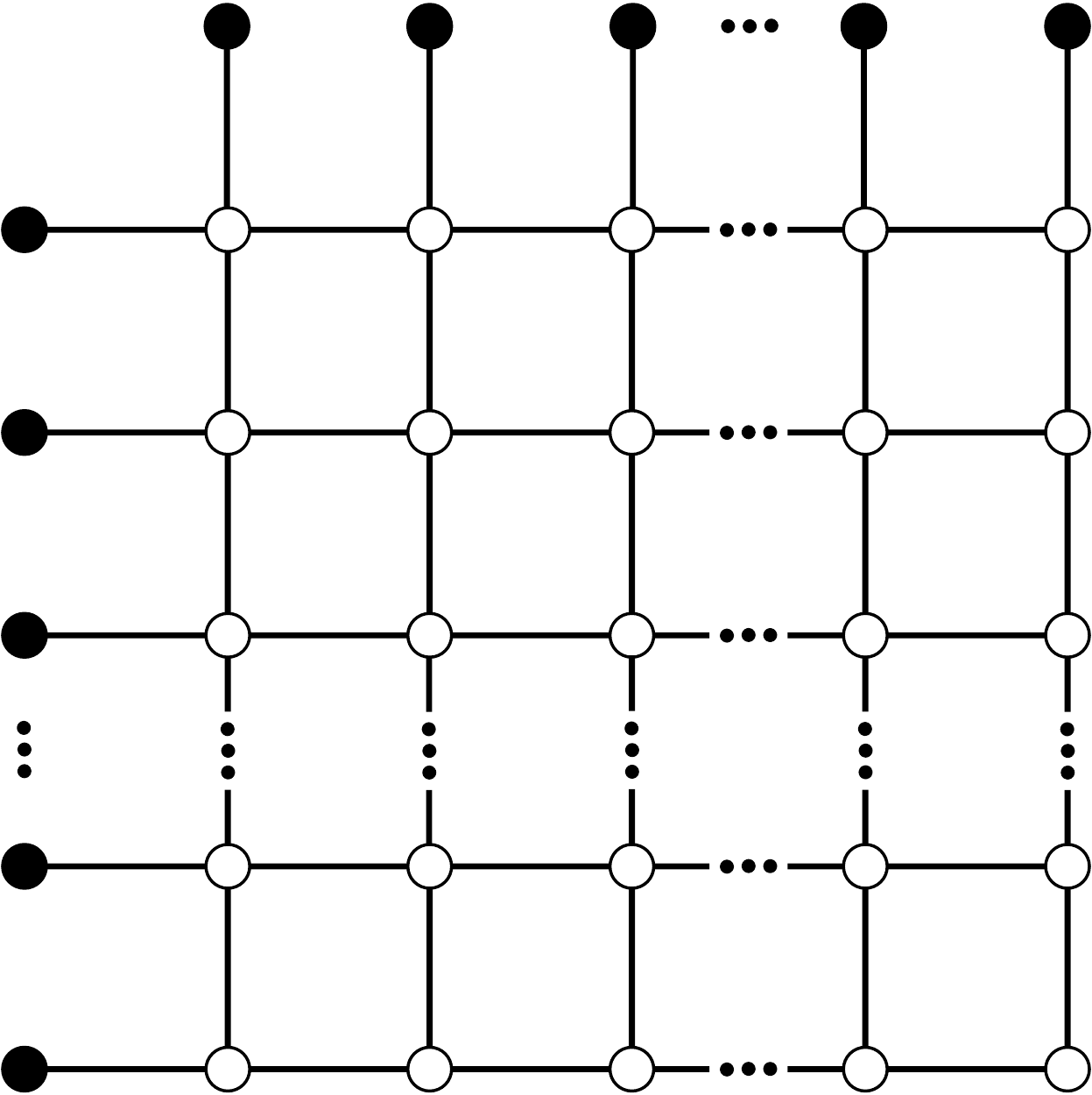}}
    \qquad 
    \subfloat[ \label{fig:inter_model_2d}]{%
    \includegraphics[height=4.8cm]{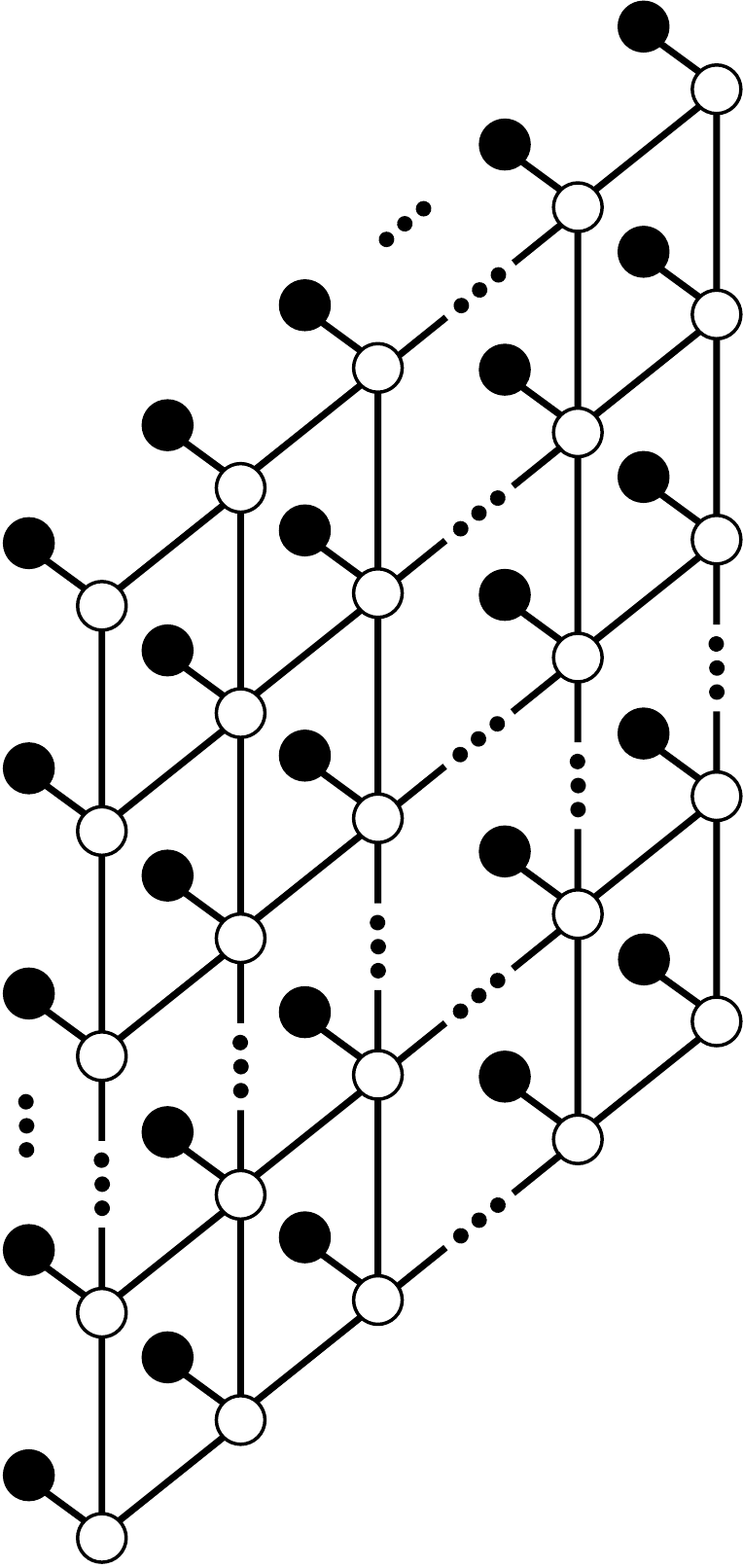}}
\caption{2-D GMRF models for (a) intra and (b) inter predicted signals. Black filled vertices correspond to reference pixels obtained (a) from neighboring blocks and (b) from other frames via motion compensation. Unfilled vertices denote the pixels to be predicted and then transform coded.}   
\label{fig:intra_inter_models_2D}
\end{figure}

In statistical modeling of image/video signals, it is generally assumed that adjacent pixel values are positively correlated  \cite{Jain:1989:FDI,Tekalp:2015:DVP}. The assumption is intuitively reasonable for video signals, since neighboring pixel values are often similar to each other due to spatial and temporal redundancy. With this general assumption, we propose attractive GMRFs to model intra/inter predicted video block signals. 
 Figs.~\ref{fig:intra_inter_models_1D} and \ref{fig:intra_inter_models_2D} illustrate the 
1-D and 2-D GMRFs {defining \emph{line} and \emph{grid} graphs that are} used to design separable and nonseparable GBTs, respectively. 
 
We formally define GBST and GBNT as follows. 
\begin{definition} [Graph-based Separable Transform--GBST] Let $\mathbf{U}_{\textnormal{row}}$ and $\mathbf{U}_{\textnormal{col}}$ be $N \times N$ GBTs associated with two line graphs with $N$ vertices, then the GBST of an $N \! \times \! N$ matrix $\mathbf{X}$ is  
\begin{equation}
\widehat{\mathbf{X}} = \mathbf{U}_{\textnormal{col}}\tran \mathbf{X} \mathbf{U}_{\textnormal{row}}, 
\label{eqn:gbst_video}
\end{equation}
where $\mathbf{U}_{\textnormal{row}}$ and $\mathbf{U}_{\textnormal{col}}$ 
are {the transforms applied to} rows and columns of the block signal $\mathbf{X}$, respectively.
\end{definition}
\begin{definition}[Graph-based Nonseparable Transform--GBNT] 
Let $\gbt$ be an $N^2 \! \times \! N^2$ GBT associated with a graph with $N^2$ vertices, then the GBNT of an $N \! \times \! N$ matrix $\mathbf{X}$ is 
\begin{equation}
 \widehat{\mathbf{X}} = \mathrm{block}(\gbt\tran \mathrm{vec}(\mathbf{X})),
\label{eqn:gbnt_video}
\end{equation}
 where $\gbt$ is applied on vectorized signal $\vx = \mathrm{vec}(\mathbf{X})$, and the $\mathrm{block}$ operator restructures the signal back in block form. 
\label{def:gbnt_video} 
\end{definition}

In our previous work \cite{egilmez:16:gbst}, we introduced the 1-D GMRFs illustrated in Fig.~\ref{fig:intra_inter_models_1D} for intra/inter predicted signals and also derived closed-form expressions of their precision matrices (i.e., $\Precision$). The following section presents a unified 2-D extension of those models. While the work in \cite{microsoft:15:gsp_prob_framework} has noted the relation between graph Laplacians and GMRFs in the context of predictive transform coding, the following section further shows that residual signals obtained from MMSE prediction follow an attractive GMRF model. Hence, the optimal linear transform decorrelating such residual signals is a GBT derived from a GGL.

\subsection{2-D GMRF Model for Residual Signals} 
We introduce a general 2-D GMRF model for intra/inter predicted  $N \! \times \! N$ block signals depicted in Fig.~\ref{fig:intra_inter_models_2D} by deriving the precision matrix of the residual signal $\vr$, obtained after predicting the signal $\smash{\vx = [x_1 \; x_2 \; \cdots \; x_n]\tran}$ with $\nvert\!=\!N^2$
from $\nvert_p$ reference samples in $\smash{\vy = [y_1 \; y_2 \; \cdots \; y_{n_p}]\tran}$ (i.e.,  predicting  unfilled vertices from black filled vertices in Fig.~\ref{fig:intra_inter_models_2D}), where $\vx$ and $\vy$ are zero-mean and jointly Gaussian with respect to the following attractive 2-D GMRF:
\begin{equation}
\label{eqn:GMRF_model_xy}
\Prob(\left[{\begin{smallmatrix}\vx \\ \vy \end{smallmatrix}}\right] | \Precision) = \frac{1}{{(2\pi)}^{n/2} {\mathrm{det}({\Precision})}^{-1/2} } \mathrm{exp} \hspace{-0.1cm} \left( {- \frac{1}{2} \left[{\begin{smallmatrix}\vx \\ \vy \end{smallmatrix}}\right]\tran  {\Precision}\left[{\begin{smallmatrix}\vx \\ \vy \end{smallmatrix}}\right]} \right).
\end{equation}
The precision matrix $\Precision$ and the covariance matrix $\Covar = \Precision\inv$ can be partitioned as follows \cite{rue:2005:GMRF,microsoft:15:gsp_prob_framework}:
\begin{equation}
\label{eqn:precision_partition_xy}
\Precision = \begin{bmatrix}
\Precision_{\vx} & \Precision_{\vx \vy} \\
\Precision_{\vy \vx} & \Precision_{\vy} \\
\end{bmatrix} = 
\begin{bmatrix}
\Covar_{\vx} & \Covar_{\vx \vy} \\
\Covar_{\vy \vx} & \Covar_{\vy} \\
\end{bmatrix}\inv = \Covar\inv.
\end{equation}
Irrespective of the type of prediction (intra/inter), the MMSE prediction of $\vx$ from the reference samples in $\vy$ is
\begin{equation}
\vp=\Expect[{\vx|\vy}] = \Covar_{\vx\vy} \Covar\inv_{\vy}\vy = -\Precision_{\vx}\inv \Precision_{\vx\vy}\vy,
\end{equation}
and the resulting residual vector is $\vr = \vx - \vp$ with covariance 
\begin{equation}
\begin{aligned}
\Covar_{\vr} &  = \Expect[{\vr\vr\tran}] = \Expect[{(\vx-\vp) (\vx-\vp)\tran}] \\
&  = \Expect[{\vx\vx\tran + \vp\vp\tran - 2\vx\vp\tran  }] \\
& = \Covar_{\vx} +  \Covar_{\vx\vy} \Covar\inv_{\vy} \Covar_{\vy\vx} -  2 \Covar_{\vx\vy} \Covar\inv_{\vy} \Covar_{\vy\vx} \\
& = \Covar_{\vx} - \Covar_{\vx\vy} \Covar_{\vy}\inv  \Covar_{\vy\vx}.
\end{aligned}
\end{equation}
By the matrix inversion lemma  \cite{Woodbury:1950:MIL}, the precision matrix of the residual $\vr$ is shown to be equal to $\Precision_{\vx}$, that is the submatrix in (\ref{eqn:precision_partition_xy}), 
\begin{equation}
\Covar_{\vr}\inv =  (\Covar_{\vx} - \Covar_{\vx\vy} \Covar_{\vy}\inv  \Covar_{\vy\vx})\inv = 
\Precision_{\vx}.
\end{equation} 
Since we also have $\smash{\Covar_{\vx} = (\Precision_{\vx}-\Precision_{\vx\vy}\Precision_{\vy}\inv\Precision_{\vy\vx})\inv}$ by  \cite{Woodbury:1950:MIL}, the desired precision matrix can also be written as 
\begin{equation}
\Precision_{\text{residual}} = \Covar_{\vr}\inv = \Precision_{\vx} = \Covar_{\vx}\inv + \Precision_{\vx\vy}\Precision_{\vy}\inv\Precision_{\vy\vx}.
\label{eqn:residual_GMRF_general_form}
\end{equation}
This construction leads us to the following proposition formally stating the conditions for a residual signal (i.e., $\vr$) to be modeled by an attractive GMRF.
\begin{theorem}
Let the signals $\vx\!\in\!\mathbb{R}^\nvert$ and $\vy\!\in\!\mathbb{R}^{\nvert_p}$ be distributed 
based on the attractive GMRF model in (\ref{eqn:GMRF_model_xy}) with precision matrix $\Precision$.
If the residual signal $\vr$ is estimated by minimum mean square error (MMSE) 
prediction of $\vx$ from $\vy$ (i.e., $\vr\!=\!\vx \!-\! \Expect[{\vx|\vy}]$), 
then the residual signal $\vr$ is distributed as an attractive GMRF whose precision matrix is a generalized graph Laplacian (i.e., $\Precision_{\text{\normalfont residual}}$ in (\ref{eqn:residual_GMRF_general_form})). 
\label{proposition:ggl_residual}
\end{theorem}
\begin{proof}
The proof follows from (\ref{eqn:GMRF_model_xy})--(\ref{eqn:residual_GMRF_general_form}) where the inverse covariance of residual signal $\vr$, $\Covar_{\vr}\inv$, is shown to be equal to $\Precision_{\vx}$. Since $\Precision_{\vx}$ is a submatrix of $\Precision$ in (\ref{eqn:precision_partition_xy}) and $\Precision$ is a GGL, $\Precision_{\text{\normalfont residual}}=\Precision_{\vx}$ is also a GGL. Hence, $\vr$ is distributed as an attractive GMRF whose precision is $\Precision_{\vx}$.
\end{proof}
 Note that Proposition \ref{proposition:ggl_residual} also applies to the 1-D signal models presented in \cite{egilmez:16:gbst} which are special cases of (\ref{eqn:GMRF_model_xy}).

\subsection{Interpretation of Graph Weights for Predictive Transform Coding}
\label{subsec:interpret_graph_weights}
 
Based on Proposition \ref{proposition:ggl_residual}, the distribution of residual signals, denoted as $\vr = {[ r_1 \cdots r_\nvert ]}\tran$, is defined by the following GMRF whose precision matrix is a GGL matrix $\lap$ (i.e., $\lap$ = $\Precision_{\text{\normalfont residual}}$), 
\begin{equation}
\label{eqn:GMR_residual_laplacian}
\Prob(\vr | \lap) =  \frac{1}{{(2\pi)}^{n/2} {\mathrm{det}({\lap})}^{-1/2}} \mathrm{exp} \hspace{-0.1cm} \left( {- \frac{1}{2} \vr\tran {\lap}\vr} \right),
\end{equation} 
where the quadratic term in the exponent can be decomposed in terms of graph weights {(i.e., $\SLoop$ and $\Adj$)} as 
\begin{equation}
\vr\tran \lap \vr  =  \sum_{i=1}^{\nvert}(\SLoop)_{ii} \, r_i^2  + \sum_{ (i,j) \in \mathcal{I} } (\Adj)_{ij} \, {(r_i - r_j)}^2
\label{eqn:GMRF_laplacian_quadratic_term}
\end{equation} 
such that $(\Adj)_{ij} \! = \! -(\lap)_{ij} $, $(\SLoop)_{ii} \! = \! \sum_{j=1}^\nvert (\lap)_{ij} $, and $\mathcal{I}  \! = \! \{ (i,j) \, | \, (v_i,v_j) \! \in \! \mathcal{E} \}$ is the set of index pairs of all vertices associated with the edge set $\mathcal{E}$. 

Based on (\ref{eqn:GMR_residual_laplacian}) and (\ref{eqn:GMRF_laplacian_quadratic_term}), it is clear that the distribution of the residual signal $\vr$ depends on edge weights ($\Adj$) and vertex weights ($\SLoop$) where  
\begin{itemize}
\item a model with larger (resp. smaller) edge weights (e.g., $\smash{(\Adj)_{ij}}$) increases the probability of having a smaller (resp. larger) squared difference between corresponding residual pixel values (e.g., $r_i$ and $r_j$),
\item a model with larger (resp. smaller) vertex weights (e.g., $\smash{(\SLoop)_{ii}}$) increases the probability of pixel values (e.g., $r_i$) with smaller (resp. larger) magnitude.
\end{itemize}

In practice, a characterization of the edge and vertex weights ($\Adj$ and $\SLoop$) can be made by estimating $\lap$ from data, which depend on inherent signal statistics and the type of prediction used for predictive coding. We empirically investigate the graph weights associated with residual signals in Section \ref{sec:graph_res_charac}.

\subsection{DCTs/DSTs as GBTs Derived from Line Graphs}
{Some types of DCTs and DSTs, including DCT-2 
and DST-7, are in fact special cases of GBTs derived from Laplacians of specific line graphs.} The relation between different types of DCTs and graph Laplacian matrices is originally discussed in \cite{Strang:1999:DCT} where DCT-2 is shown to be equal to {the GBT uniquely obtained from graph Laplacians of the following form:}
\begin{equation}
\lap_u = \begin{bmatrix}
c & -c &  &   & 0  \\
-c & 2c & -c &   \\
  & \ddots & \ddots & \ddots  \\
   & & -c & 2c & -c  \\
0 &  & & -c & c   
\end{bmatrix} \text{ for $c>0$},
\end{equation}
{which represents uniformly weighted line graphs with no self-loops (i.e., all edge weights are equal to a positive constant and vertex weights to zero).} Moreover, in \cite{egilmez:16:gbst,gene:15:GGL}, it has been shown that DST-7 is the GBT derived from a graph Laplacian $\smash{{\lap} = \lap_u + {\SLoop}}$ where $\smash{{\SLoop} =  \diag ([ c\; 0\; \cdots \; 0 ]\tran)}$ including a self-loop at vertex $v_1$ with weight $f_v(v_1)=c$. 
Based on the results in \cite{Strang:1999:DCT,Moura_Pueschel:08:algebraic_signal_proc_space}, various other types of DCTs and DSTs can be characterized using graphs. Table \ref{table:dct_dst_gbt} specifies the line graphs (with $n$ vertices $v_1,v_2,\dotsc,v_\nvert$ having self-loops at $v_1$ and $v_\nvert$) corresponding to different types of DCTs and DSTs, which are GBTs derived from graph Laplacians of the form $\smash{\widetilde{\lap} = \lap_u + \widetilde{\SLoop}}$ where $\smash{\widetilde{\SLoop} =  \diag ([f_v(v_1)\; 0 \; \cdots \; 0 \; f_v(v_\nvert) ]\tran)}$.

\begin{table}[!htb]
\caption{DCTs/DSTs corresponding to $\smash{\widetilde{\lap}}$ with different vertex weights.}
\label{table:dct_dst_gbt}
\centering
\begin{tabular}{|c|c|c|c|}
\hline
Vertex weights & $f_v(v_1)\!=\!0$ & $f_v(v_1)\!=\!{c}$& $f_v(v_1)\!=\!2{c}$ \\ \hline
$f_v(v_\nvert)\!=\!0$ & DCT-2 & DST-7 & DST-4   \\ \hline 
$f_v(v_\nvert)\!=\!{c}$ & DCT-8 & DST-1 &  DST-6   \\ \hline 
$f_v(v_\nvert)\!=\!2{c}$ & DCT-4 & DST-5 & DST-2  \\ \hline
\end{tabular}
\end{table}

{The current state-of-the-art VVC standard \cite{Bross:20:vvc10} uses separable transforms specified by pairs of DCT-2, DST-7 and DCT-8, that can be viewed as GBSTs. The GBST designs proposed in this paper are also based on line graphs whose weights, unlike those of the DCTs/DSTs, are learned from data.}

\section{Graph Learning for Graph-based Transform Design }
\label{sec:graph_learning_video}
\subsection{Generalized Graph Laplacian Estimation}
As justified in Proposition \ref{proposition:ggl_residual}, the residual signal $\vr \in \mathbb{R}^\nvert$ is modeled as an attractive GMRF, $\vr \sim \Normal(\vzeros,\Covar=\lap\inv)$, whose precision matrix is a GGL denoted by $\lap$. Assuming that we have $\nsamp$ residual signals, $\vr_1,\dotsc,\vr_\nsamp$, sampled from $\Normal(\vzeros,\Covar=\lap\inv)$,  the likelihood of a candidate $\lap$ is  
 \begin{equation}
\prod_{i=1}^{\nsamp}  \Prob({\vr_i} | \lap) \! = \!  {{(2\pi)}^{-\frac{kn}{2}} {\mathrm{det}({\lap})}^{\frac{\nsamp}{2}} } \prod_{i=1}^{\nsamp} \mathrm{ exp} \left( {- \frac{1}{2} {{\vr}_{i}}\tran {\mathbf{L}} \vr_{i}} \right).
 \label{eqn:residual_ML}
\end{equation} 
The maximization of the likelihood in (\ref{eqn:residual_ML}) can be equivalently formulated as minimizing the negative log-likelihood, that is
\begin{equation}
\begin{aligned}
\widehat{\lap}_{\textnormal{ML}} = & \argmin_{\lap}  \left\lbrace  \frac{1}{2} \sum_{i=1}^{\nsamp} \mathrm{Tr} \left(  {\vr_i}\tran \lap \vr_i  \right) - \frac{\nsamp}{2} \mathrm{ logdet} (\lap) \right\rbrace  \\
 = & \argmin_{\lap}  \left\lbrace  \mathrm{Tr} \left( \lap \CS \right) -  \mathrm{ logdet} (\lap)    \right\rbrace \\
\end{aligned}
\label{eqn:ML_est_residual}
\end{equation}
where $\CS=\frac{1}{\nsamp} \sum_{i=1}^{\nsamp} \vr_i{\vr_i}\tran$ is the sample covariance, and $\widehat{\lap}_{\textnormal{ML}}$ denotes the maximum likelihood (ML) estimate  of $\lap$. 
To find the best GGL from a set of residual signals $\{ \vr_1,\dotsc,\vr_\nsamp \}$ in a maximum likelihood sense, 
we solve the following GGL estimation problem {with connectivity constraints:} 
\begin{equation}
 \begin{aligned}
& \minimize_{\lap  \succeq 0}
& & 
 \mathrm{Tr} \left( \lap \CS \right) -\mathrm{logdet}  ( \lap )  \\
& \subjto
& &   (\lap)_{ij} \leq 0  \; \text{ if}\ (\Conn)_{ij} = 1  \\
 & & &  (\lap)_{ij} = 0    \; \text{ if}\ (\Conn)_{ij} = 0
\end{aligned}
\label{eqn:ggl_prob_residual}
\end{equation}
where $\CS$ denotes the sample covariance of residual signals, and $\Conn$ is the connectivity matrix representing the graph structure (i.e., the set of graph edges). 
In order to optimally solve (\ref{eqn:ggl_prob_residual}), we use the \emph{GGL estimation algorithm} proposed in our previous work on graph learning \cite{egilmez:2017:gl_from_data_jstsp,GLL_package:v1.0}.

\subsection{Graph-based Transform Design}
\label{subsec:opt_gbt_design_graph_learning}
To design separable and nonseparable GBTs {(GBSTs and GBNTs)}, we solve instances of  (\ref{eqn:ggl_prob_residual}) {denoted as $\mathsf{GGL}(\CS,\Conn)$ with different connectivity constraints represented by $\Conn$.} Then, the optimized GGL matrices are used to derive GBTs.

\noindent {\bf{{Graph learning oriented GBST (GL-GBST)}}.}
For the {GL-GBST} design, we solve two instances of (\ref{eqn:ggl_prob_residual}) to optimize two separate line graphs used to derive $\mathbf{U}_{\text{row}}$ and $\mathbf{U}_{\text{col}}$ in (\ref{eqn:gbst_video}). Since we wish to design a separable transform, each line graph can be optimized independently\footnote{Alternatively, joint optimization of the transforms associated with rows and columns has been  proposed in \cite{egilmez:2016:rct,Pavez:2017:joint_gbst}.}.  
Thus, our basic goal is finding the best line graph pair based on sample covariance matrices ${{\mathbf{S}}}_{\text{row}}$ and ${{\mathbf{S}}}_{\text{col}}$ created from rows and columns of residual block signals. 
For $N \times N$ residual blocks, the proposed {GL-GBST} construction has the following steps:
\begin{enumerate}
\item Create the connectivity matrix ${\Conn_{\text{line}}}$ representing a line graph structure with $\nvert = N$ vertices as in Fig.~\ref{fig:intra_inter_models_1D}.
\item Obtain two $N \!\times\! N$ sample covariances, $\CS_{\text{row}}$ and $\CS_{\text{col}}$, from rows and columns of size $N$, respectively, obtained 
from residual blocks in the dataset.
\item Solve instances of the problem in (\ref{eqn:ggl_prob_residual}), $\mathsf{GGL}(\CS_{\text{row}},{\Conn_{\text{line}}})$ and $\mathsf{GGL}(\CS_{\text{col}},{\Conn_{\text{line}}})$, 
by using {the GGL estimation algorithm \cite{egilmez:2017:gl_from_data_jstsp}}
to learn Laplacians $\lap_{\text{row}}$ and $\lap_{\text{col}}$ representing line graphs, respectively.
\item Perform eigendecomposition on $\lap_{\text{row}}$ and $\lap_{\text{col}}$ to obtain GBTs, $\gbt_{\text{row}}$ and $\gbt_{\text{col}}$, which define the GBST as in (\ref{eqn:gbst_video}).
\end{enumerate}
\noindent {\bf{{Graph learning oriented GBNT (GL-GBNT)}}.} 
Similarly, for $N\times N$ residual block signals, we propose the following steps to design a {GL-GBNT}:
\begin{enumerate}
\item Create the connectivity matrix $\Conn$ based on a desired graph structure {supporting $\nvert = N^2$ vertices}. {In this paper, we use connectivity constraints based on grid graphs\footnote{{According to our experiments across different prediction modes, using grid graph constraints led to a better coding efficiency as compared to using constraints that further allow edges between the nearest diagonal pixels.}}, as illustrated in Fig.~\ref{fig:intra_inter_models_2D}.} 
\item Obtain $N^2 \!\times\! N^2$ sample covariance $\CS$ using residual block signals in the dataset (after vectorizing the block signals).
\item Solve the problem  $\mathsf{GGL}(\CS,\Conn)$ by using {the GGL estimation algorithm \cite{egilmez:2017:gl_from_data_jstsp}}
to estimate a Laplacian $\lap$.
\item Perform eigendecomposition on $\lap$ to obtain the $N^2 \!\times\! N^2$ GBNT, $\gbt$ defined in (\ref{eqn:gbnt_video}).
\end{enumerate}

{\subsection{Computational complexity of GL-GBTs}
\label{subsec:complexity_glgbt}
In practice, GL-GBTs do not introduce an additional computational complexity over other transform types implemented using full-matrix multiplications\footnote{{In VVC \cite{Bross:20:vvc10}, the adopted variant of nonseparable transforms in \cite{xin:2016:nnst} as well as the DST-7 and DCT-8 are implemented using full-matrix multiplications.}}. 
{The recent work in \cite{egilmez:2020:paramteric_gbst_arxiv} also empirically demonstrates that GL-GBSTs do not increase encoder/decoder run-time over VVC \cite{Bross:20:vvc10}, if they are used to replace the separable transforms derived from DST-7 and DCT-8.}
However, it is important to note that nonseparable transforms (such as the ones in VVC) are inherently complex as they require $N^2$ multiplications per pixel for an $N \times N$ block signal, while separable transforms only need $2N$ multiplications per pixel. To alleviate the complexity of nonseparable transforms, the methods providing low-complexity approximations or decompositions in \cite{egilmez:2016:rct,said:2016:hygt,Said:19:taf} can be applied instead of using full-matrix multiplications.}

\subsection{Theoretical Justification for GL-GBTs} 
\label{sec:opt_ggl}
It has been shown that KLT is optimal for transform coding of jointly Gaussian sources in terms of mean-square error (MSE) criterion under high-bitrate assumptions \cite{GershoGray:1991:VQ_book,goyal:2001:TC_Magazine,Mallat:2008:WaveletTour}. Since GMRF models lead to jointly Gaussian distributions, the corresponding KLTs are optimal in theory. However, in practice, a KLT 
is obtained by eigendecomposition of the associated sample covariance, which has to be estimated from a training dataset where 
the number of data samples may not be sufficient to accurately recover the parameters. As a result, the sample covariance may lead a poor estimation of the actual model parameters\cite{johnstone2009consistency,ravikumar:2011:bounds_inverse_cov_estimation}. To improve estimation accuracy and alleviate overfitting, it is often useful to reduce the number of model parameters by introducing model constraints and regularization. From the statistical learning theory perspective \cite{vapnik:1999:slt,loxburg:2011:slt}, 
the advantage of our proposed GL-GBT over KLT is that KLT requires learning $\bigO(n^2)$ model parameters while GL-GBT only needs $\bigO(n)$, given the connectivity constraints in (\ref{eqn:ggl_prob_residual}). Therefore, our graph learning approach provides better \emph{generalization} in learning the signal model by taking into account variance-bias tradeoff. This advantage can also be justified based on the following error bounds characterized in \cite{vershynin:2012:how_close_covariance,ravikumar:2011:bounds_inverse_cov_estimation}. 
Assuming that $\nsamp$ residual blocks are used for calculating the sample covariance $\CS$, under general set of assumptions, the error bound for estimating $\Covar$ with $\CS$ derived in \cite{vershynin:2012:how_close_covariance} is 
\begin{equation}
|| \Covar - \CS ||_F= \bigO \left( \sqrt{\frac{\nvert^2 \mathrm{log}(\nvert)}{\nsamp}} \right),
\end{equation}
while estimating the precision matrix $\Precision$ by using the proposed graph learning approach leads to the following bound shown in \cite{ravikumar:2011:bounds_inverse_cov_estimation},  
\begin{equation}
|| \Precision - \lap ||_F= \bigO \left( \sqrt{\frac{\nvert \mathrm{log}(\nvert)}{\nsamp}} \right),
\end{equation}
where $\lap$ denotes the estimated GGL. Thus, in terms of the worst-case errors (based on Frobenius norm), 
the proposed method provides a better model estimation as compared to the estimation based on the sample covariance. 
Section \ref{sec:results} empirically justifies the advantage of GL-GBT against KLT.

\section{Edge-Adaptive Graph-based Transforms}
\label{sec:eagbt}
The optimality of GL-GBTs relies on the assumption that the residual signal characteristics are the same across different blocks. 
However, in practice, video blocks often exhibit 
complex image edge structures that can degrade the coding performance  
when the transforms are designed from average statistics without any classification based on image edges. 
In order to achieve better compression for video signals with image edges,
we propose edge-adaptive GBTs (EA-GBTs), 
designed on a per-block basis, by constructing a graph 
whose weights are determined based on the salient image edges in each residual block. 

\subsection{EA-GBT Construction}
To design an EA-GBT for a residual block, we first detect image edges based on a threshold $T_{\text{edge}}$ applied on gradient values, 
obtained using the \emph{Prewitt} operator on the block. Then, edge weights of a predefined graph are modified according to the locations of detected image edges, and the resulting graph is used to derive the associated GBT. 
As depicted in Fig.~\ref{fig:example_eagbt_const}, 
to construct a graph, we start with a uniformly weighted grid graph for which all edge weights are equal to a fixed constant $w_c$ (Fig.~\ref{fig:example_eagbt_edgy_init_graph}). 
Then, the detected image edges on a given residual block (Fig.~\ref{fig:example_eagbt_edgy_image}) are used to determine the co-located edges in the graph, and the corresponding weights are reduced as $w_e = w_c/s_{\text{edge}}$ (Fig.~\ref{fig:example_eagbt_edgy_graph}), where $s_{\text{edge}} \geq 1$ is a parameter modeling the sharpness of image edges (i.e., the level of differences between pixel values in presence of an image edge). Thus, a larger $s_{\text{edge}}$ leads to smaller weights on edges connecting pixels (vertices) with an image edge in between.

According to our simulations on residual data, coding gains are observed for $s_{\text{edge}} > 10$, which empirically corresponds to residuals with an intensity difference of at least 12 (between pixels adjacent to an image edge)\footnote{{Residual signals are in 9-bit signed integer precision, since the input video sequences used in our experiments are all in 8-bit unsigned integer precision.}}. In our experiments, a conservative threshold of $T_{\text{edge}} = 10$ is used for image edge detection, and the parameter $s_{\text{edge}}$ is set to 10. 
Although the proposed EA-GBT design can be extended with multiple $s_{\text{edge}}$ parameters, our experiments showed that such extensions do not provide a good rate-distortion trade-off due to the additional signaling overhead. For efficient signaling of detected edges, we employ arithmetic edge coding (AEC) \cite{gene:2014:aec}, a state-of-the-art binary edge-map codec.

From the compression perspective, the EA-GBT construction can also be viewed as a classification procedure, so that each residual block (e.g., in Fig.~\ref{fig:example_eagbt_edgy_image}) is assigned to a class of signals associated with an attractive GMRF, whose corresponding graph (i.e., GGL) is determined by $s_{\text{edge}}$ and image edge detection based on $T_{\text{edge}}$ (e.g., in Fig.~\ref{fig:example_eagbt_edgy_graph}). 
By using attractive GMRFs, 
the following subsection theoretically validates our experimental observation of achieving coding gains for $s_{\text{edge}} > 10$.

{\subsection{Computational complexity of EA-GBTs}
As GL-GBTs, EA-GBTs need to be implemented based on full-matrix multiplications, whose implications on complexity are discussed in Section \ref{subsec:complexity_glgbt}. 
Moreover, block-adaptive coding schemes such as \cite{Shen:10:eat,hu:14:PWS,Fracastoro:2018:apprx_graph} and EA-GBT introduce an additional computational overhead since transforms need to be constructed on a per-block basis. In practice, this additional complexity can be eliminated for EA-GBT by using a predetermined set of graphs (such as graph templates introduced in \cite{egilmez:15:gbt_inter}) capturing a fixed number of image edge patterns in an RDOT scheme, where the associated GBTs are stored and not constructed per-block. In this paper, in order to evaluate the best case performance of EA-GBT we do not impose any restrictions on the image-edge patterns and evaluate the coding efficiency of EA-GBT as a block-adaptive scheme.}

\begin{figure}[!t]
\centering
        \subfloat[Initial graph]{ %
       \includegraphics[trim=60 20 60 20,clip,width=.16\textwidth]{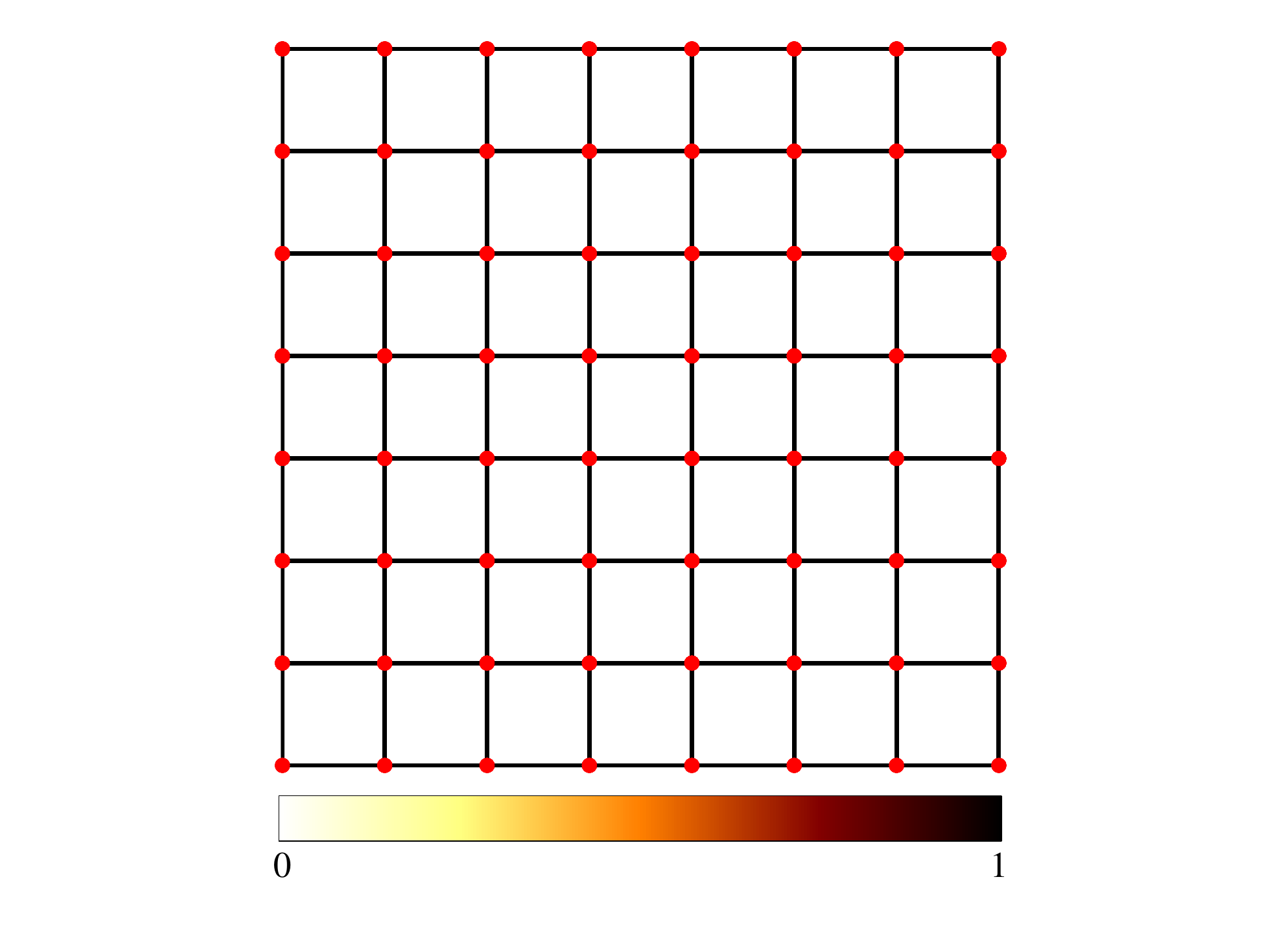}\label{fig:example_eagbt_edgy_init_graph}}
        \subfloat[Residual block signal]{%
       \includegraphics[trim=60 0 60 20,clip,width=.16\textwidth]{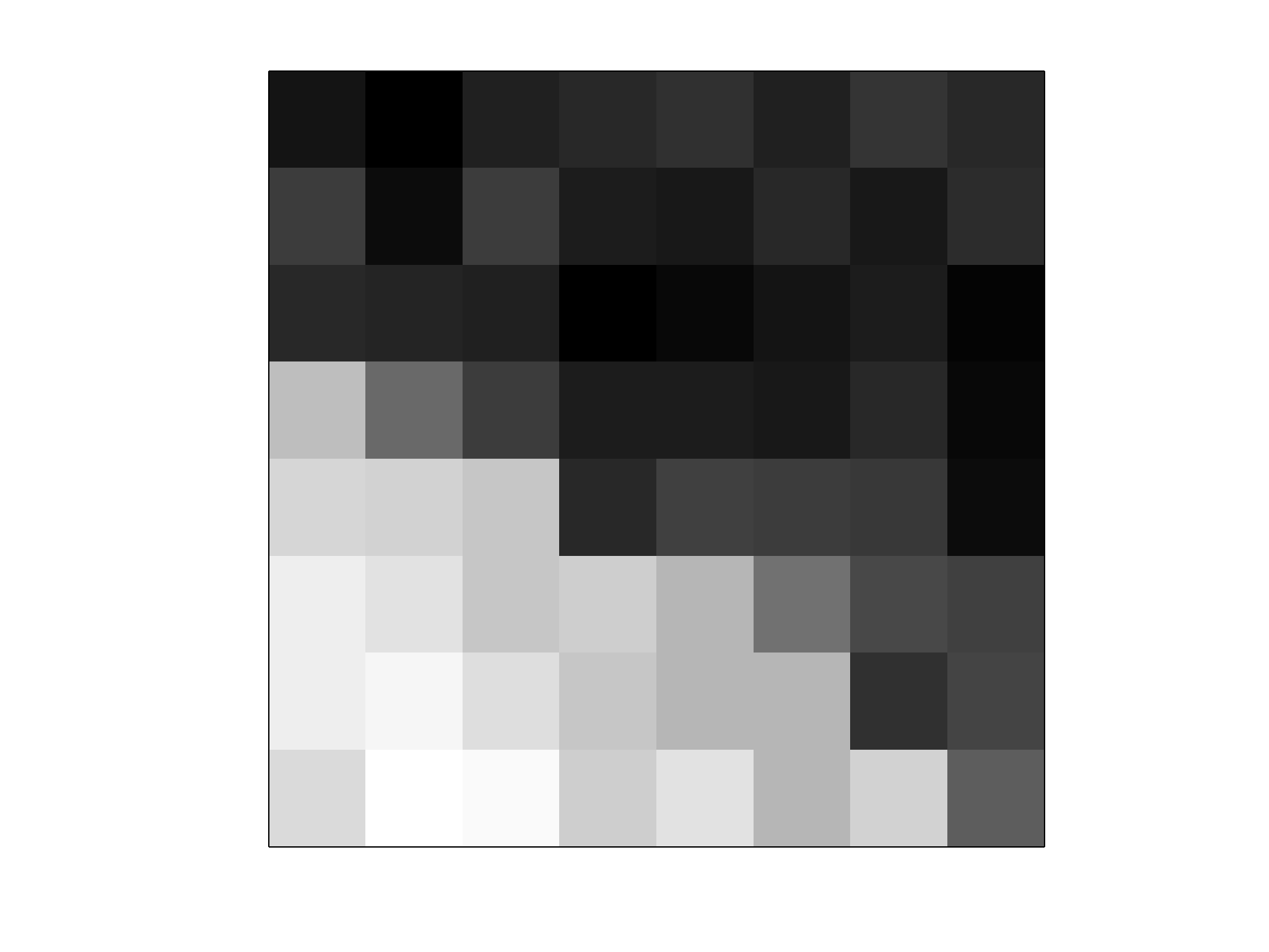} 
       \label{fig:example_eagbt_edgy_image}}
        \subfloat[Constructed graph]{%
       \includegraphics[trim=60 20 60 20,clip,width=.16\textwidth]{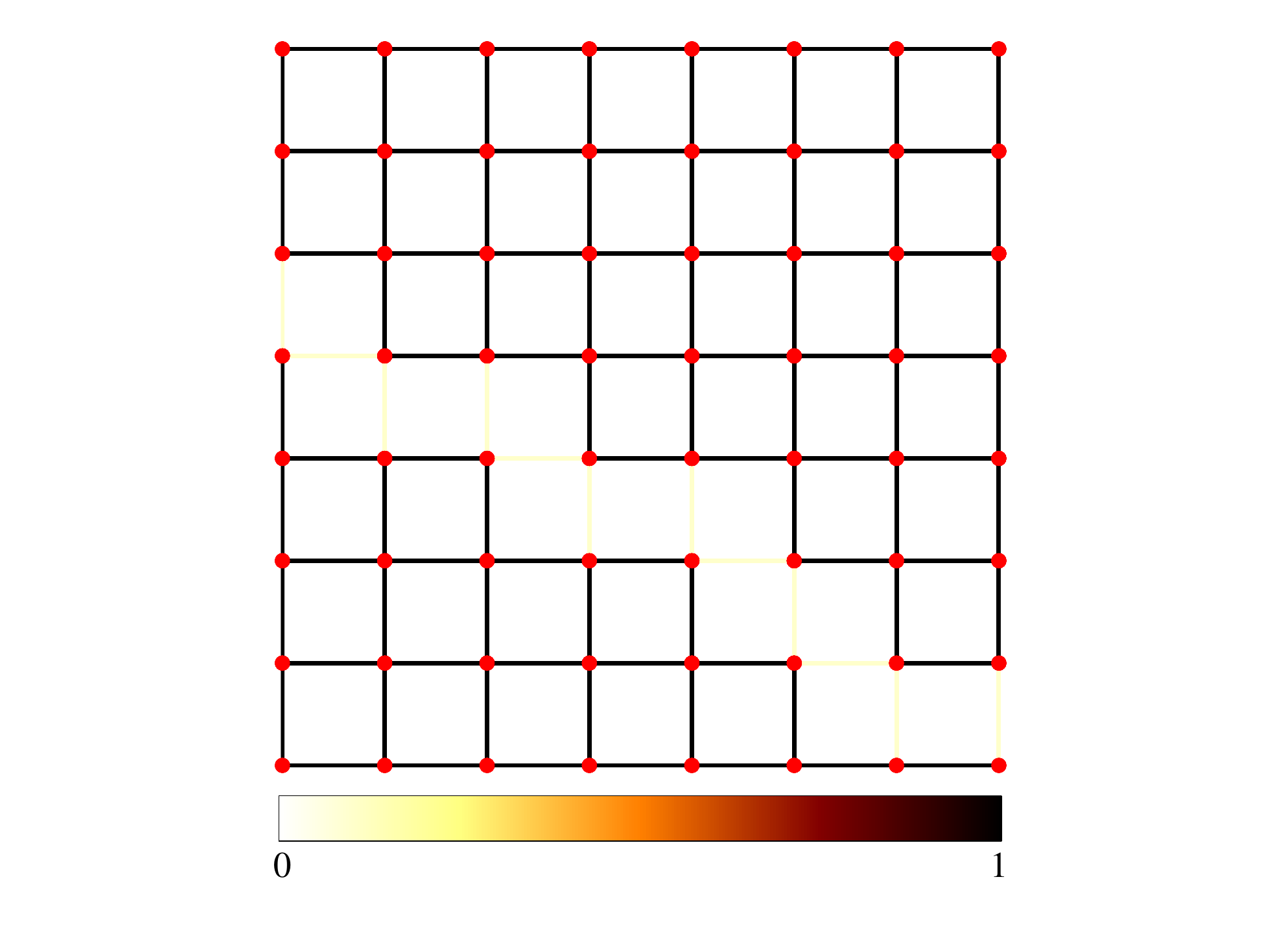}
       \label{fig:example_eagbt_edgy_graph}}
       \caption{An illustration of graph construction for a given $ 8\times 8$ residual block signal where $w_c=1$ and $w_e = w_c/s_{\text{edge}}=0.1$ where $s_{\text{edge}}=10$ .} 
       \label{fig:example_eagbt_const}
\end{figure}

\begin{figure}[!t]
\centering
       \includegraphics[width=.4\textwidth]{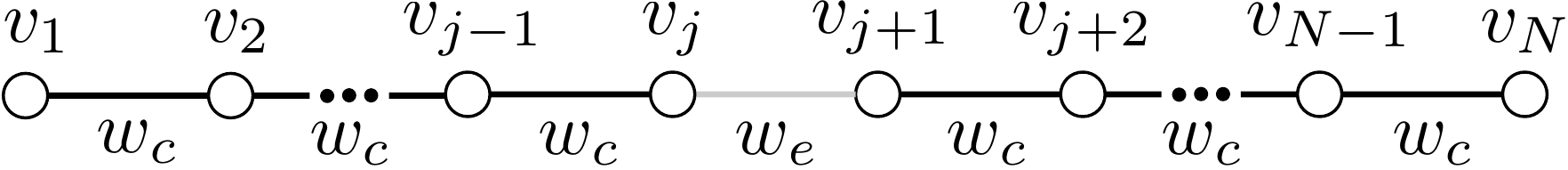}
       \caption{A 1-D graph-based model with an image edge at location $l=j$. All black colored edges have weights equal to $w_c$, and the gray edge between vertices $v_j$ and $v_{j+1}$ is weighted as $w_e = w_c/s_{\text{edge}}$.} 
       \label{fig:edge_justification}
\end{figure}

\subsection{Theoretical Justification for EA-GBTs}
We present a theoretical justification for advantage of EA-GBTs over KLTs. For the sake of simplicity, our analysis is based on 1-D 
models with a single image edge, whose the location 
$l$ is uniformly distributed as 
\begin{equation}
\Probm(l=j) = \begin{cases}  
\frac{1}{N-1} & \text{ for } j=1,\dotsc,N-1 \\
\hspace{0.25cm} 0   & \text{ otherwise } 
\end{cases} 
\label{eqn:uniform_dist_edge}
\end{equation}
where $N$ is the number of pixels (i.e., vertices) on the line graph depicted in Fig.~\ref{fig:edge_justification}. This construction leads to a Gaussian mixture distribution based on $M=N-1$ attractive GMRFs,
\begin{equation}
\label{eqn:mixed_distr}
\Prob(\vx) = \sum_{j=1}^{M} \Probm(l=j) \, \Normal(\mathbf{0},\Covar_j)
\end{equation}
with $\Covar_j$ denoting the covariance of the $j$-th attractive GMRF,
whose corresponding graph has an image edge between pixels $v_j$ and $v_{j+1}$ as illustrated in Fig.~\ref{fig:edge_justification}. 
Since $\vx$ follows a Gaussian mixture distribution, the KLT obtained from the covariance of $\vx$ (which implicitly performs a second-order approximation of the distribution) is suboptimal in MSE sense \cite{effros:2004:suboptimal_klt}. 
Especially, with many possible image edge locations and different orientations, the underlying distribution may contain a large number of mixtures (i.e., a large $M$), which makes learning a model from average statistics inefficient.
On the other hand, the proposed EA-GBT removes the uncertainty due to the random variable $l$ by detecting the location of the image edge in pixel (vertex) domain, and then constructing a GBT based on the detected image edge. Yet, EA-GBT requires allocating additional bits to represent the image edge (side) information, while KLT only allocates bits for coding transform coefficients. 

\begin{figure}[!t]
\centering
 {\resizebox{0.5\textwidth}{!}{\input{cg_sharp2.tex}}}
\caption{Coding gain ($\mathsf{cg}$) versus $s_{\text{edge}}$  
for block sizes with $N=4,8,16,32,64$. 
EA-GBT provides better coding gain (i.e., $\mathsf{cg}$ is negative) when $s_{\text{edge}}$ is larger than $10$ across different block sizes.}
\label{fig:coding_gain_eagbt_high_bitrate}
\centering
 {\resizebox{0.5\textwidth}{!}{\input{cg_waterfill.tex}}}
\caption{Coding gain ($\mathsf{cg}$) versus bits per pixel ($R/N$) for different edge sharpness parameters $s_{\text{edge}}=10,20,40,100,200$. EA-GBT provides better coding gain (i.e., $\mathsf{cg}$ is negative) if $s_{\text{edge}}$ is larger than $10$ for different block sizes.}
\label{fig:coding_gain_eagbt_water_filling}
\end{figure}

To demonstrate the rate-distortion tradeoff between KLT and EA-GBT based coding schemes, we use classical rate-distortion theory results with high-bitrate assumptions\cite{GershoGray:1991:VQ_book,goyal:2001:TC_Magazine,Mallat:2008:WaveletTour}, in which the distortion ($D$) can be written as a function of bitrate ($R$), 
\begin{equation}
D(\bar{R}) = \frac{N}{12} 2^{2\bar{H}_d} 2^{-2\bar{R}}
\label{eqn:D(R)_high_bitrate}
\end{equation}
with 
\begin{equation}
\bar{R} = \frac{R}{N} \quad \text{ and } \quad \bar{H}_d = \frac{1}{N} \sum_{i=1}^{N} {H}_d((\mathbf{c})_i)
\end{equation}
where $R$ denotes the total bitrate allocated to code transform coefficients in $\mathbf{c} \!=\! \gbt\tran \vx$, and $\smash{{H}_d((\mathbf{c})_i)}$ is the differential entropy of transform coefficient $\smash{(\mathbf{c})_i}$. For EA-GBT, $R$ is allocated to code both transform coefficients ($R^{\text{coeff}}_{\text{EA-GBT}}$) and side information ($R^{\text{edge}} $), so we have
\begin{equation}
R = R^{\text{coeff}}_{\text{EA-GBT}} + R^{\text{edge}} =  R^{\text{coeff}}_{\text{EA-GBT}} + \mathrm{log}_2 (M)
\end{equation}
while for KLT, the bitrate is allocated only to code transform coefficients ($R^{\text{coeff}}_{\text{KLT}}$), so that 
$R = R^{\text{coeff}}_{\text{KLT}}$. 
Fig.~\ref{fig:coding_gain_eagbt_high_bitrate} shows the coding gain of EA-GBT over KLT for different sharpness parameters (i.e., $s_{\text{edge}}$) in terms of  the following metric, called coding gain,
\begin{equation}
\mathsf{cg}({D_{\text{EA-GBT}}},{D_{\text{KLT}}}) = 10 \, \mathrm{log}_{10} \left( \frac{D_{\text{EA-GBT}}}{D_{\text{KLT}}} \right)
\label{eqn:coding_gain_eagbt}
\end{equation}
where ${D_{\text{EA-GBT}}}$ and ${D_{\text{KLT}}}$ denote distortion levels measured at high-bitrate regime for EA-GBT and KLT, respectively. EA-GBT provides better compression for negative $\mathsf{cg}$ values in Fig.~\ref{fig:coding_gain_eagbt_high_bitrate} which appear when the sharpness of edges  $s_{\text{edge}}$ is large (e.g., $s_{\text{edge}} > 10$). 

Note that the distortion function in (\ref{eqn:D(R)_high_bitrate}) is derived based on high-bitrate assumptions. 
To characterize rate-distortion tradeoff for different bitrates, we employ the reverse water-filling technique \cite{Cover:1991:EIT,GershoGray:1991:VQ_book} by varying the parameter $\theta$ in (\ref{eqn:distortion_Di}) to obtain rate and distortion measures as follows
\begin{equation}
R(D) = \sum_{i=1}^{N} \frac{1}{2} \mathrm{log}_2\left( \frac{\lambda_i}{D_i}  \right)
\end{equation} 
where $\lambda_i$ is the $i$-th eigenvalue of the signal covariance, and 
\begin{equation}
D_i = \begin{cases} 
       \lambda_i & \text{ if } \lambda_i \geq \theta  \\
       \theta & \text{ if }  \theta < \lambda_i
\end{cases}
\label{eqn:distortion_Di}
\end{equation}
so that $D = \sum_{i=1}^{N} D_i$. {The rate ($R(D)$) and distortion ($D$) for KLT are estimated based on the covariance associated with (\ref{eqn:mixed_distr}). 
For EA-GBT, they are averaged over the metrics obtained for $M$ covariances 
(i.e., $\Covar_j$ for $j=1,\dotsc,M$), as their frequency of occurrence follows the uniform distribution in (\ref{eqn:uniform_dist_edge}).}

Figure~\ref{fig:coding_gain_eagbt_water_filling} illustrates the coding gain formulated in (\ref{eqn:coding_gain_eagbt}) achieved at different bitrates, where each curve correspond to a different $s_{\text{edge}}$ parameter. 
Similar to Fig.~\ref{fig:coding_gain_eagbt_high_bitrate}, EA-GBT 
leads to a better compression if the sharpness of edges, $s_{\text{edge}}$, is large  (e.g., $s_{\text{edge}}>10$ for $R/N >0.6$)\footnote{In practice, $R/N > 0.6$ is typically achieved at quantization parameters \cite{Sullivan:12:hevc} smaller than 32 in video video coding.}. At low-bitrates (e.g., $R/N<0.6$), EA-GBT can perform worse than KLT for $s_{\text{edge}}\!=\!20,40$, yet EA-GBT outperforms as bitrate increases.

\section{Residual Block Signal Characteristics and Graph-based Models}
\label{sec:graph_res_charac}


In this section, we discuss statistical characteristics of intra and inter predicted residual blocks, and empirically justify our theoretical analysis and observations in Section \ref{sec:models}.
Our empirical results are based on residual signals obtained by encoding 5 different video sequences  (\emph{City}, \emph{Crew}, \emph{Harbour}, \emph{Soccer} and \emph{Parkrun}{\cite{Xiph:video_dataset}}) using the HEVC reference software (HM-14) \cite{Sullivan:12:hevc} at 4 different QP parameters ($QP\!=\!\{22,27,32,37\}$). 
Although the HEVC standard does not implement optimal MMSE prediction (which is the main assumption in Section \ref{sec:models}), it includes 35 intra and 8 inter prediction modes, which provide reasonably good prediction for different classes of block signals.

Figs.~\ref{fig:graph_weights_intra_planar}--\ref{fig:graph_weights_inter_rectangular} depict statistical characteristics of $8 \times 8$ residual signals for a few intra and inter prediction modes\footnote{For $4\times 4$ and $16\times 16$ residual blocks, the structure of sample variances and graphs are quite similar to the ones in Figs.~\ref{fig:sample_variance_Planar}--\ref{fig:sample_variance_Nx2N}.}. In these figures, sample variances of residuals and corresponding graph-based models are illustrated. Both grid and line graphs with normalized edge and vertex weights\footnote{{Section \ref{sec:notation_prelim} formally defines edge and vertex weights, denoted by $(\Adj)_{ij} = (\Adj)_{ji} = f_w((v_i,v_j))$ and $(\SLoop)_{ii}=f_v(v_i)$, used for deriving GGLs.}} are estimated from  residual data by solving the GGL estimation problem in (\ref{eqn:ggl_prob_residual}) used for {GL-GBNT} and {GL-GBST} construction.

Naturally, residual blocks have different statistical characteristics depending on the type of prediction and the prediction mode.
The sample variances shown in Figs.~\ref{fig:sample_variance_Planar}--\ref{fig:sample_variance_Nx2N}
for different prediction modes lead us to the following observations:
\begin{itemize}
\item As expected, inter predicted blocks have smaller sample variance (energy) across pixels compared to intra predicted blocks, because inter prediction provides better prediction with larger number of reference pixels as shown in Fig.~\ref{fig:intra_inter_models_2D}.
\item In intra prediction, sample variances are generally larger at the bottom-right part of residual blocks, since reference pixels are located at the top and left of a block where the prediction is relatively better. This holds specifically for planar, DC and diagonal modes using pixels on both top and left as references for prediction. 
\item For some angular modes including intra horizontal/vertical mode, only left/top pixels are used as references. In such cases 
the sample variance gets larger as distance from reference pixels increases. Fig.~\ref{fig:sample_variance_Hori} illustrates sample variances corresponding to the horizontal mode.
\item In inter prediction, sample variances are larger around the boundaries and corners of the residual blocks 
mainly because of occlusions leading to partial mismatches between reference and predicted blocks.
\item In inter prediction, PU partitions lead to larger residual energy around the partition boundaries as shown in Fig.~\ref{fig:sample_variance_Nx2N} corresponding to horizontal PU partitioning ($N\!\times\! 2N$) . 
\end{itemize}
Moreover, inspection of the estimated graphs in  Figs.~\ref{fig:graph_weights_intra_planar}--\ref{fig:graph_weights_inter_rectangular} leads to following observations, which validate our theoretical analysis and justify the interpretation of model parameters in terms of graph weights discussed in Section \ref{sec:models}:
\begin{itemize}
\item  Irrespective of the prediction mode/type, vertex (self-loop) weights tend to be larger for the pixels 
that are connected to reference pixels. Specifically, in intra prediction, graphs have larger vertex weights 
for vertices (pixels) located at the top and/or left boundaries of the block (Figs.~\ref{fig:graph_weights_intra_planar}--\ref{fig:graph_weights_intra_diagonal}), while the vertex weights are approximately uniform across vertices in inter prediction (Figs.~\ref{fig:graph_weights_inter_square} and \ref{fig:graph_weights_inter_rectangular}).

\item In intra prediction, the grid and line graphs associated with planar and DC modes are similar in structure (Figs.~\ref{fig:graph_weights_intra_planar} and \ref{fig:graph_weights_intra_dc}), where their edge weights decrease as the distance of edges to the reference pixels increase. Also, vertex weights are larger for the pixels located at top and left boundaries, 
since planar and DC modes use reference pixels from the both sides (top and left). 
These observations indicate that the prediction performance gradually decreases for pixels increasingly farther away 
 from the reference pixels.

\item For intra prediction with horizontal mode (Fig.~\ref{fig:graph_weights_intra_horizontal}), 
the grid graph has larger vertex weights at the left boundary of the block. 
This is because the prediction only uses reference pixels on the left side of the block.
Therefore, the line graph associated to rows has a large self-loop at the first pixel, 
while the other line graph has no dominant vertex weights.  
However, grid and line graphs for the diagonal mode (Fig.~\ref{fig:graph_weights_intra_diagonal}), 
are more similar to the ones for planar and DC modes, 
since the diagonal mode also uses the references from both top and left sides. 

\item For inter prediction with PU mode $2N\! \times \! 2N$ (do not perform any partitioning), 
the graph weights (both vertex and edge weights) are approximately uniform 
across the different edges and vertices (Fig.~\ref{fig:graph_weights_inter_square}). 
This shows that the prediction performance is similar at different locations (pixels).  
In contrast,  the graphs for the PU mode $N\! \times \! 2N$ (performs horizontal partitioning) leads to smaller edge weights around the PU partitioning (Fig.~\ref{fig:graph_weights_inter_rectangular}). In the grid graph, we observe smaller weights between the partitioned vertices. Among line graphs, only the line graph designed for columns has a small weight in the middle, as expected. 
\end{itemize}

\begin{figure*}[htbp!]
\centering
\subfloat[Variances per pixel\label{fig:sample_variance_Planar}]
{\includegraphics[trim=60 80 35 67,clip,height=.18\textwidth,valign=c]{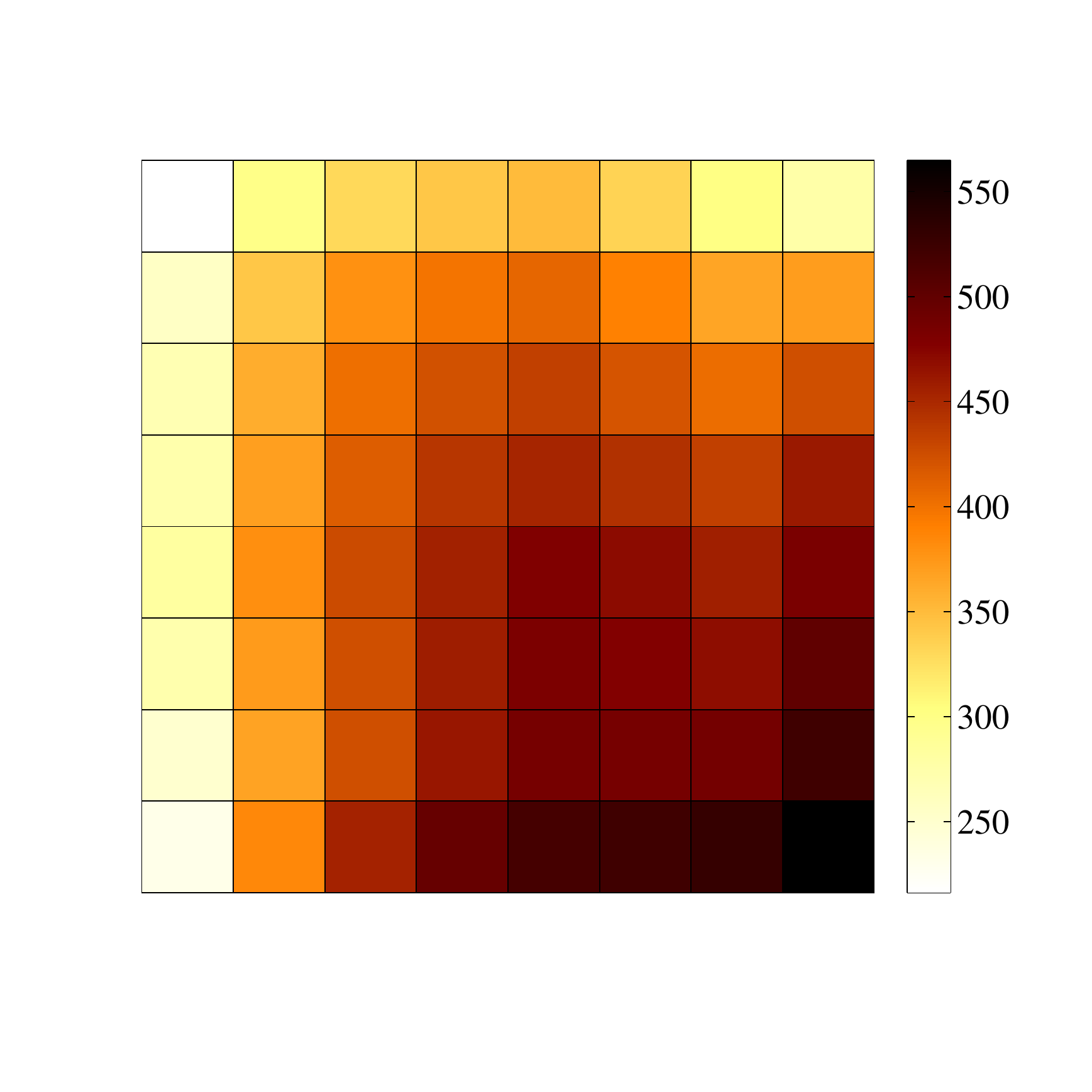}\vphantom{\includegraphics[trim=55 33 50 20,clip,width=.18\textwidth,valign=c]{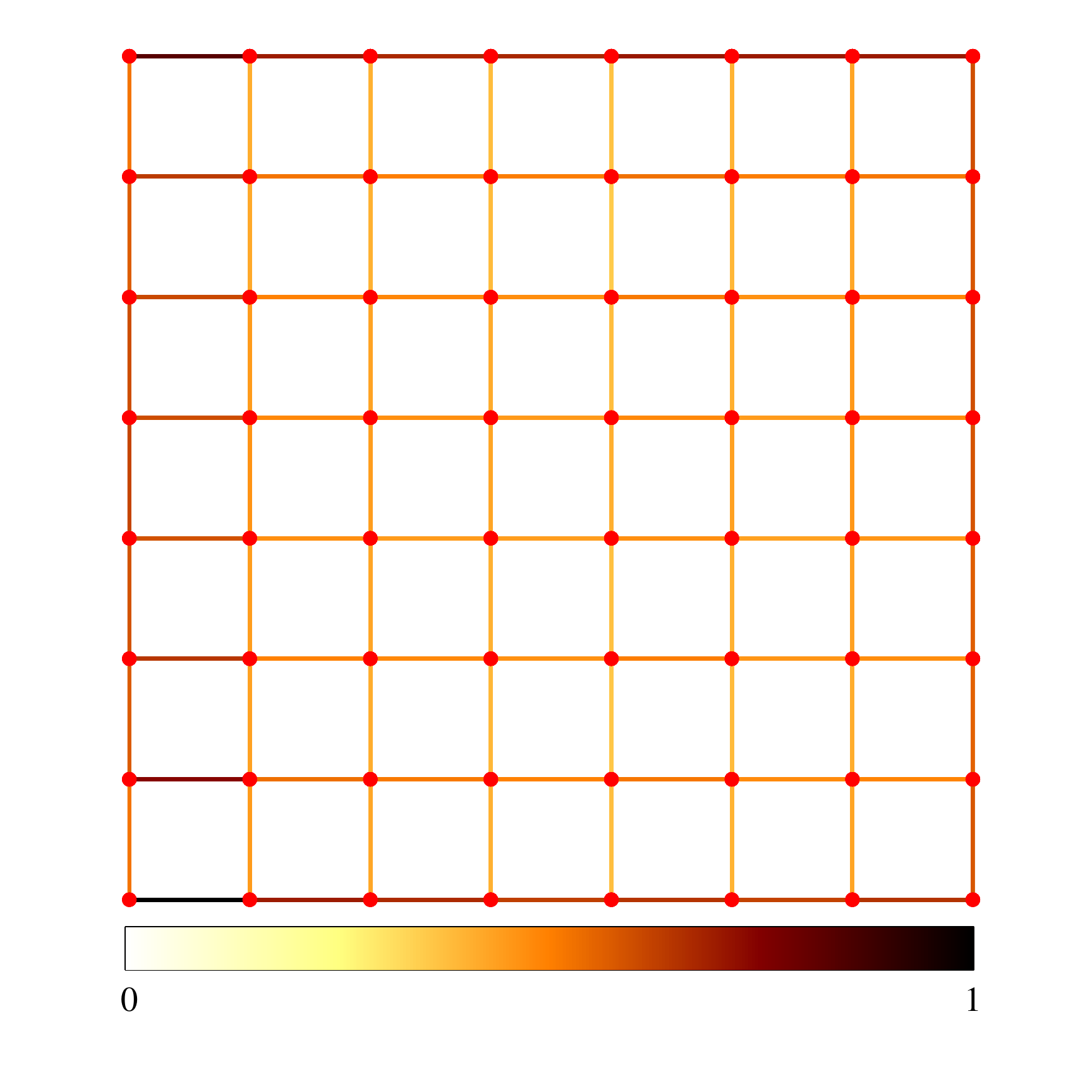}}}\;
    \subfloat[Grid graph weights]{\includegraphics[trim=55 33 50 20,clip,width=.18\textwidth,valign=c]{NonSep8Intra1_Graph-eps-converted-to.pdf}\;
    \includegraphics[trim=75 56 60 35,clip,width=.18\textwidth,valign=c]{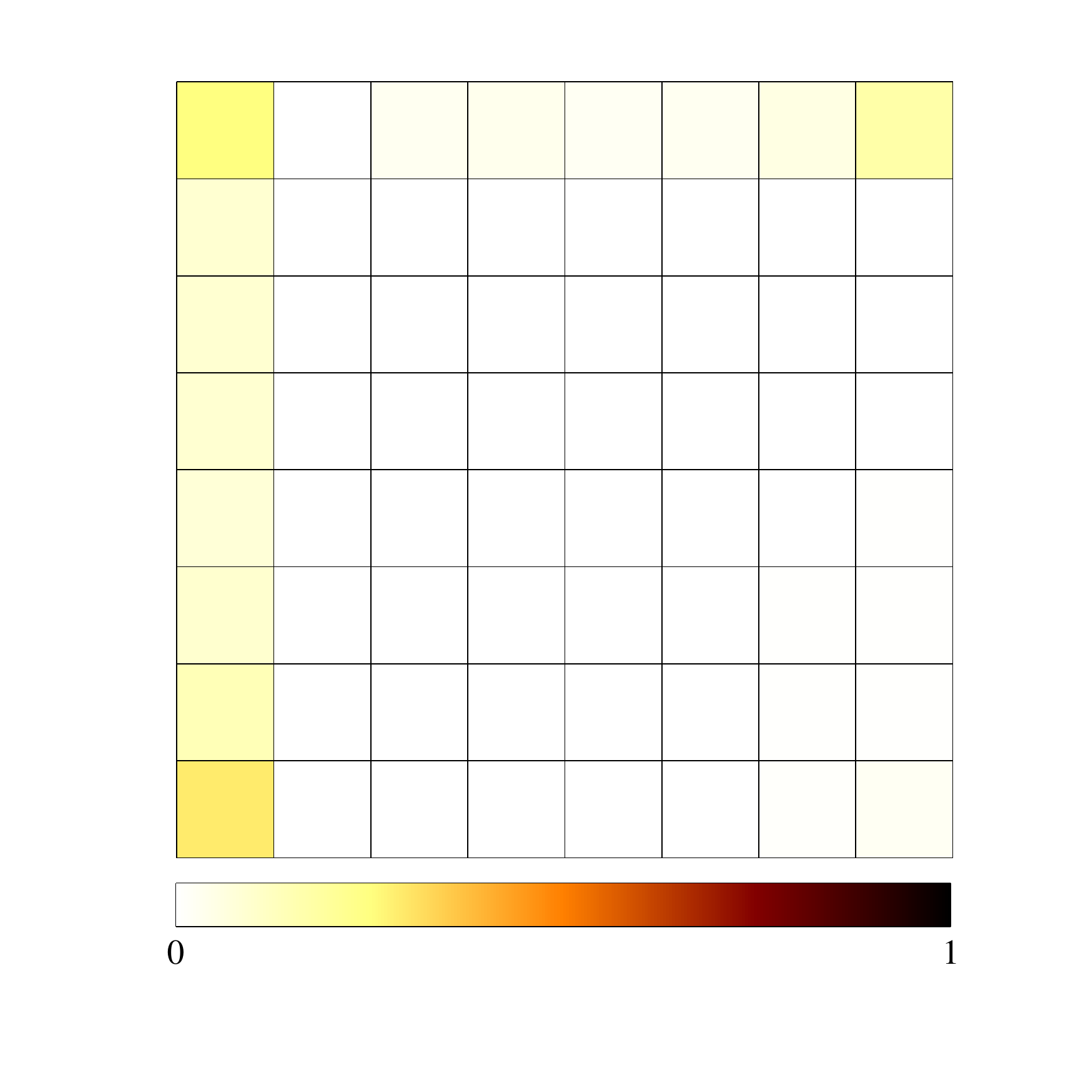}}\;
  \subfloat[Line graph weights for (top) rows  and (bottom) columns]{
  \begin{tabular}{@{}c@{}}
  \includegraphics[trim=55 200 50 185,clip,width=.18\textwidth,valign=c]{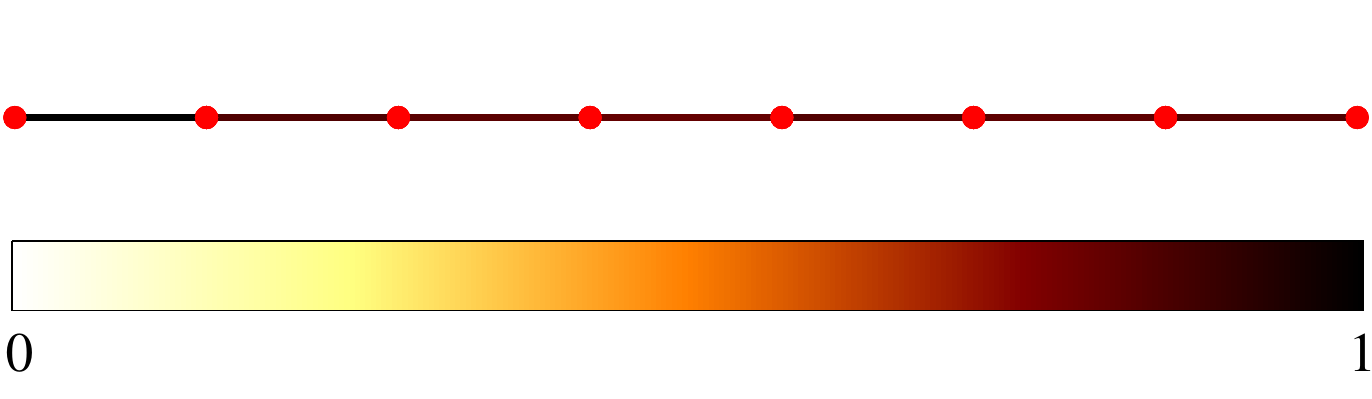} 
   \\  
    \includegraphics[trim=55 200 50 185,clip,width=.18\textwidth,valign=c]{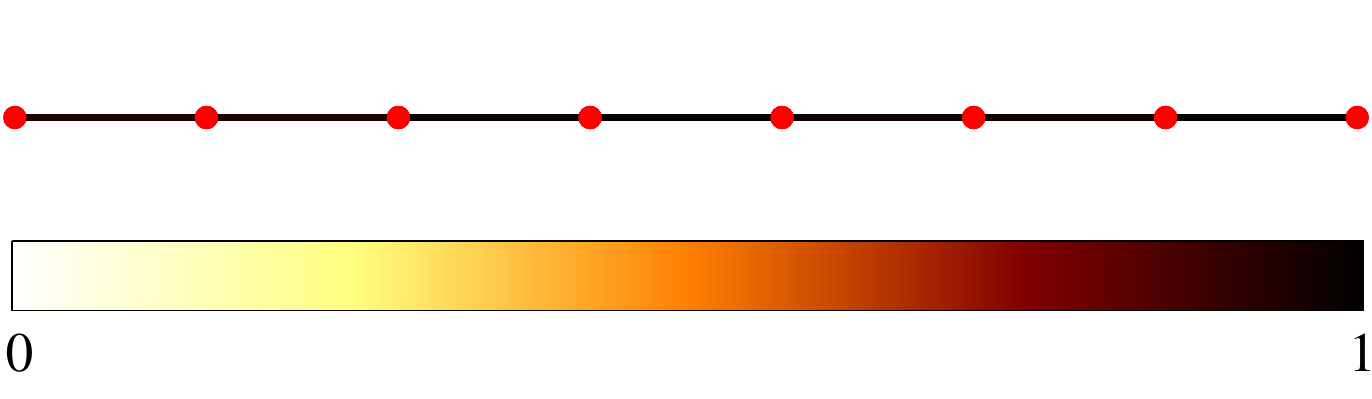} 
    \end{tabular}
  \begin{tabular}{@{}c@{}}
    \includegraphics[trim=55 195 40 185,clip,width=.18\textwidth,valign=c]{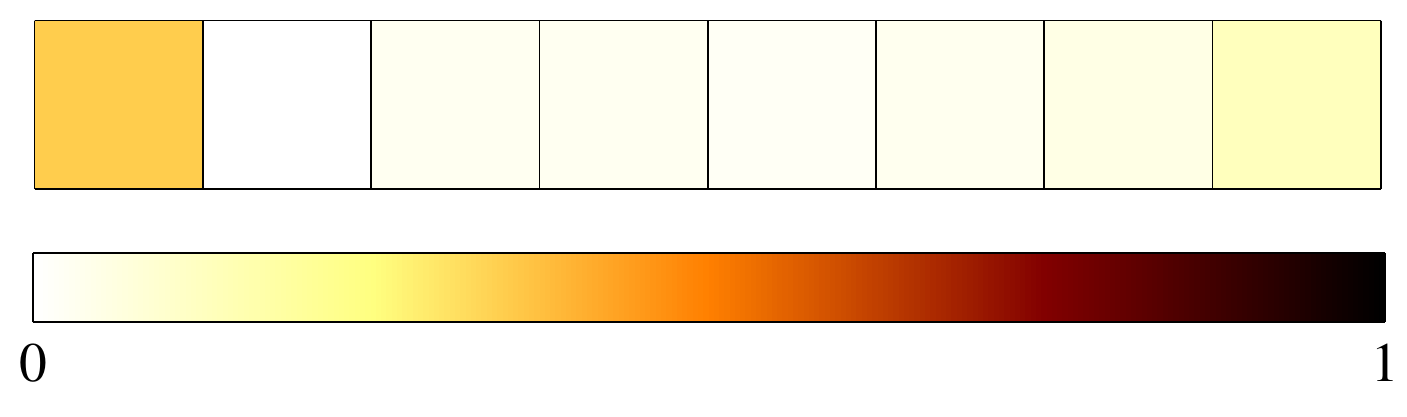}  
	\\
    \includegraphics[trim=55 195 40 185,clip,width=.18\textwidth,valign=c]{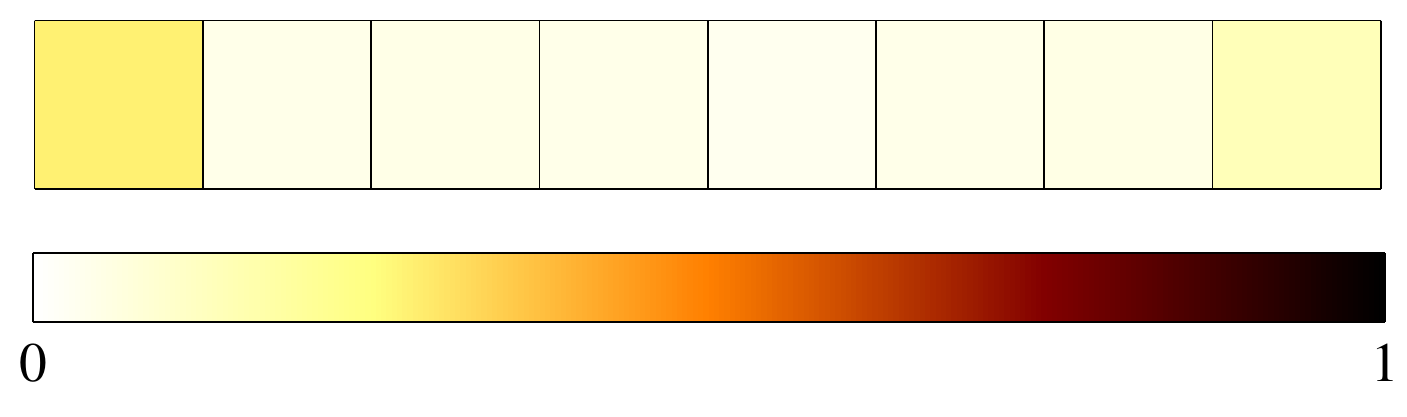}  
  \end{tabular}   
  \vphantom{\includegraphics[trim=55 33 50 20,clip,width=.18\textwidth,valign=c]{NonSep8Intra1_Graph-eps-converted-to.pdf}} 
  }
\caption{For the \emph{planar mode} in \emph{intra prediction} (a) shows the estimated sample variances of $8\times8$ residual signals.  In (b) and (c), edge and vertex weights are shown for grid and line graphs learned from residual data, respectively. Darker colors represent larger values.}
\label{fig:graph_weights_intra_planar}

\subfloat[Variances per pixel\label{fig:sample_variance_DC}]
{\includegraphics[trim=60 80 35 67,clip,height=.18\textwidth,valign=c]{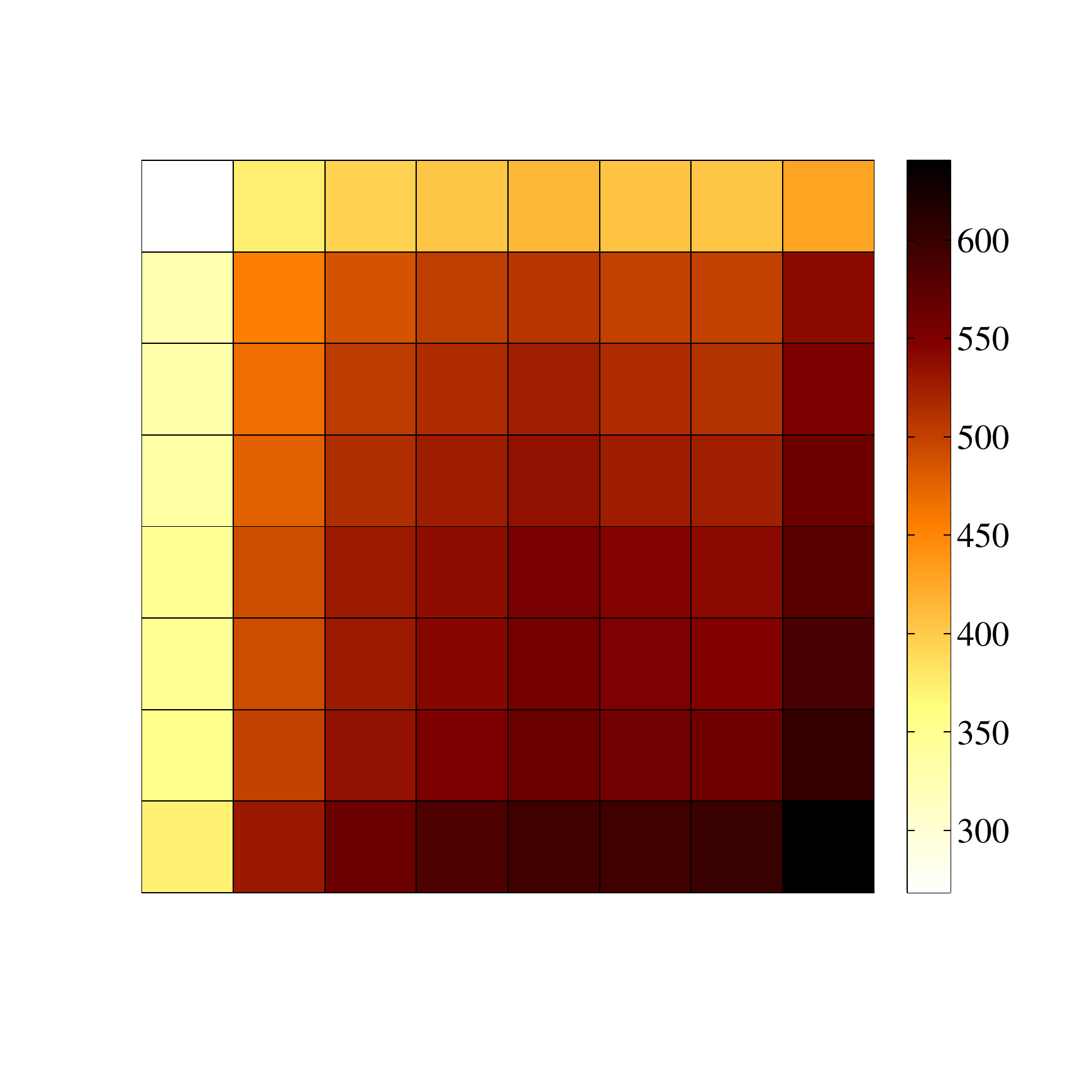}\vphantom{\includegraphics[trim=55 33 50 20,clip,width=.18\textwidth,valign=c]{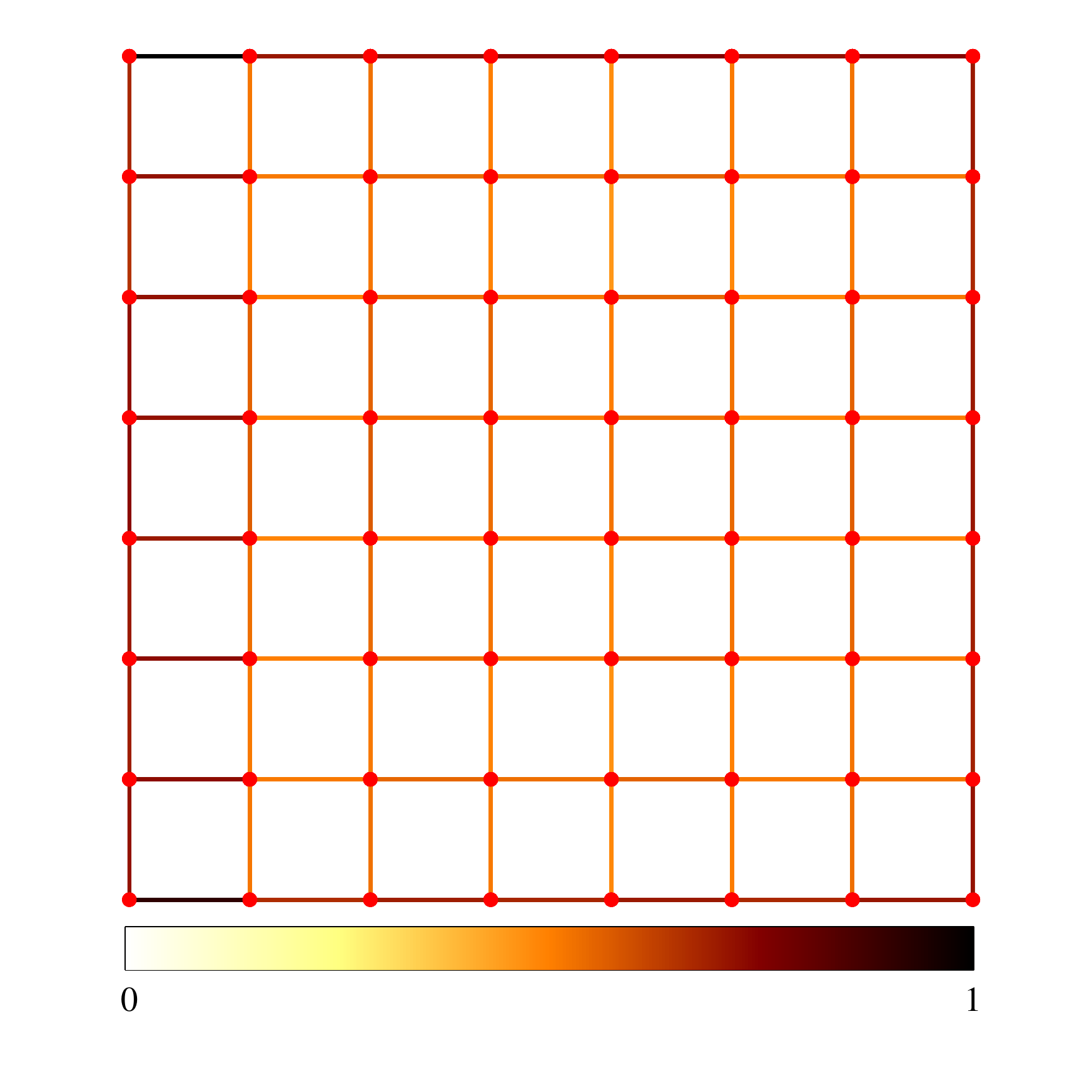}}}\;
    \subfloat[Grid graph weights]{\includegraphics[trim=55 33 50 20,clip,width=.18\textwidth,valign=c]{NonSep8Intra2_Graph-eps-converted-to.pdf}\;
    \includegraphics[trim=75 56 60 35,clip,width=.18\textwidth,valign=c]{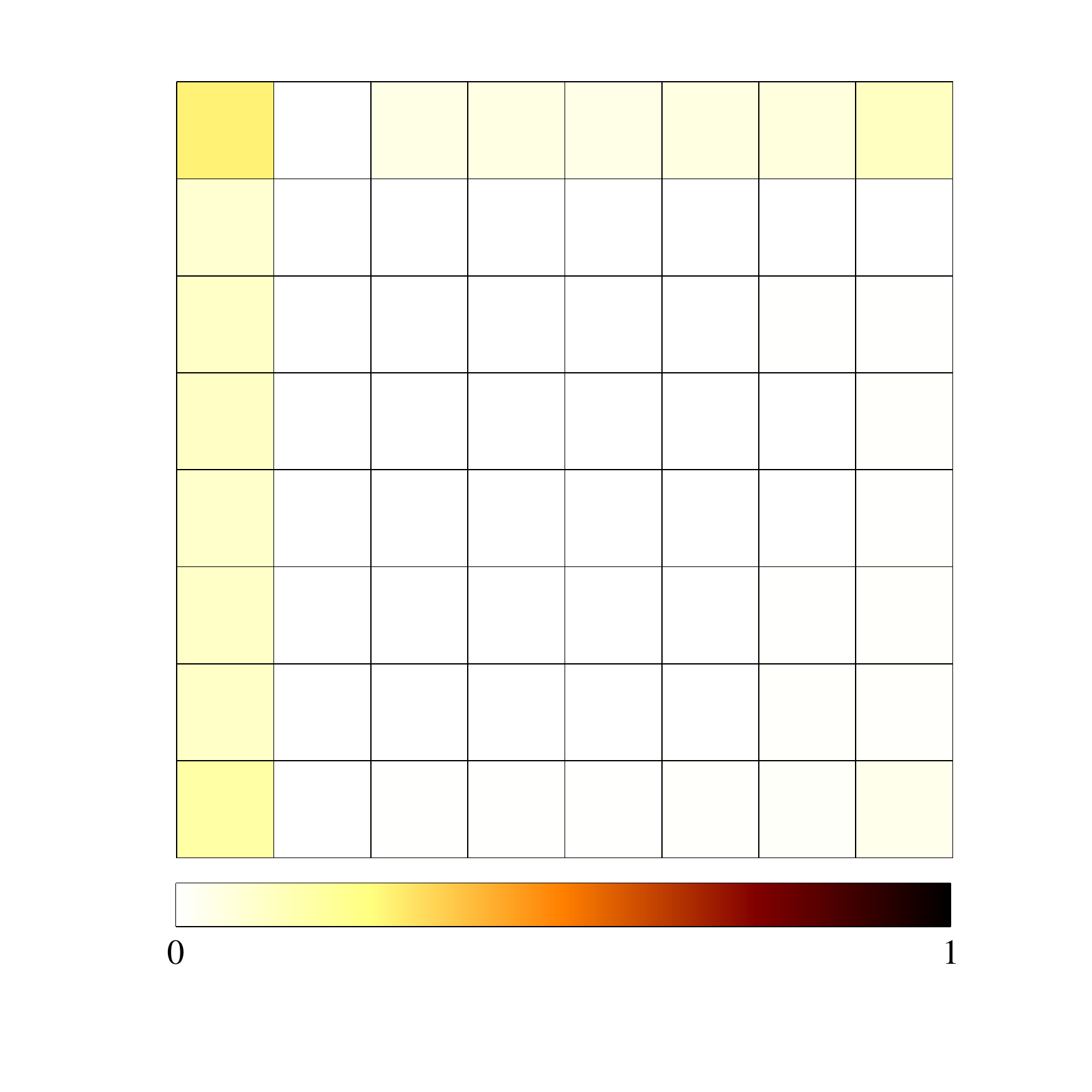}}\;
  \subfloat[Line graph weights for (top) rows  and (bottom) columns]{
  \begin{tabular}{@{}c@{}}
  \includegraphics[trim=55 200 50 185,clip,width=.18\textwidth,valign=c]{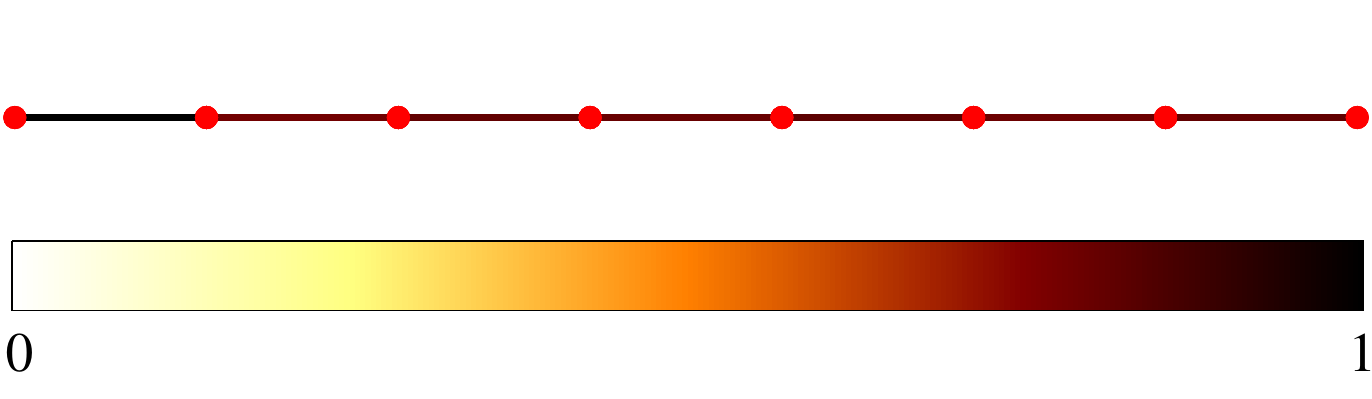} 
   \\  
    \includegraphics[trim=55 200 50 185,clip,width=.18\textwidth,valign=c]{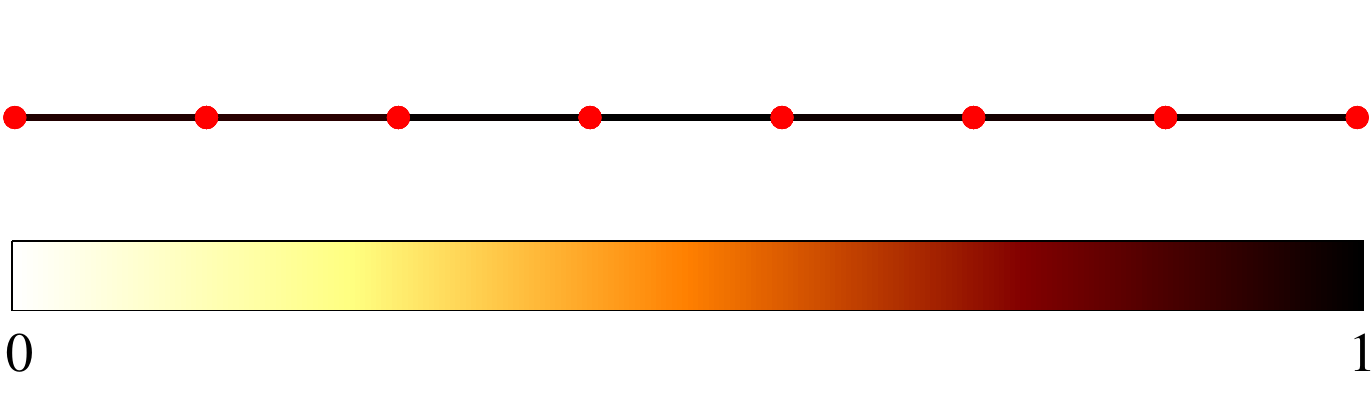} 
    \end{tabular}
  \begin{tabular}{@{}c@{}}
    \includegraphics[trim=55 195 40 185,clip,width=.18\textwidth,valign=c]{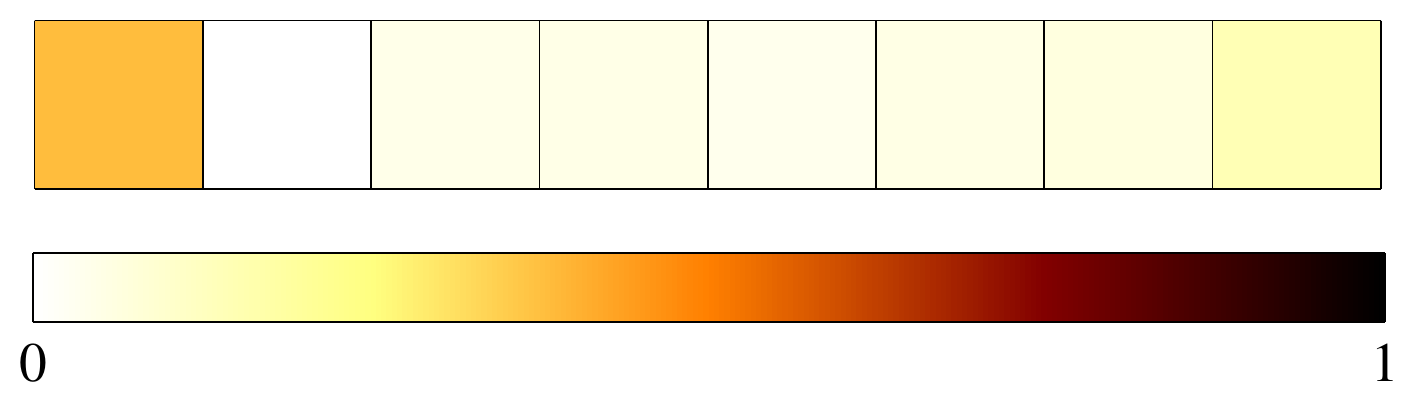}  
	\\
    \includegraphics[trim=55 195 40 185,clip,width=.18\textwidth,valign=c]{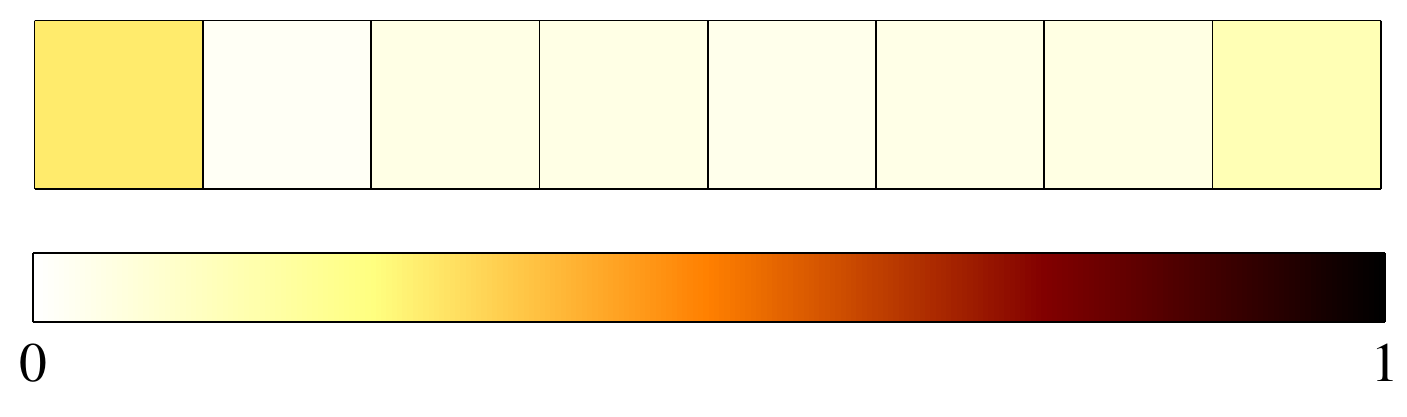}  
  \end{tabular}   
  \vphantom{\includegraphics[trim=55 33 50 20,clip,width=.18\textwidth,valign=c]{NonSep8Intra2_Graph-eps-converted-to.pdf}} 
  }
\caption{For the \emph{DC mode} in \emph{intra prediction} (a) shows the estimated sample variances of $8\times8$ residual signals. In (b) and (c), edge and vertex weights are shown for grid and line graphs learned from residual data, respectively. Darker colors represent larger values.}
\label{fig:graph_weights_intra_dc}

\subfloat[Variances per pixel\label{fig:sample_variance_Hori}]
{\includegraphics[trim=60 80 35 67,clip,height=.18\textwidth,valign=c]{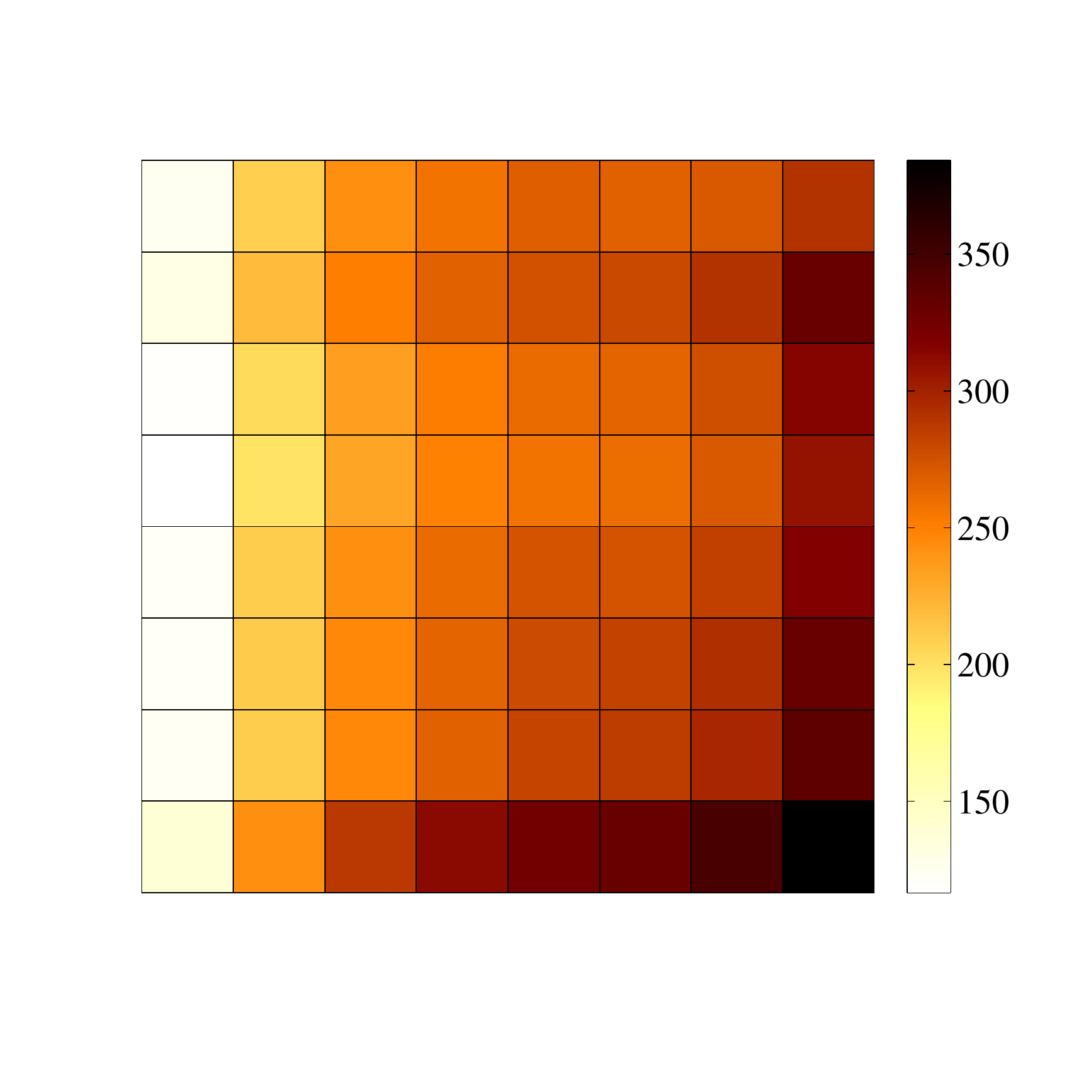}\vphantom{\includegraphics[trim=55 33 50 20,clip,width=.18\textwidth,valign=c]{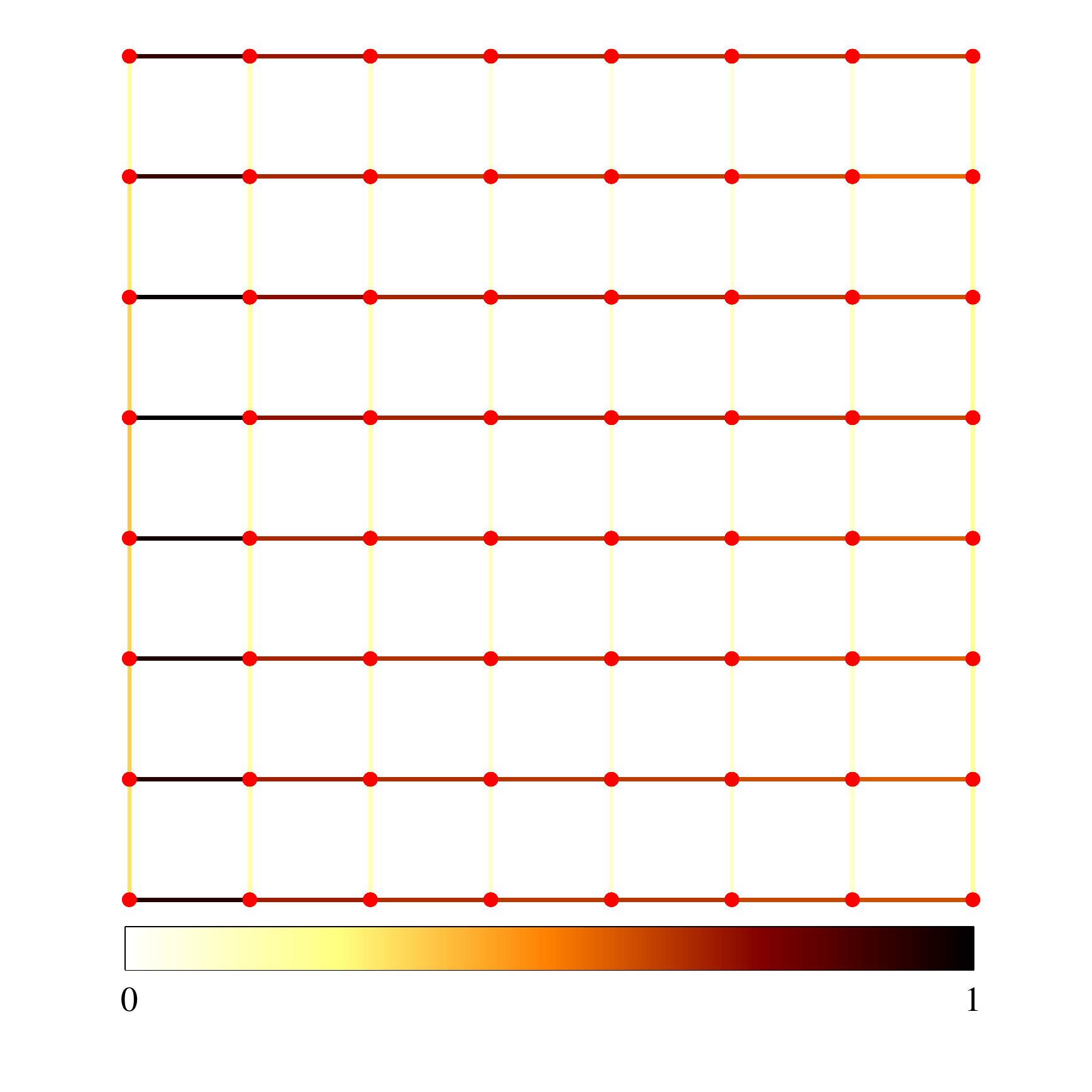}}}\;
    \subfloat[Grid graph weights]{\includegraphics[trim=55 33 50 20,clip,width=.18\textwidth,valign=c]{NonSep8Intra11_Graph-eps-converted-to.pdf}\;
    \includegraphics[trim=75 56 60 35,clip,width=.18\textwidth,valign=c]{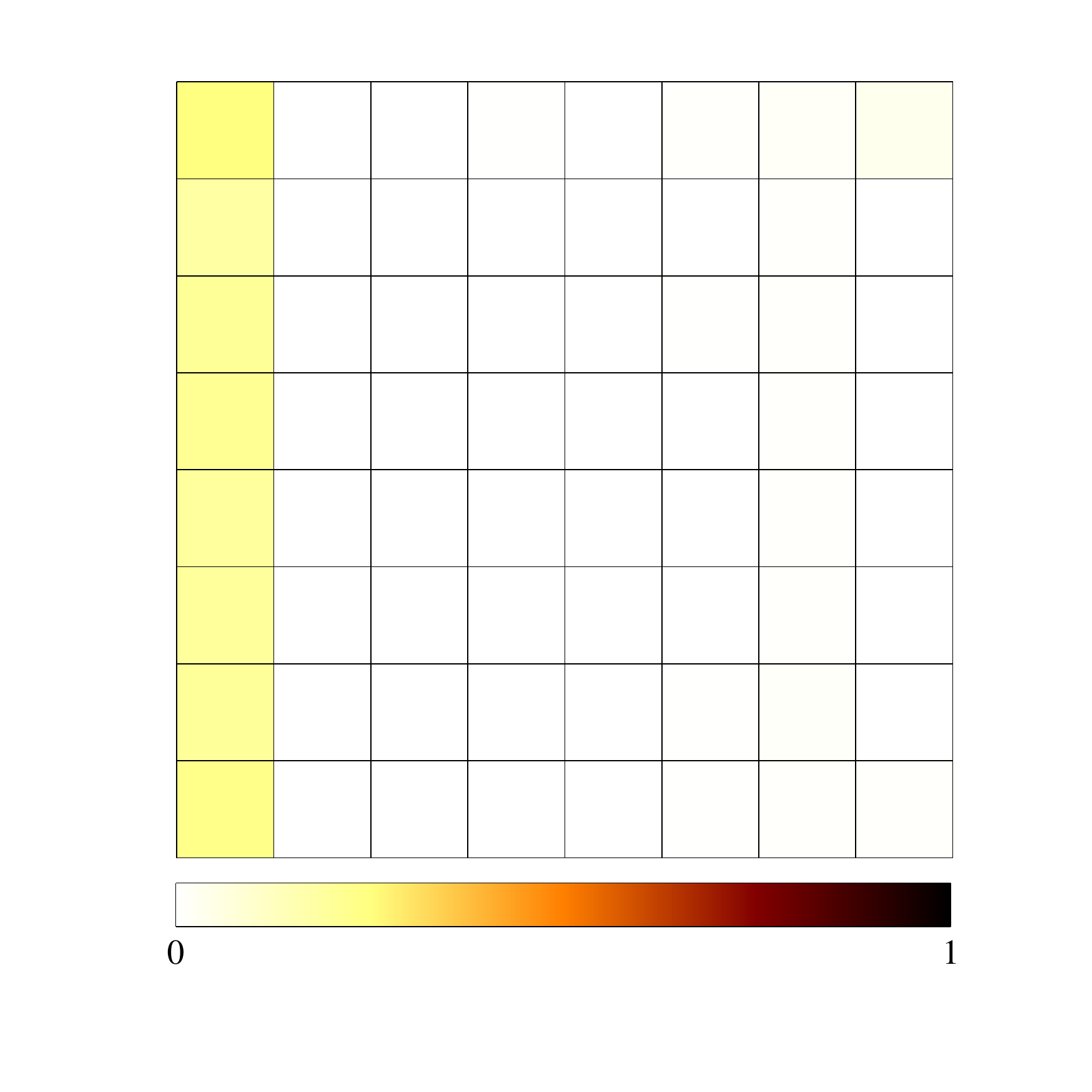}}\;
  \subfloat[Line graph weights for (top) rows  and (bottom) columns]{
  \begin{tabular}{@{}c@{}}
  \includegraphics[trim=55 200 50 185,clip,width=.18\textwidth,valign=c]{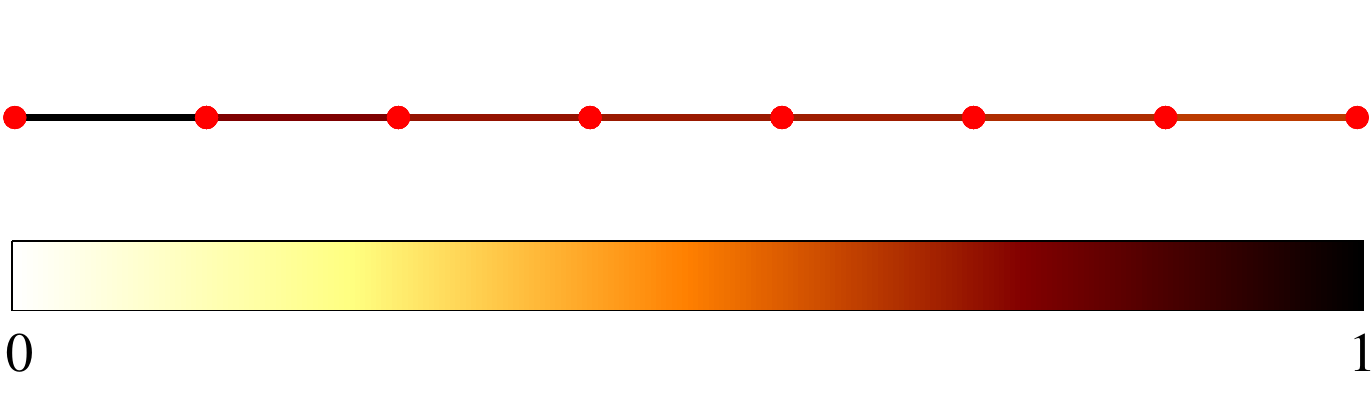} 
   \\  
    \includegraphics[trim=55 200 50 185,clip,width=.18\textwidth,valign=c]{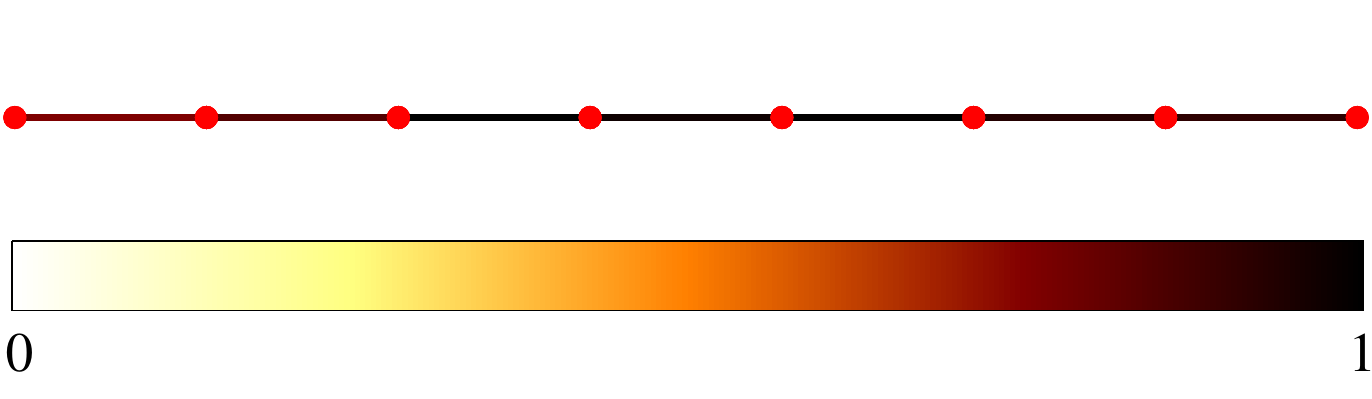} 
    \end{tabular}
  \begin{tabular}{@{}c@{}}
    \includegraphics[trim=55 195 40 185,clip,width=.18\textwidth,valign=c]{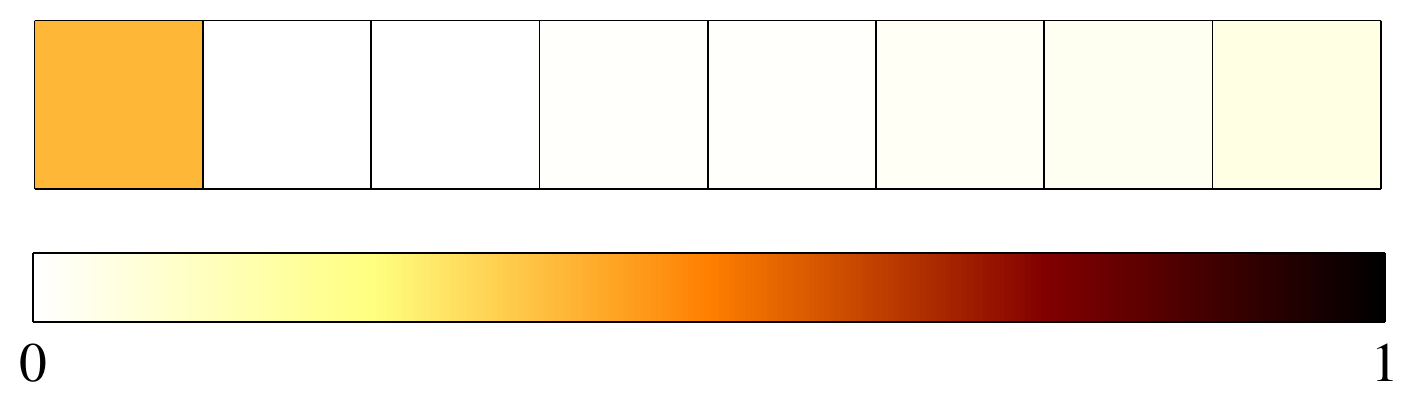}  
	\\
    \includegraphics[trim=55 195 40 185,clip,width=.18\textwidth,valign=c]{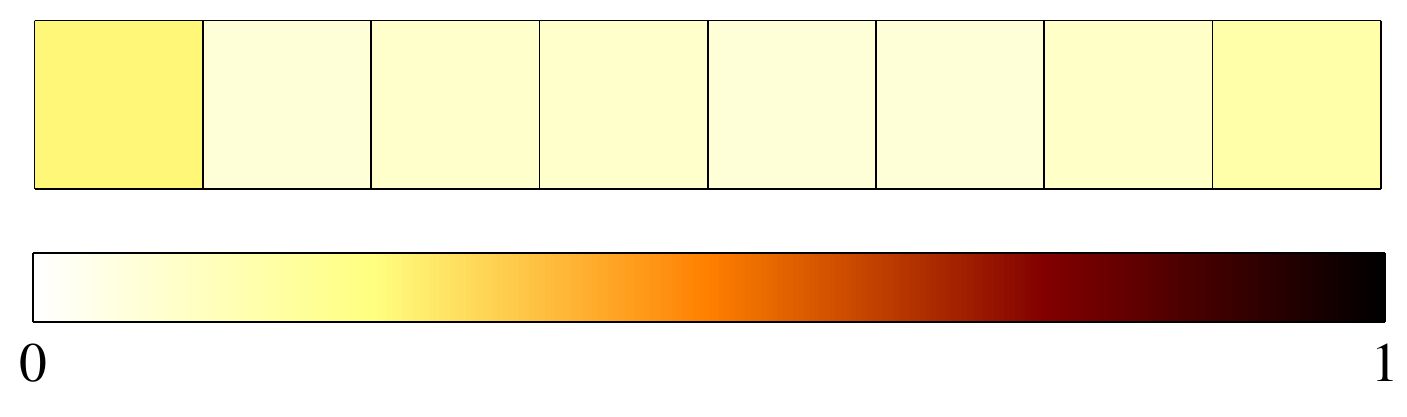}  
  \end{tabular}   
  \vphantom{\includegraphics[trim=55 33 50 20,clip,width=.18\textwidth,valign=c]{NonSep8Intra11_Graph-eps-converted-to.pdf}} 
  }
\caption{For the \emph{horizontal mode} in \emph{intra prediction} (a) shows the estimated sample variances of $8\times8$ residual signals. In (b) and (c), edge and vertex weights are shown for grid and line graphs learned from residual data, respectively. Darker colors represent larger values.}
\label{fig:graph_weights_intra_horizontal}

\subfloat[Variances per pixel\label{fig:sample_variance_Diag}]
{\includegraphics[trim=60 80 35 67,clip,height=.18\textwidth,valign=c]{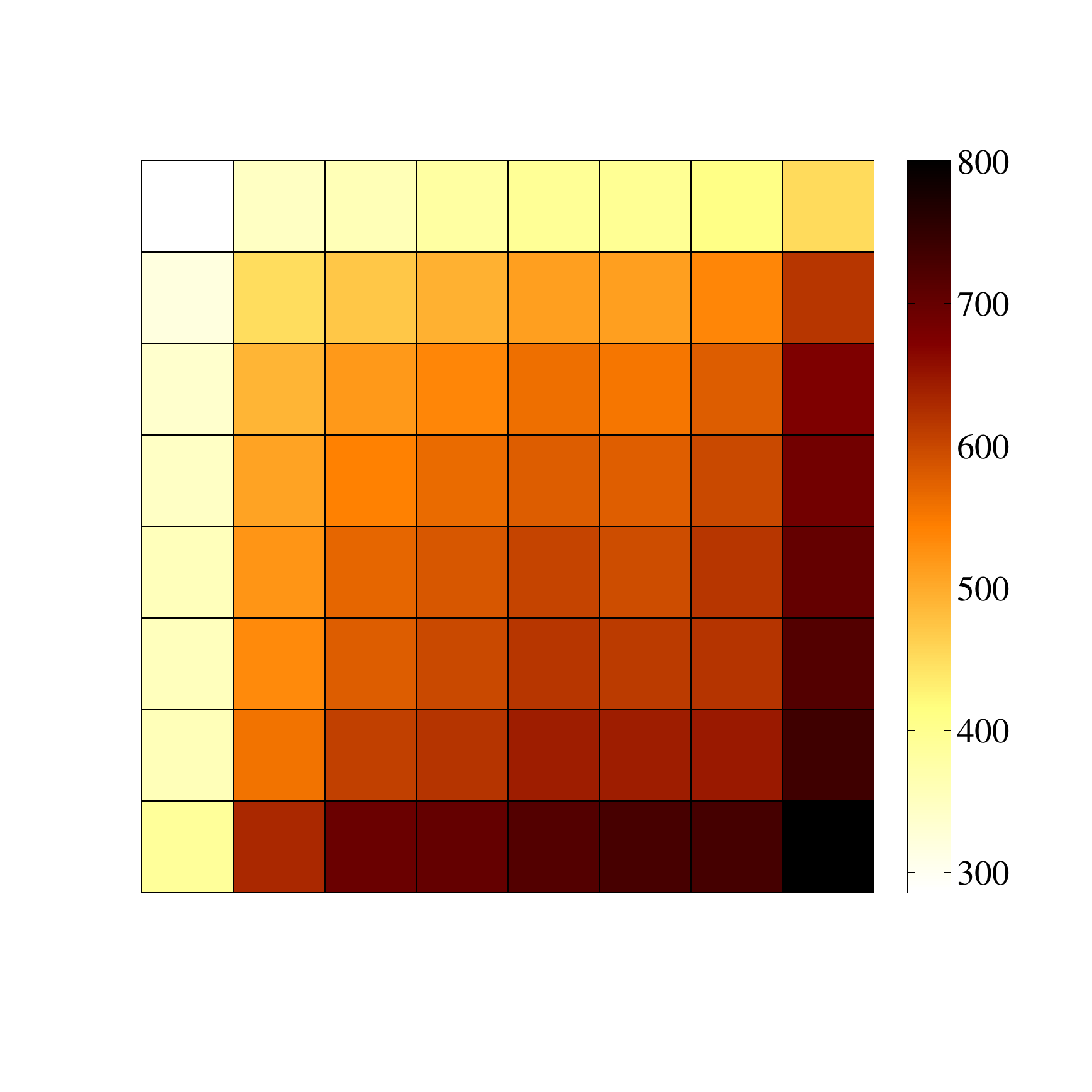}\vphantom{\includegraphics[trim=55 33 50 20,clip,width=.18\textwidth,valign=c]{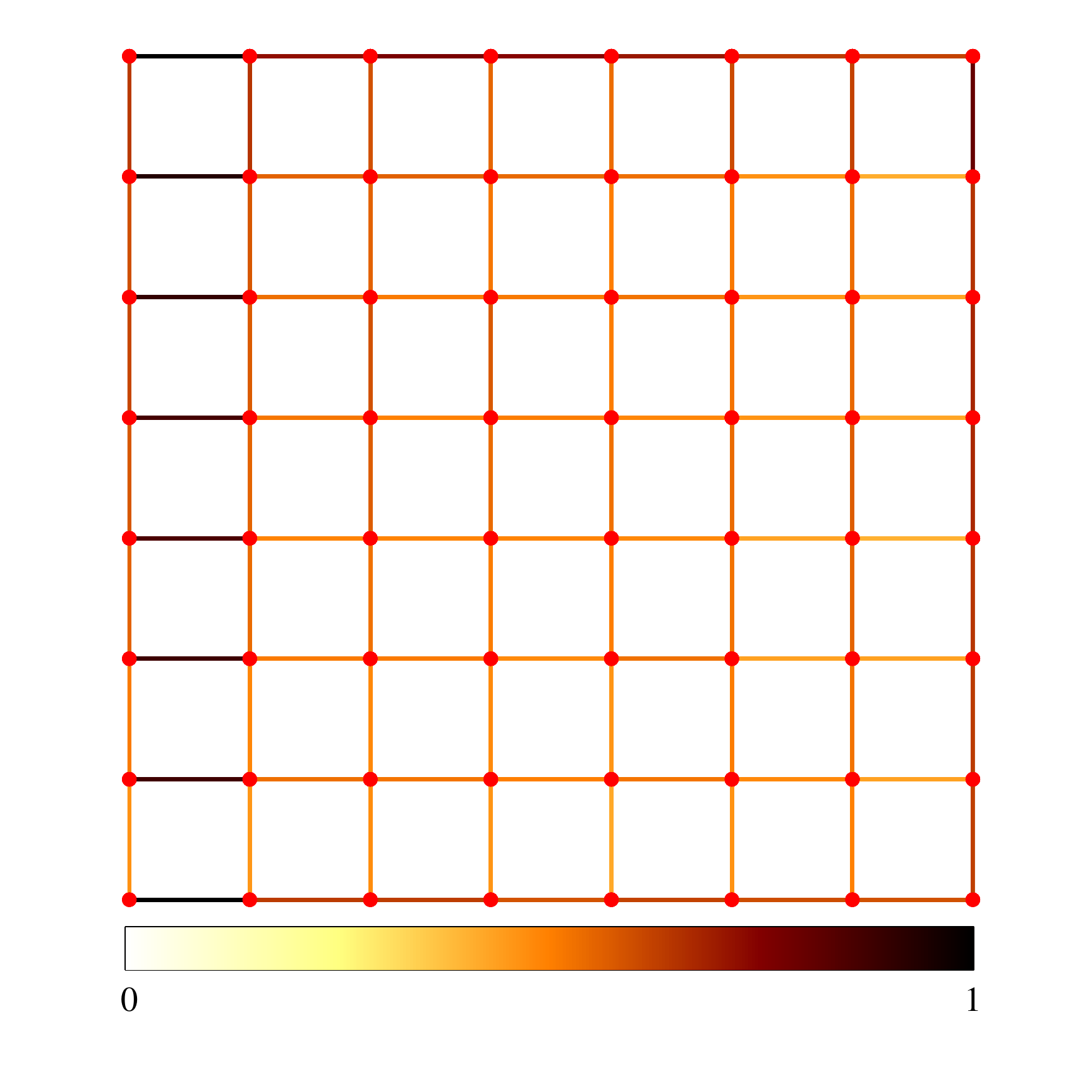}}}\;
    \subfloat[Grid graph weights]{\includegraphics[trim=55 33 50 20,clip,width=.18\textwidth,valign=c]{NonSep8Intra19_Graph-eps-converted-to.pdf}\;
    \includegraphics[trim=75 56 60 35,clip,width=.18\textwidth,valign=c]{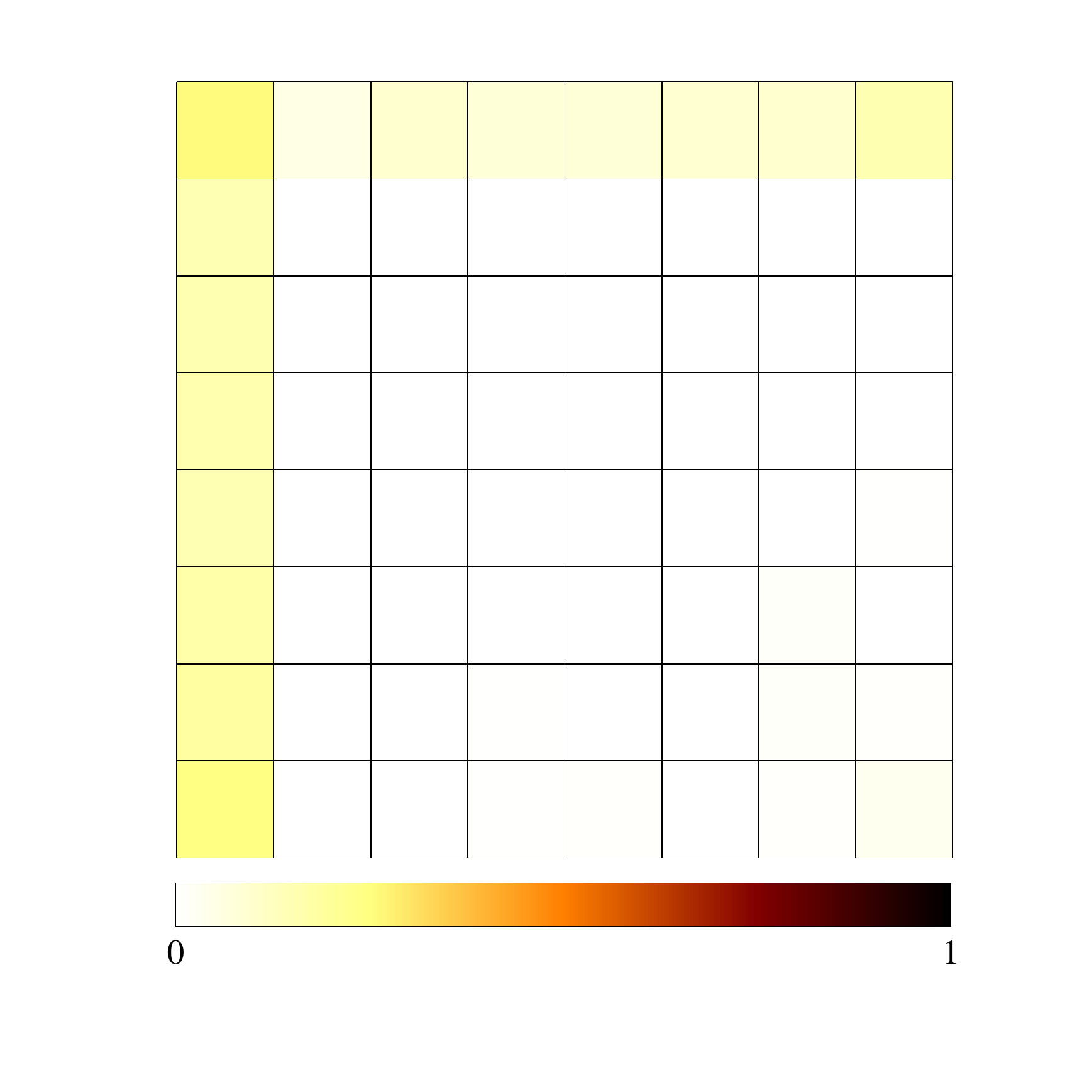}}\;
  \subfloat[Line graph weights for (top) rows  and (bottom) columns]{
  \begin{tabular}{@{}c@{}}
  \includegraphics[trim=55 200 50 185,clip,width=.18\textwidth,valign=c]{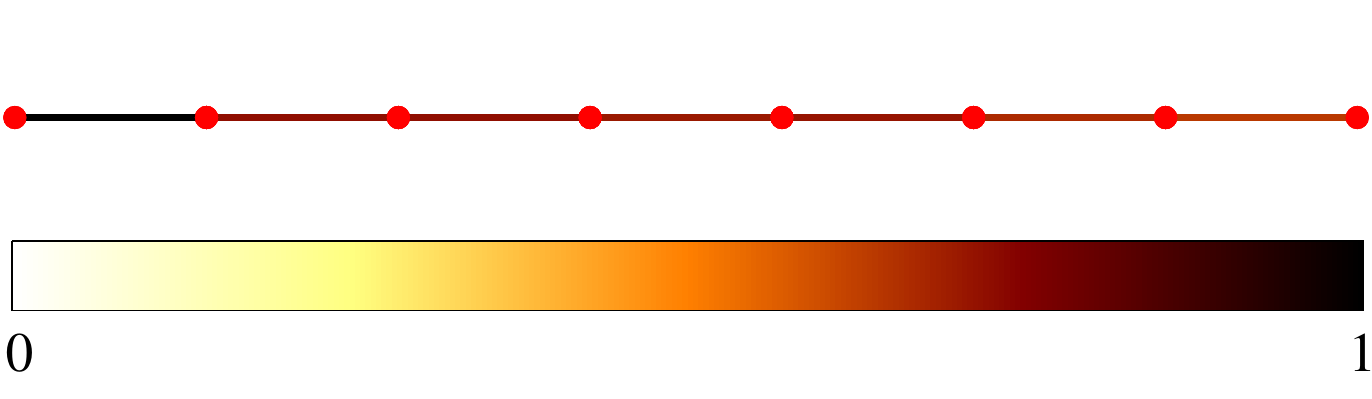} 
   \\  
    \includegraphics[trim=55 200 50 185,clip,width=.18\textwidth,valign=c]{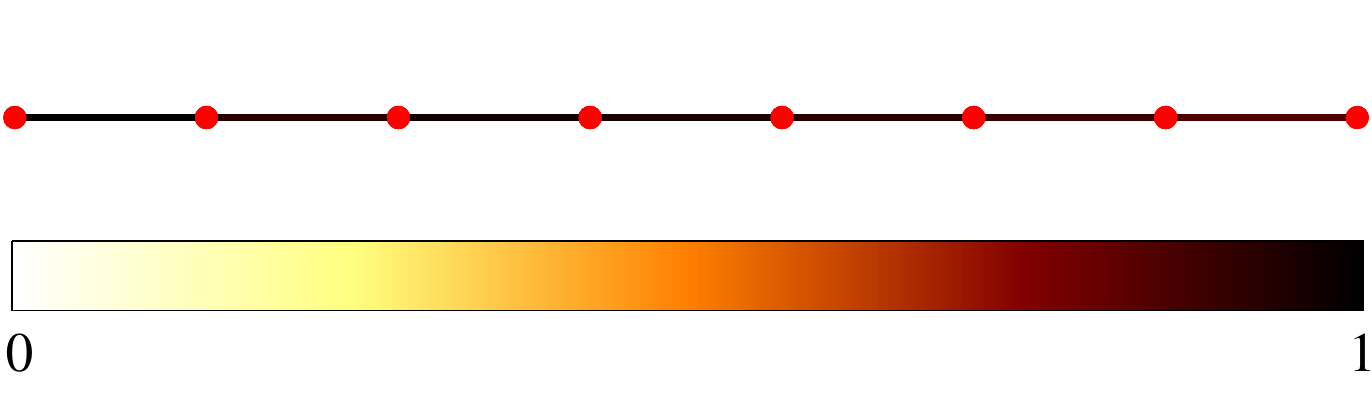} 
    \end{tabular}
  \begin{tabular}{@{}c@{}}
    \includegraphics[trim=55 195 40 185,clip,width=.18\textwidth,valign=c]{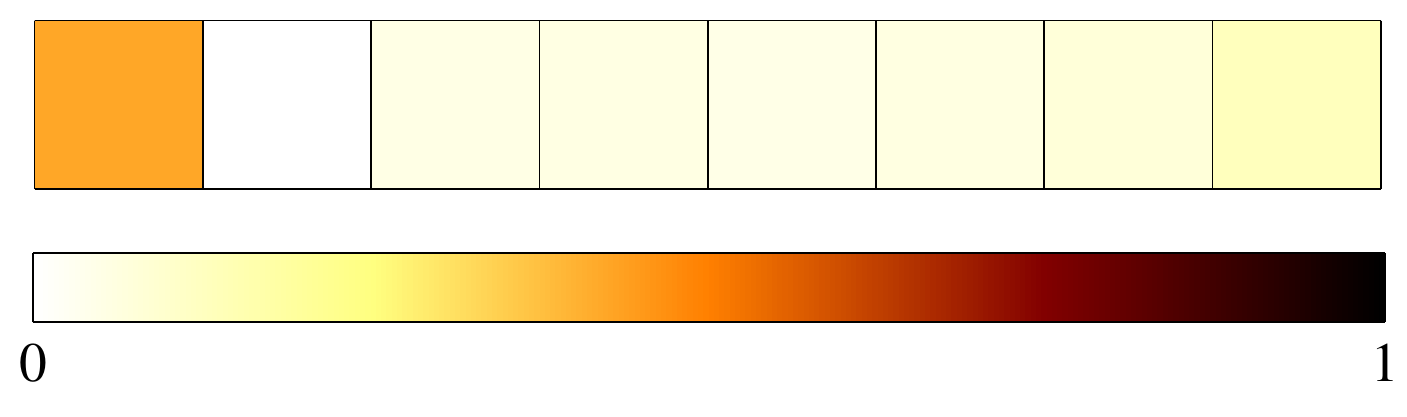}  
	\\
    \includegraphics[trim=55 195 40 185,clip,width=.18\textwidth,valign=c]{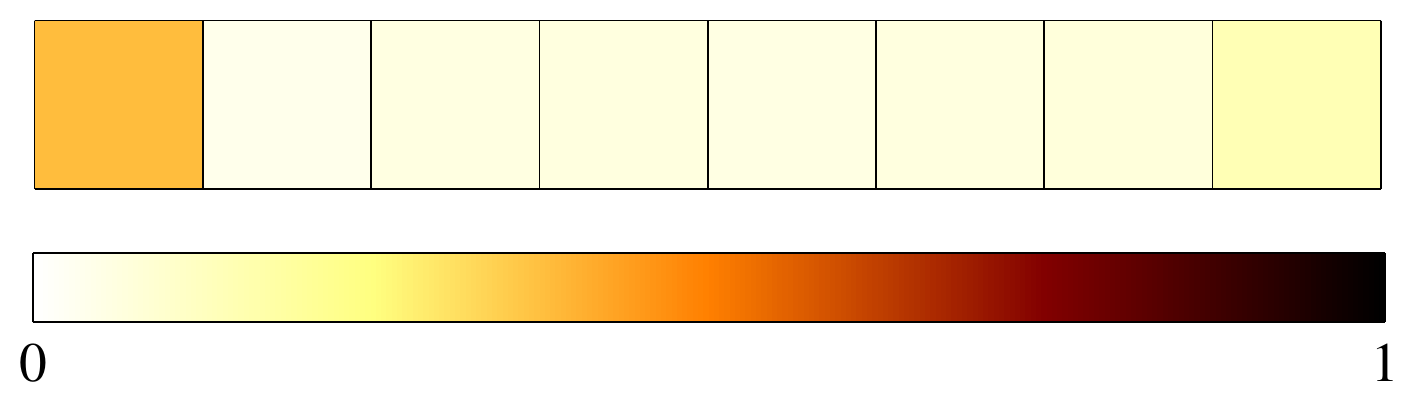}  
  \end{tabular}   
  \vphantom{\includegraphics[trim=55 33 50 20,clip,width=.18\textwidth,valign=c]{NonSep8Intra19_Graph-eps-converted-to.pdf}} 
  }
\caption{For the \emph{diagonal mode} in \emph{intra prediction} (a) shows the estimated sample variances of $8\times8$ residual signals. In (b) and (c), edge and vertex weights are shown for grid and line graphs learned from residual data, respectively. Darker colors represent larger values.}
\label{fig:graph_weights_intra_diagonal}

\end{figure*}

\begin{figure*}[htbp!]
\centering
\subfloat[Variances per pixel\label{fig:sample_variance_2Nx2N}]
{\includegraphics[trim=60 80 35 67,clip,height=.18\textwidth,valign=c]{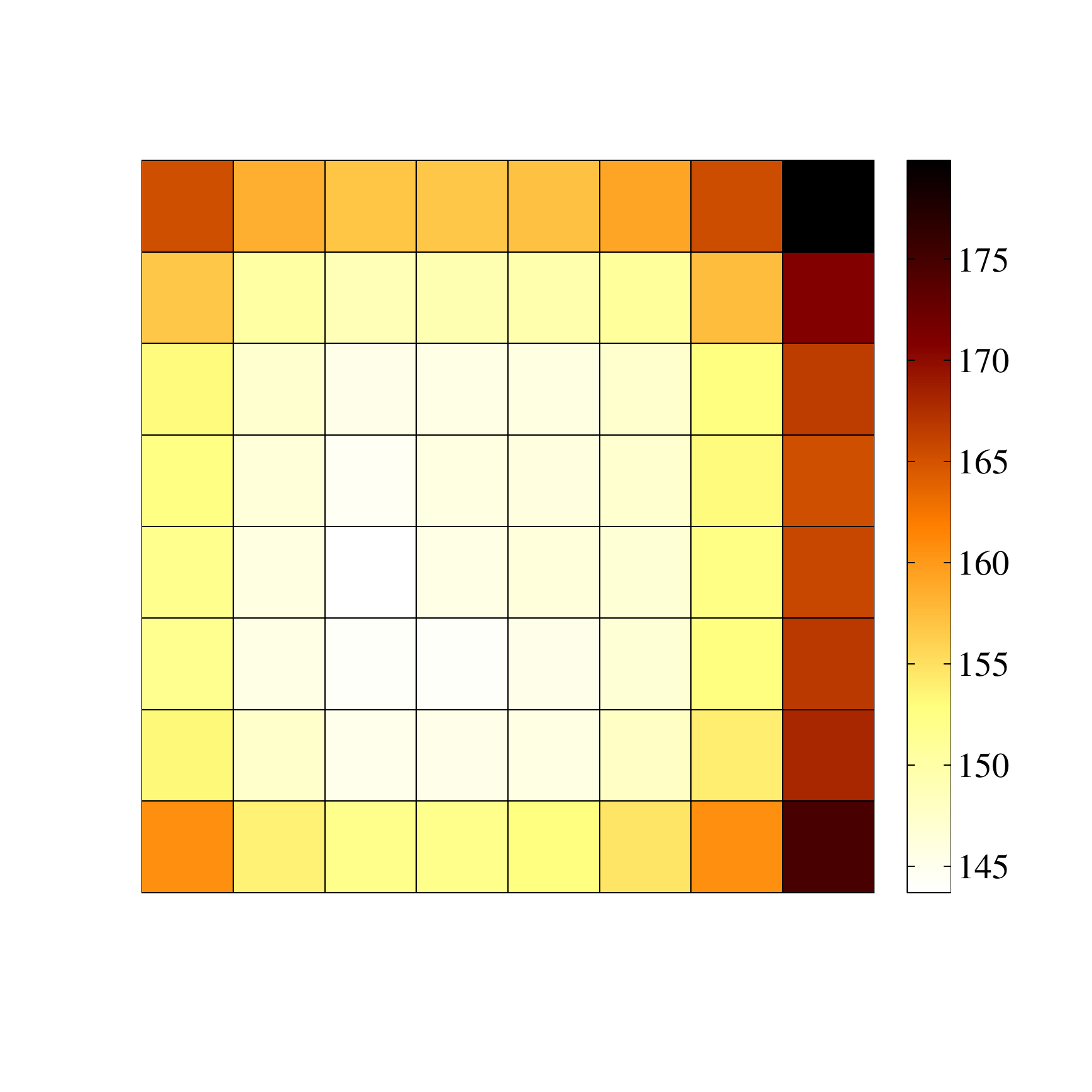}\vphantom{\includegraphics[trim=55 33 50 20,clip,width=.18\textwidth,valign=c]{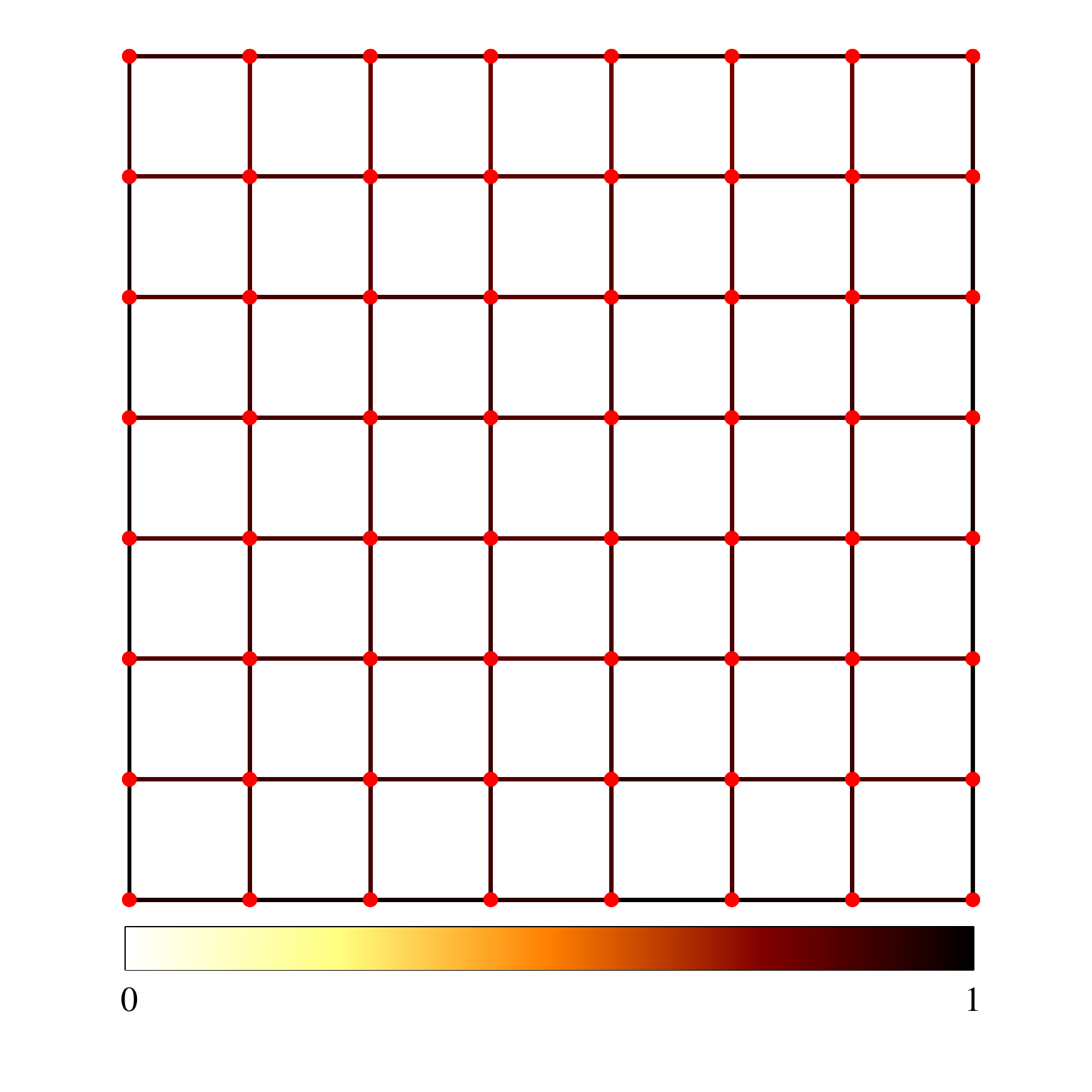}}}\;
    \subfloat[Grid graph weights]{\includegraphics[trim=55 33 50 20,clip,width=.18\textwidth,valign=c]{NonSep8Inter36_Graph-eps-converted-to.pdf}\;
    \includegraphics[trim=75 56 60 35,clip,width=.18\textwidth,valign=c]{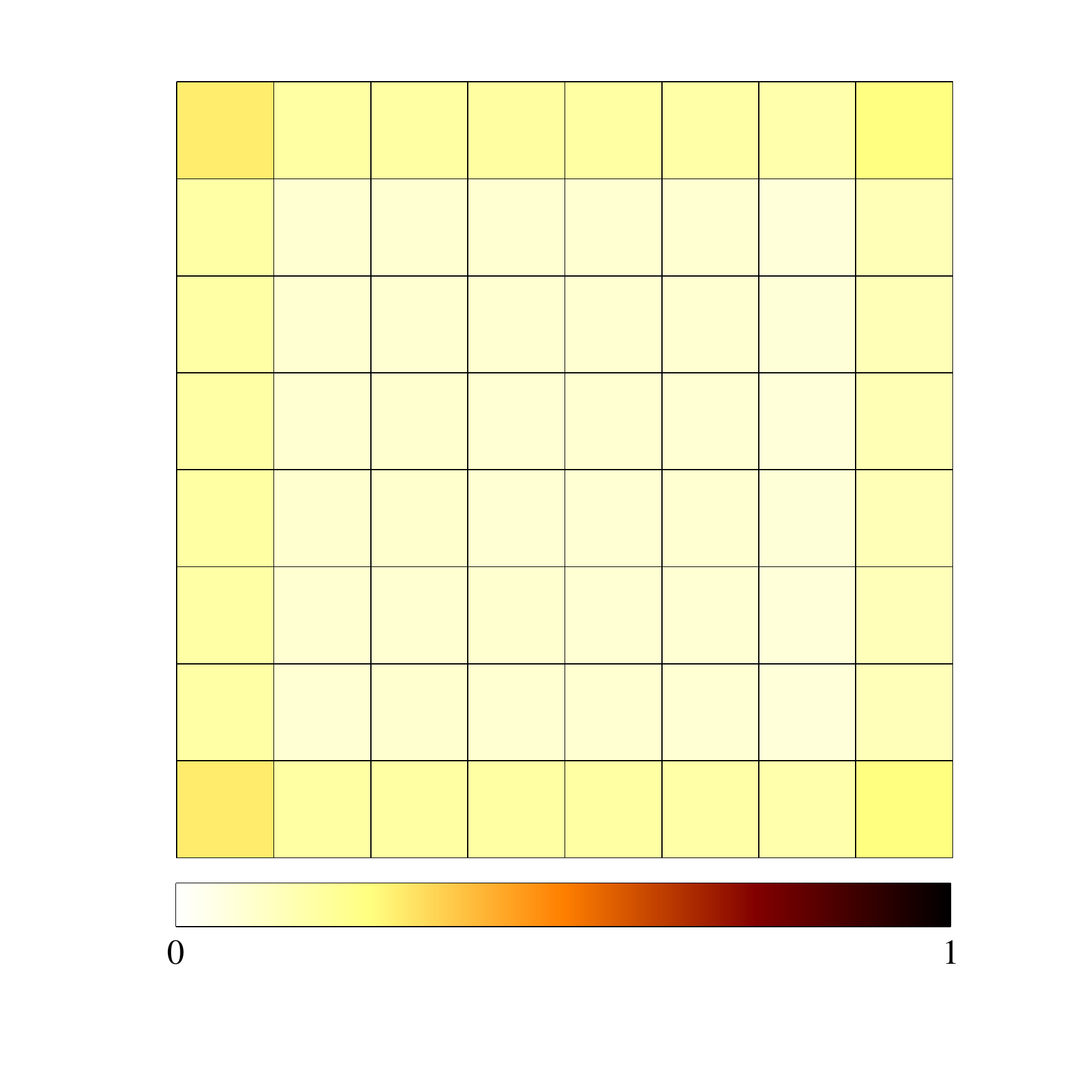}}\;
  \subfloat[Line graph weights for (top) rows  and (bottom) columns]{
  \begin{tabular}{@{}c@{}}
  \includegraphics[trim=55 200 50 185,clip,width=.18\textwidth,valign=c]{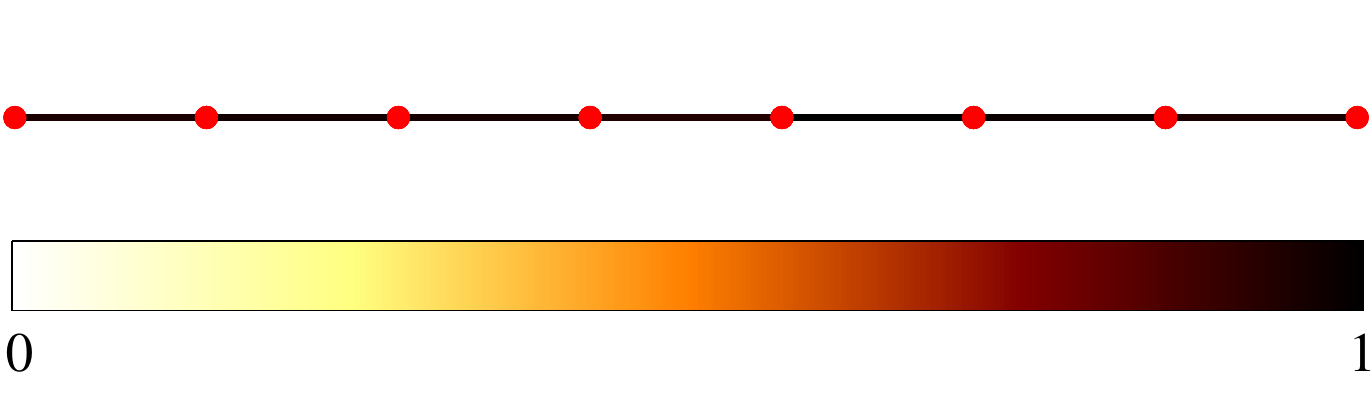} 
   \\  
    \includegraphics[trim=55 200 50 185,clip,width=.18\textwidth,valign=c]{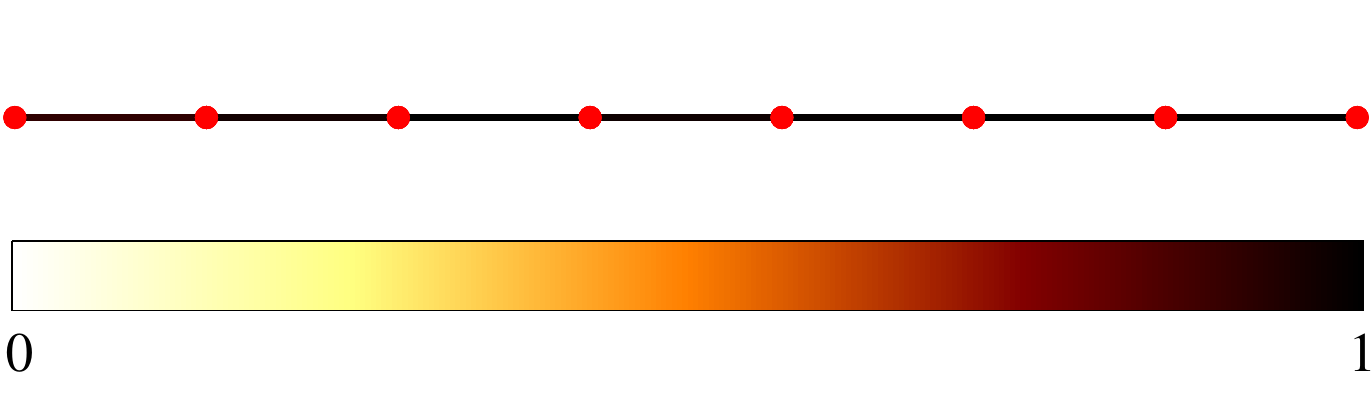} 
    \end{tabular}
  \begin{tabular}{@{}c@{}}
    \includegraphics[trim=55 195 40 185,clip,width=.18\textwidth,valign=c]{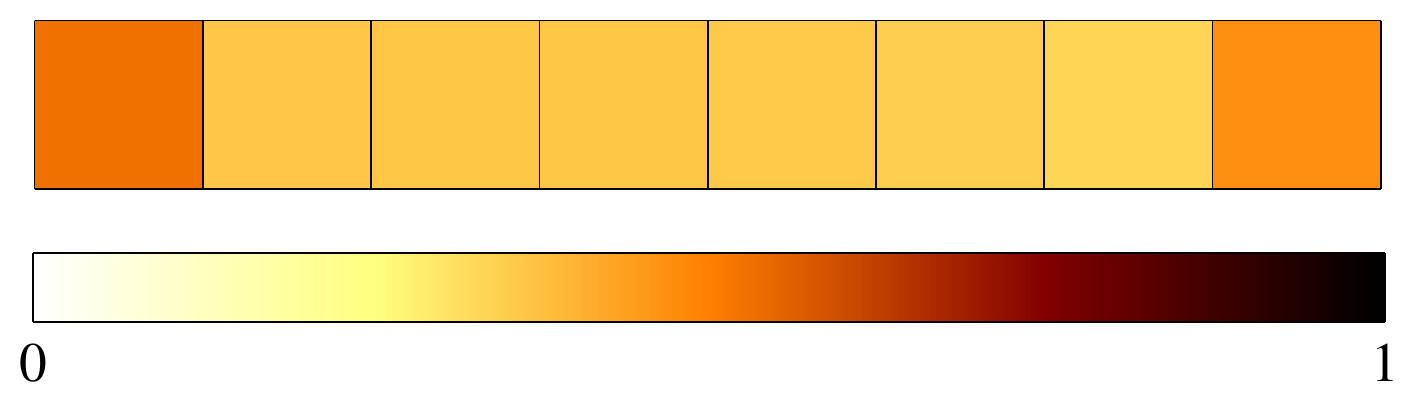}  
	\\
    \includegraphics[trim=55 195 40 185,clip,width=.18\textwidth,valign=c]{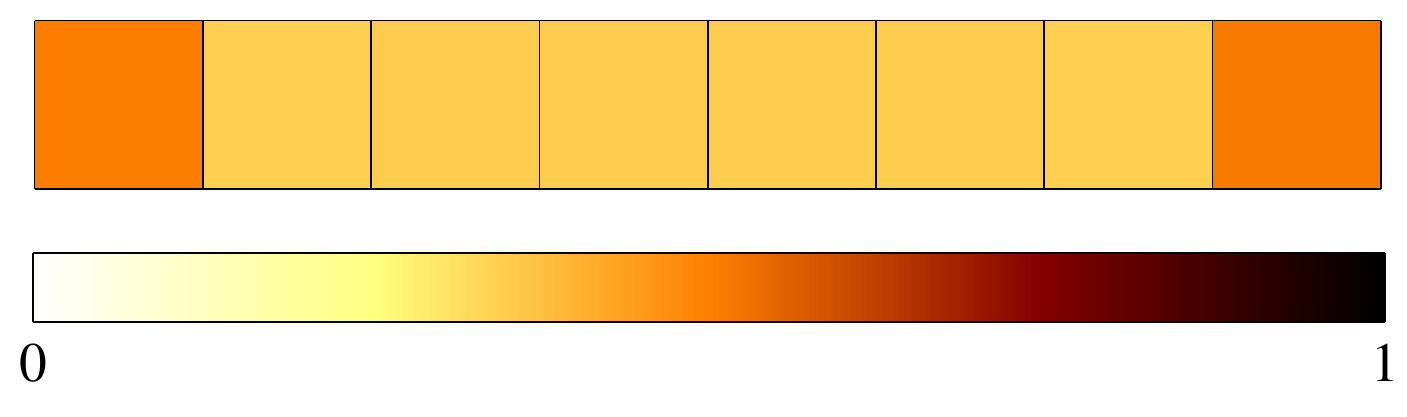}  
  \end{tabular}   
  \vphantom{\includegraphics[trim=55 33 50 20,clip,width=.18\textwidth,valign=c]{NonSep8Inter36_Graph-eps-converted-to.pdf}} 
  }
\caption{For the \emph{PU mode $2N\times 2N$} in \emph{inter prediction} (a) shows the estimated sample variances of $8\times8$ residual signals.  In (b) and (c), edge and vertex weights are shown for grid and line graphs learned from residual data, respectively. Darker colors represent larger values.}
\label{fig:graph_weights_inter_square}

\subfloat[Variances per pixel\label{fig:sample_variance_Nx2N}]
{\includegraphics[trim=60 80 35 67,clip,height=.18\textwidth,valign=c]{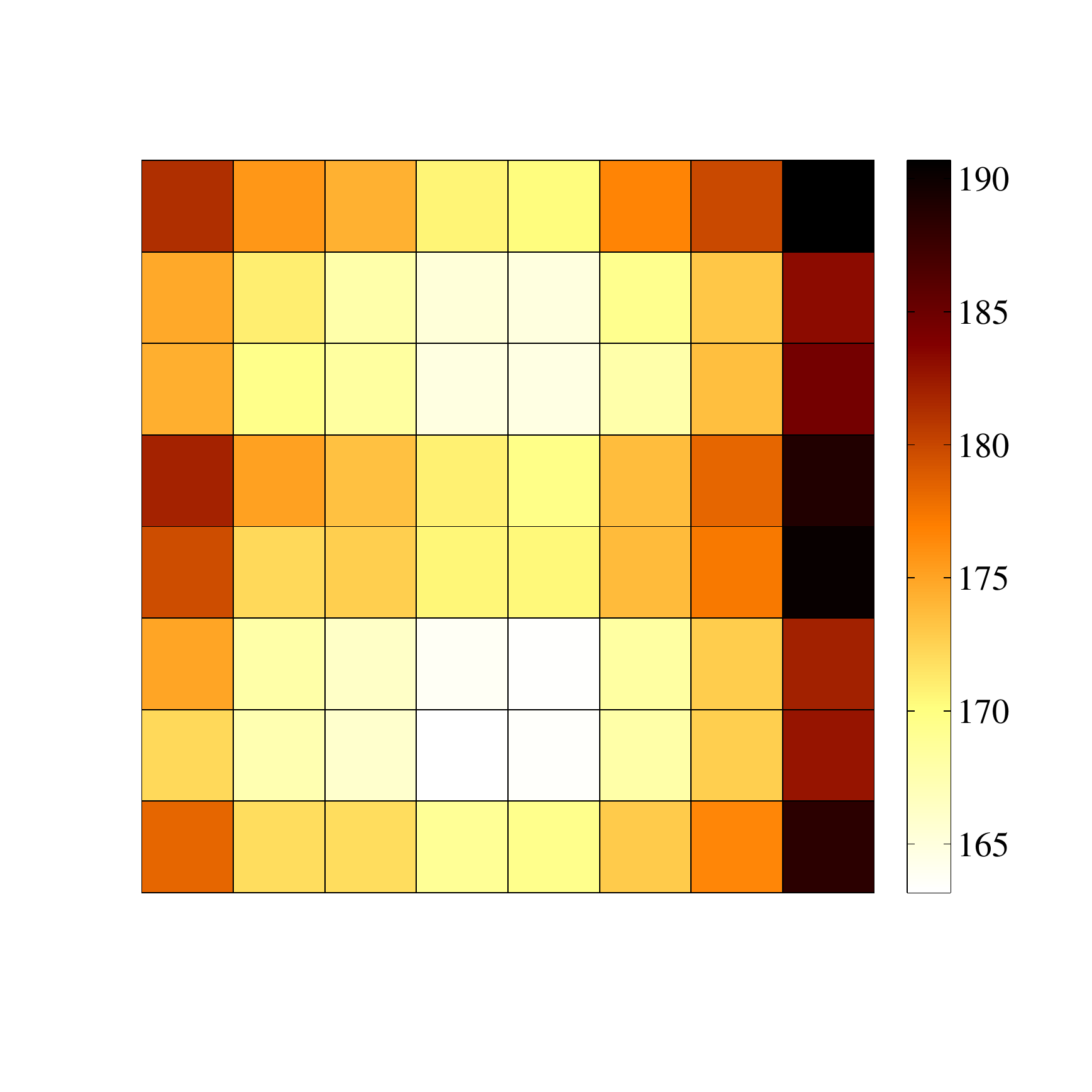}\vphantom{\includegraphics[trim=55 33 50 20,clip,width=.18\textwidth,valign=c]{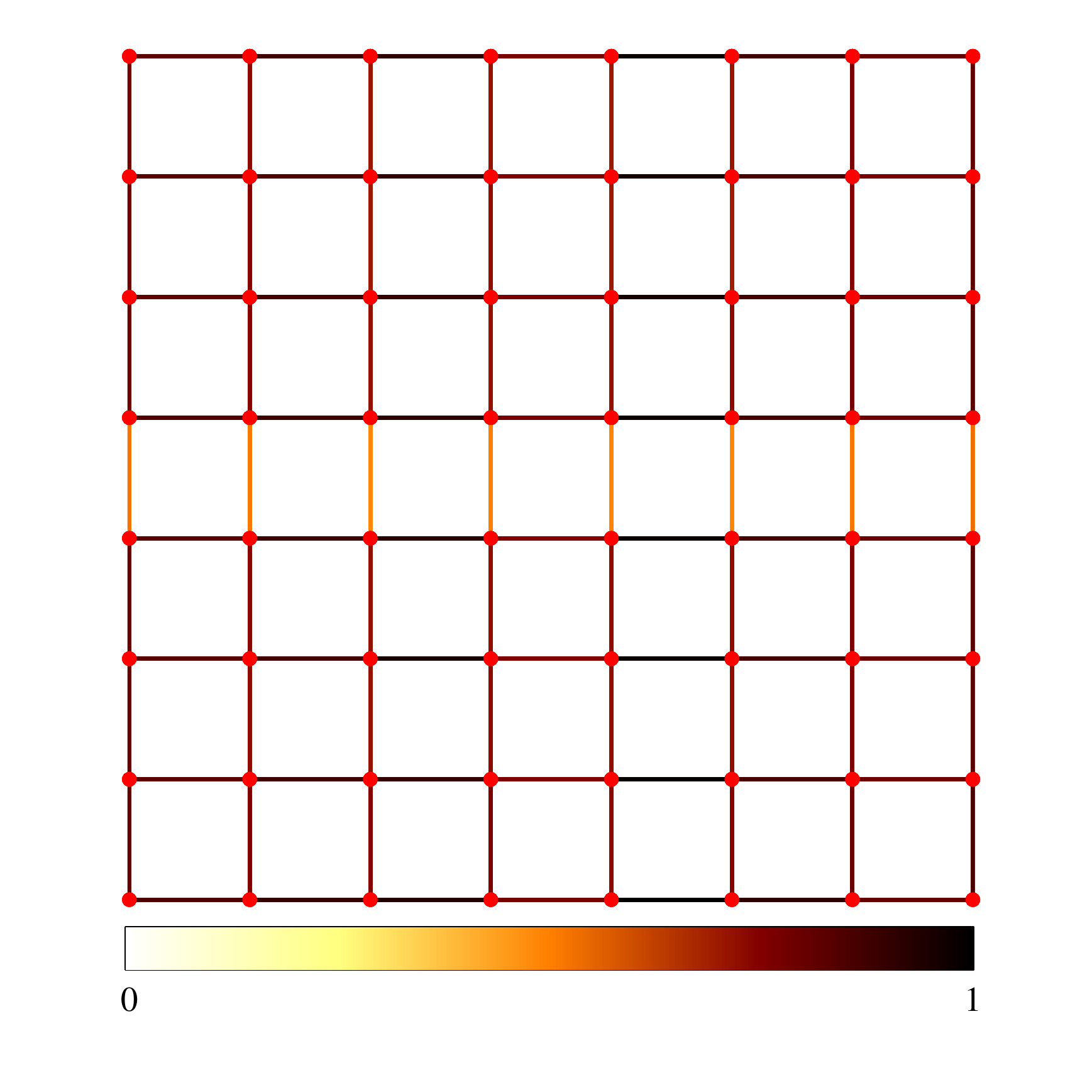}}}\;
    \subfloat[Grid graph weights]{\includegraphics[trim=55 33 50 20,clip,width=.18\textwidth,valign=c]{NonSep8Inter37_Graph-eps-converted-to.pdf}\;
    \includegraphics[trim=75 56 60 35,clip,width=.18\textwidth,valign=c]{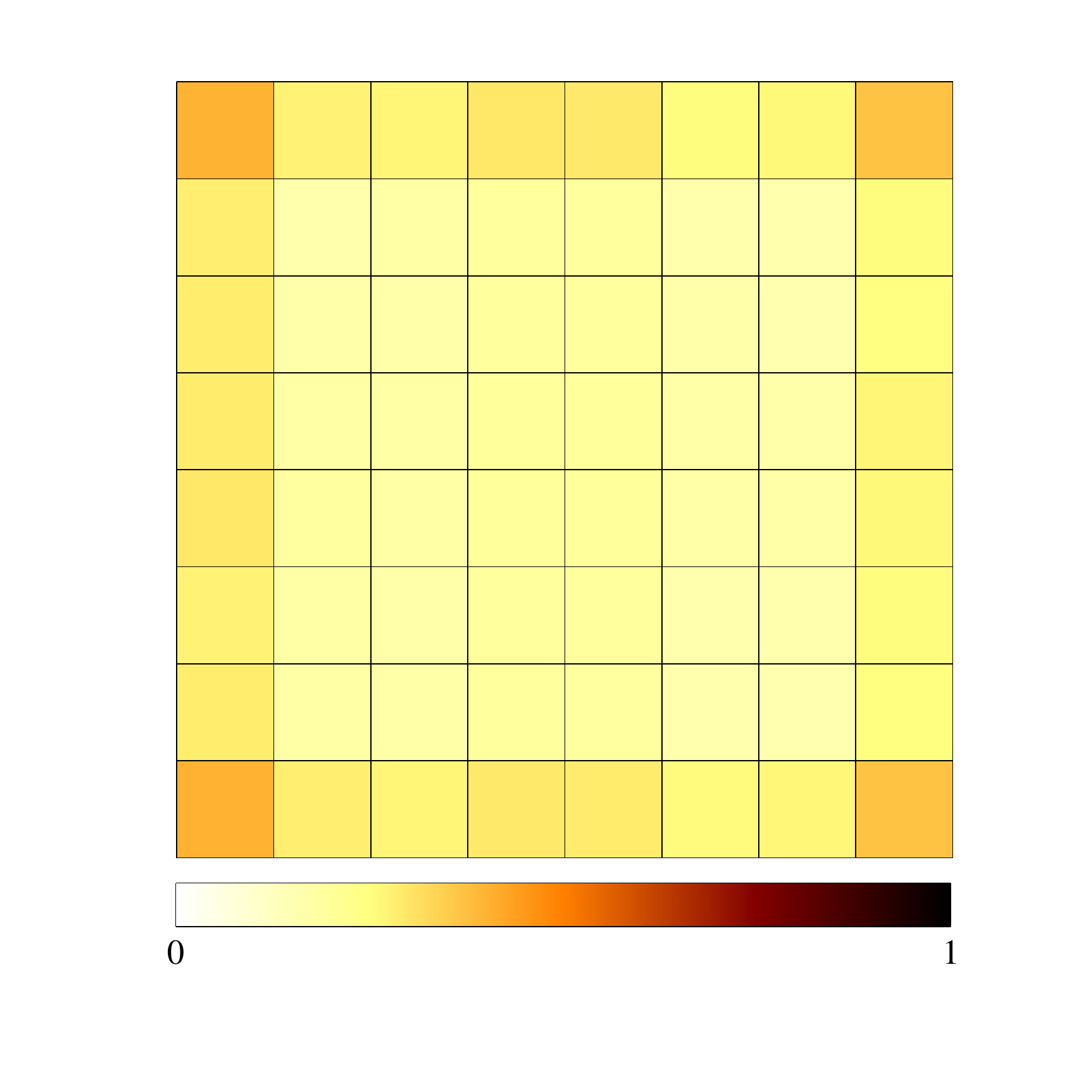}}\;
  \subfloat[Line graph weights for (top) rows  and (bottom) columns]{
  \begin{tabular}{@{}c@{}}
  \includegraphics[trim=55 200 50 185,clip,width=.18\textwidth,valign=c]{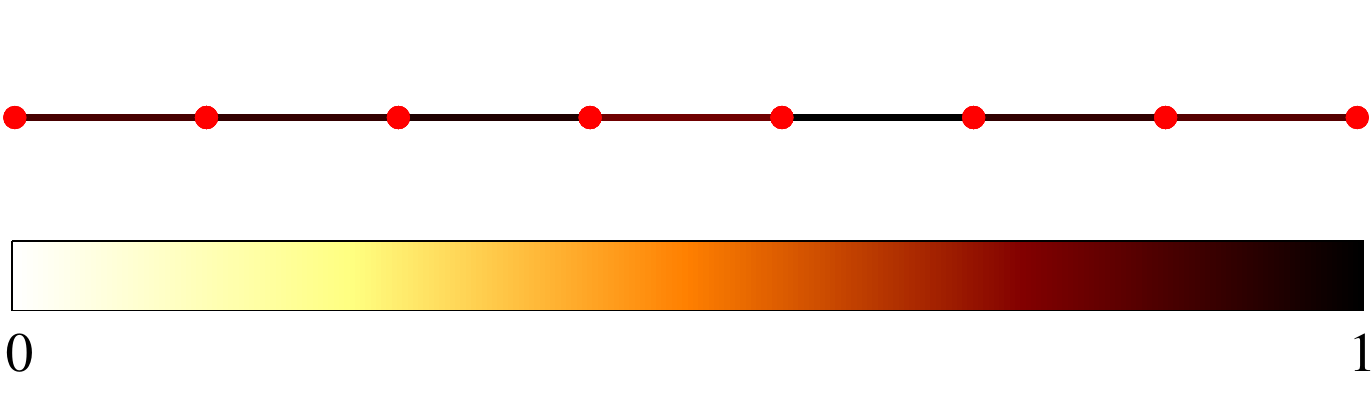} 
   \\  
    \includegraphics[trim=55 200 50 185,clip,width=.18\textwidth,valign=c]{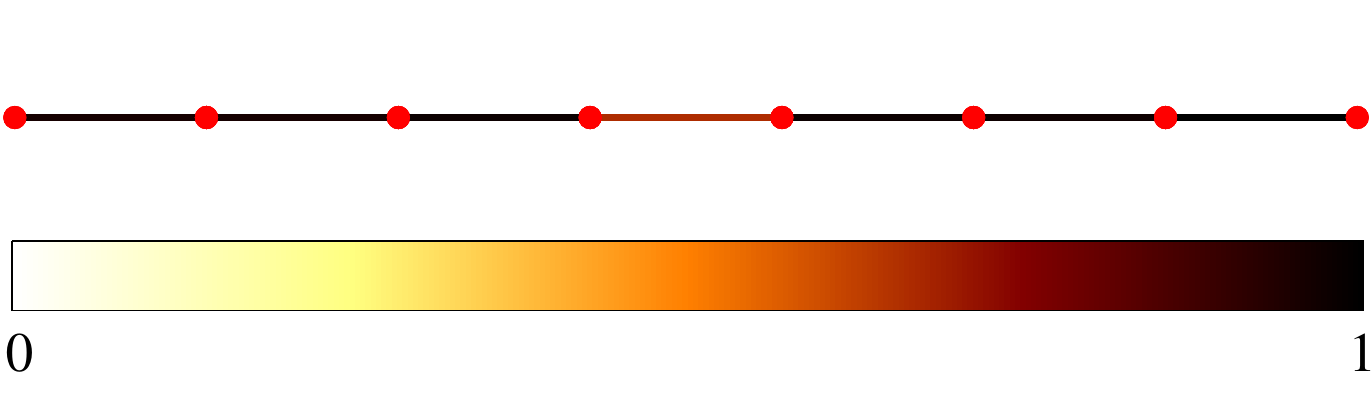} 
    \end{tabular}
  \begin{tabular}{@{}c@{}}
    \includegraphics[trim=55 195 40 185,clip,width=.18\textwidth,valign=c]{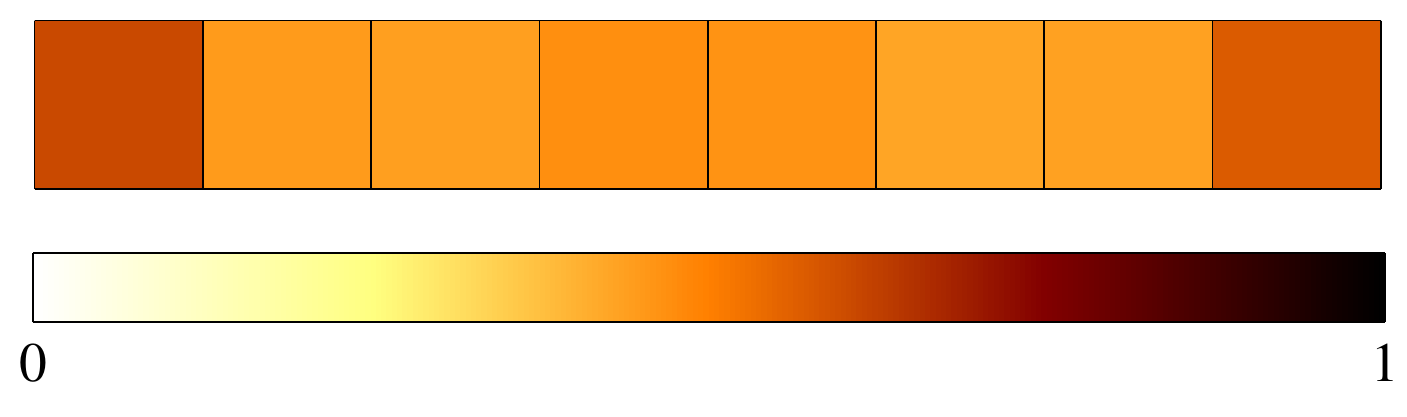}  
	\\
    \includegraphics[trim=55 195 40 185,clip,width=.18\textwidth,valign=c]{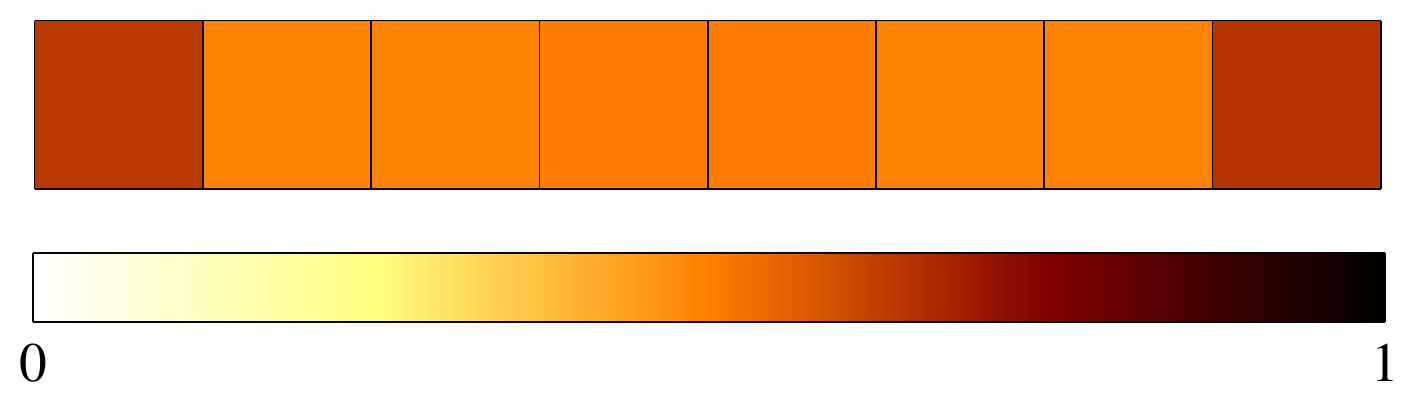}  
  \end{tabular}   
  \vphantom{\includegraphics[trim=55 33 50 20,clip,width=.18\textwidth,valign=c]{NonSep8Inter37_Graph-eps-converted-to.pdf}} 
  }
\caption{For the \emph{PU mode $N\times 2N$} in \emph{inter prediction} (a) shows the estimated sample variances of $8\times8$ residual signals.  In (b) and (c), edge and vertex weights are shown for grid and line graphs learned from residual data, respectively. Darker colors represent larger values.}
\label{fig:graph_weights_inter_rectangular}
\end{figure*}

\section{Experimental Results}
\label{sec:results} 
\subsection{{Experimental Setup}}
In our experiments, we generate two residual block datasets, one for training and the other for testing. The residual blocks {(samples of $\vr$ as shown in Fig.~\ref{fig:encoder_decoder} at the encoder side)} are collected by using the {encoder implemented in} HEVC reference software (HM version 14)\footnote{{HM-14 was the latest version at that time when datasets were generated and also used in our previous work \cite{egilmez:15:gbt_inter,egilmez:16:gbst}. HM-14 implements all normative functionalities of HEVC and the changes between HM-14 and the current latest version (HM-16) are incremental.}}\cite{Sullivan:12:hevc}{, where both datasets are created by encoding the  training and test video sequences} at four different quantization parameters {QPs} from $\{22,27,32,37\}$ {as specified in \cite{bossen:13:hevc_ctc}}.
For the training dataset, residual blocks are obtained by encoding five video sequences, \emph{City}, \emph{Crew}, \emph{Harbour}, \emph{Soccer} and \emph{Parkrun}, and  for the test dataset, we use another five video sequences, \emph{BasketballDrill}, \emph{BQMall}, \emph{Mobcal}, \emph{Shields} and \emph{Cactus}{\cite{Xiph:video_dataset,bossen:13:hevc_ctc}}. 
In both datasets, transform block sizes are restricted to $4 \times 4$, $8 \times 8$ and $16 \times 16$, and {each} residual block {is} classified {(labeled)} using the mode and block size information {determined} 
by the HM encoder {after RD optimization}. Specifically, intra predicted blocks are classified based on 35 intra prediction modes offered in HEVC. Similarly, inter predicted blocks are classified based on 7 different prediction unit (PU) partitions, such that 2 square PU partitions are grouped as one class and other 6 rectangular PU partitions determine other classes. Hence, datasets consist of $35\!+\!7\!=\!42$ classes in total.
For {each combination of 42 classes and block sizes ($4 \times 4$, $8 \times 8$ and $16 \times 16$)}, the optimal {GL-GBSTs}, {GL-GBNTs} and {nonseparable} KLTs are designed {offline} using the residual blocks in the training dataset. {As designing separate sets of transforms per QP 
does not improve the coding performance based on our experiments, a single set of transforms (per class and block size) is designed for all QPs without splitting the training dataset for each QP.} 
On the other hand, EA-GBTs are constructed {on a per-block basis} based on the detected image edges. The details of transform the construction are discussed in Sections \ref{sec:graph_learning_video} and \ref{sec:eagbt}.

{In order to solely evaluate the effect of transforms on coding performance, transform coding
is performed out of the HM encoding loop\footnote{{The simulation setups in  \cite{Fracastoro:2018:apprx_graph,Guleryuz:2015:SOT_TIP,hu:14:PWS,icip17_compression_challange} are a few examples where experiments are conducted out of the encoding loop. Our setup specifically decouples the transform coding part from the rest of the HM encoding loop.}}} 
{while retaining all the other coding decisions made by HM on partitioning, prediction and filtering.} 
{Note that all transforms are compared based on rate-distortion performance on the \textit{same} residual data. That is, we generate residual data first using HM and then compare the performance of DCT, KLT and GBT on those residuals.}
{In our experimental setup, the}
{transform coding first applies designated transforms (e.g., DCT or GBT) on the residual blocks,  {associated with a given QP in the test dataset}, 
then the resulting transform coefficients are uniformly quantized {using the same quantization step size derivation in HM (derived from QP)} and entropy coded using arithmetic coding \cite{Said:2004:arithmetic_tech_report} {to generate bitstreams}. 
{The proposed and baseline transforms are tested on mode-dependent transform (MDT) and rate-distortion optimized transform (RDOT) schemes. Specifically, the MDT scheme assigns a single GBT/KLT trained for each class/mode and block size.
The RDOT scheme selects the best transform from a predefined set of transforms, denoted as $\mathcal{T}$,} 
by minimizing the rate-distortion cost $J(\lambda_{\text{rd}}) = D + \lambda_{\text{rd}} R$  \cite{Ortega:98:rdm} where the multiplier $\lambda_{\text{rd}}=0.85\times 2^{(QP-12)/3}$ \cite{Sullivan:12:hevc} depends on the $QP$ parameter.
{In our RDOT experiments, all transform sets include DCT in addition to a GL-GBT (trained per-class and per-block size) or an EA-GBT (constructed per-block).} 
{Any side information needed to determine transforms (at the decoder side) are included {in the bitstreams} as signaling overhead. In the RDOT scheme, the transform index is signalled using truncated unary codes \cite{Bross:20:vvc10}\footnote{{Note that MDT does not require any signaling for transforms, because a single transform is available per-mode.}}.
If the encoder chooses} EA-GBT, the necessary graph (i.e., image edge) information is further sent by using the arithmetic edge encoder (AEC)  \cite{gene:2014:aec}.} 

The compression performance {across different bitrate operating points} is measured using Bjontegaard-delta rate (BD-rate)\cite{bd_metric} metric {by performing the transform coding experiments at QPs 22, 27, 32 and 37}.

\subsection{Compression Results}
 Table \ref{table:results_mdt_vs_rdot} presents the overall coding gains achieved by using KLTs, {GL-GBST}s and {GL-GBNT}s with MDT and RDOT schemes for intra and inter predicted blocks. According to the results, {GL-GBNT} outperforms KLT irrespective of the prediction type and coding scheme. Fig.~\ref{fig:generalization} further demonstrates the advantage of proposed approach over KLT when fewer number of training samples are available, where the performance difference between {GL-GBNT} and KLT is increased as the number of available training samples are reduced. Specifically, the BD-rate gap between {GL-GBNT} and KLT increases from 0.7\% to 1.5\% when twenty-percent of the training data is used instead of the complete data. This validates our observation that the proposed graph learning method leads to a more robust transform and provides a better generalization than KLT. Table \ref{table:results_mdt_vs_rdot} also shows that {GL-GBNT} performs substantially better than {GL-GBST} for coding intra predicted blocks, while for inter blocks {GL-GBST} performs slightly better than {GL-GBNT}. This is because, inter predicted residuals tend to have a separable structure as shown in Figs.~\ref{fig:graph_weights_inter_square} and \ref{fig:graph_weights_inter_rectangular}, 
 yet intra residuals have more directional structures as shown in Figs.~\ref{fig:graph_weights_intra_planar}, \ref{fig:graph_weights_intra_dc} and \ref{fig:graph_weights_intra_diagonal}, which are better captured by using non-separable transforms. Moreover, RDOT scheme significantly outperforms MDT, since RDOT has multiple transform candidates providing a better adaptation to different block characteristics with RD optimization, while MDT offers a single transform per-mode\footnote{In our RDOT experiments on intra predicted residuals, GL-GBNT is selected 62\% of the time against DCT for the diagonal mode, while GL-GBNT is selected 33\% of the time for the DC mode. 
 This indicates that per-mode block characteristics can significantly differ, where residuals obtained from certain modes (such as DC) can be better captured with DCT.}. {Naturally, this coding gain comes at a cost for encoders as they need to search for the additional transform mode, which increased the encoder run-time about 10\% and has no run-time impact at the decoder in our experiments.}
 
\begin{table}[!t]
\centering
\caption{Comparison of KLT, {GL-GBST} and {GL-GBNT} with MDT and RDOT schemes in terms of BD-rate ($\%$ bitrate reduction) with respect to the DCT. Smaller (negative) BD-rates mean better compression.}
\begin{tabular}{|c|c|c|c|c|}
\hline
\multirow{2}{*}{Transform} & \multicolumn{2}{c|}{Intra Prediction} & \multicolumn{2}{c|}{Inter Prediction}  \\ \cline{2-5}
& MDT & RDOT & MDT & RDOT  \\ \hline
KLT & ${-1.81}$ & ${-6.02}$ & ${-0.09}$  & ${ -3.28}$ \\ \hline
{GL-GBST} & $-1.16$ & ${-4.61}$ & $\mathbf{-0.25}$ & $\mathbf{-3.89}$ \\ \hline
{GL-GBNT} & $\mathbf{-2.04}$ & $\mathbf{-6.70}$ & ${-0.18}$ & ${-3.68}$ \\ \hline
\end{tabular}
\label{table:results_mdt_vs_rdot}
\end{table} 

\begin{figure}[!t]
\centering
 {\resizebox{0.48\textwidth}{!}{\input{Generalization.tex}}}
\caption{BD-rates achieved for coding intra predicted blocks with the RDOT scheme based on KLT, {GL-GBST} and {GL-GBNT}, which are trained  on datasets with fewer number of samples. This experiment is conducted by randomly sampling 20\%, 40\%, 60\% and 80\% of the data from the original dataset and repeated 20 times to estimate average BD-rates. The BD-rates at 100\% correspond to the results in Table \ref{table:results_mdt_vs_rdot} where the complete training dataset is used.}
\label{fig:generalization}
\end{figure}
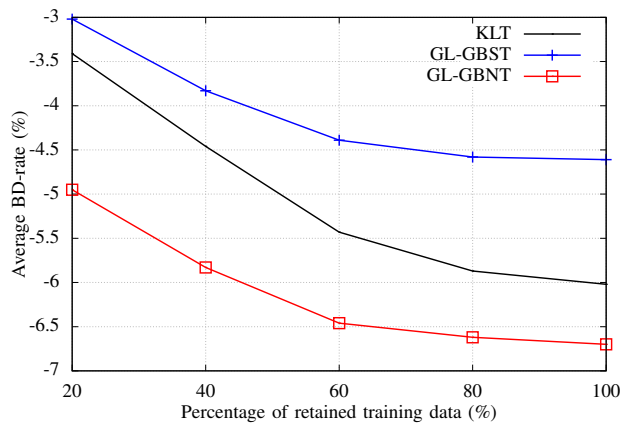

 Table \ref{table:results_gbt_vs_klt_modes} compares the RDOT coding performance of KLTs, {GL-GBST}s and {GL-GBNT}s on residual blocks with different prediction modes. In RDOT scheme the transform sets are $\mathcal{T}_{\text{KLT}}=\{\text{DCT},\text{KLT}\}$, $\mathcal{T}_{\text{{GL-GBST}}}=\{\text{DCT},\text{{GL-GBST}}\}$ and $\mathcal{T}_{\text{{GL-GBNT}}}=\{\text{DCT},\text{{GL-GBNT}}\}$, which consist of DCT and a trained transform for each mode and block size. The results show that {GL-GBNT} consistently outperforms KLT for all prediction modes. Similar to Table \ref{table:results_mdt_vs_rdot}, {GL-GBST} provides slightly better compression compared to KLT and {GL-GBST}. Also for angular modes (e.g., diagonal mode) in intra predicted coding, {GL-GBNT} significantly outperforms {GL-GBST} as expected.

 Table \ref{table:results_effect_of_eagbt} demonstrates the RDOT coding performance of EA-GBTs for different modes {by comparing against GL-GBT (corresponding to using GL-GBNT for intra and GL-GBST for inter predicted blocks in $\mathcal{T}_{\text{GL-GBT}} = \{ \text{DCT},\text{GL-GBT}\}$)}. As shown in the table, the contribution of EA-GBT within the transform set $\mathcal{T}_{\text{GL-GBT+EA-GBT}} = \{ \text{DCT},\text{GL-GBT},\text{EA-GBT}\}$ is limited to $0.3\%$ for intra predicted coding, while it is approximately $0.8\%$ for inter coding. On the other hand, if the transform set is selected as $\mathcal{T}_{\text{EA-GBT}} = \{ \text{DCT},\text{EA-GBT}\}$ the contribution of EA-GBT provides considerable coding gains, which are approximately $0.5\%$ for intra and $1\%$ for inter predicted coding. {Moreover, the percentages of transforms selected using RDOT with the sets  $\mathcal{T}_{\text{GL-GBT}}$,  $\mathcal{T}_{\text{EA-GBT}}$ and $\mathcal{T}_{\text{GL-GBT+EA-GBT}}$ are  presented in Table \ref{table:percentage}. 
 The most frequently used transforms are GL-GBT and  DCT for intra and inter predicted residuals, respectively. Besides, EA-GBT in $\mathcal{T}_{\text{EA-GBT}}$ is used about 10\% of the blocks for both intra and inter predicted residuals, and it is less frequently used within the set $\mathcal{T}_{\text{GL-GBT+EA-GBT}}$. Figure \ref{fig:sample_frames} further illustrates the transform selection from $\mathcal{T}_{\text{GL-GBT+EA-GBT}}$ 
 in two sample frames of the \emph{BasketballDrill} sequence, where EA-GBT is selected in blocks with sharp discontinuities and GL-GBTs are often selected for blocks with more directional structures as compared to DCT.
 }

\section{Conclusions}
 \label{sec:conclusion}
In this work, we discuss the class of transforms, called graph-based transforms (GBTs), with their applications to video compression. 
In particular, separable and nonseparable GBTs
are introduced and two different design strategies are proposed. Firstly, the GBT design problem is posed as a graph learning problem, where we estimate graphs from data and the resulting graphs are used to define GBTs (GL-GBTs).
Secondly, we propose edge-adaptive GBTs (EA-GBTs) 
which can be adapted on a per-block basis using side-information (image edges in a given block). We also give theoretical justifications for these two strategies and show that well-known transforms such as DCTs and DSTs are special cases of GBTs, and graphs can be used to design generalized (e.g., DCT-like or DST-like) separable transforms.  
Our experiments demonstrate that GL-GBTs can provide considerable coding gains with respect to standard transform coding schemes using DCT. In comparison with the Karhunen-Loeve transform (KLT), GL-GBTs are more robust and provide better generalization. 
Although coding gains obtained by including EA-GBTs 
in addition to GL-GBTs in the RDOT scheme are limited, using EA-GBTs only provides considerable coding gains over DCT. 

{Future work includes the development of GBTs on state-of-the-art video codecs. Along this line of research, Egilmez \emph{et.~al.~}\cite{egilmez:2020:paramteric_gbst_arxiv} have recently introduced a practical implementation of GL-GBTs by replacing the DST-7 and DCT-8 in multiple transform selection (MTS) scheme of VVC \cite{Bross:20:vvc10} with separable GBTs\footnote{{The work in \cite{egilmez:2020:paramteric_gbst_arxiv} was submitted and accepted to the IEEE International Conference on Image Processing \cite{egilmez:2020:paramteric_gbst_icip} while the present paper is under review.}}. The experiments conducted on the VVC reference software (VTM) show that considerable coding improvements (about 0.4\% in BD-rate reduction on average) can be achieved without any impact on encoder and decoder run-times. Thus, the findings in \cite{egilmez:2020:paramteric_gbst_arxiv} further validate our empirical and theoretical results in this paper and demonstrate the relevance of GBTs in a state-of-the-art codec such as VVC.}

\begin{table*}[!ht]
\centering
\caption{Comparison of KLT, {GL-GBST} and {GL-GBNT} for coding of different prediction modes in terms of BD-rate with respect to the DCT. Smaller (negative) BD-rates mean better compression.}
\resizebox{\textwidth}{!}{
\begin{tabular}{|c|C{1.7cm}|C{1.7cm}|C{1.7cm}|C{1.7cm}|C{1.7cm}|C{1.7cm}|C{1.9cm}|C{1.7cm}|}
\hline
\multirow{2}{*}{Transform Set} & \multicolumn{5}{c|}{Intra Prediction} & \multicolumn{3}{c|}{Inter Prediction} \\ 
\cline{2-9}
& {Planar} &{DC} & {Diagonal} &{Horizontal} & {All modes} &  {Square} & {Rectangular} & {All modes} \\ 
\hline
$\mathcal{T}_{\text{KLT}}$ & ${-5.79}$ &${-4.57}$ & ${-7.68}$& ${-6.14}$ & ${-6.02}$ &  ${-3.47}$ & ${-2.93}$ & ${-3.35}$ \\ 
\hline
$\mathcal{T}_{\text{{GL-GBST}}}$ & ${-5.45}$ &${-4.12}$ & ${-3.32}$ & ${-6.45}$ & ${-4.61}$ &  $\mathbf{-3.95}$ & $\mathbf{-3.25}$ & $\mathbf{-3.89}$ \\ 
\hline
$\mathcal{T}_{\text{{GL-GBNT}}}$ & $\mathbf{-6.27}$ & $\mathbf{-5.04}$ & $\mathbf{-8.74}$ & $\mathbf{-6.53}$ & $\mathbf{-6.70}$ &  ${-3.84}$ & ${-3.18}$ & ${-3.68}$ \\ 
\hline
\end{tabular}
}
\label{table:results_gbt_vs_klt_modes}
\end{table*}

\begin{table*}[!ht]
\centering
\caption{The contribution of GL-GBTs and EA-GBTs in terms of BD-rate with respect to the DCT. {GL-GBT corresponds to using GL-GBNT for intra blocks and GL-GBST for inter blocks.}}
\resizebox{\textwidth}{!}
{
\begin{tabular}{|c|C{1.7cm}|C{1.7cm}|C{1.7cm}|C{1.7cm}|C{1.7cm}|C{1.7cm}|C{1.9cm}|C{1.7cm}|}
\hline
\multirow{2}{*}{Transform Set} & \multicolumn{5}{c|}{Intra Prediction} & \multicolumn{3}{c|}{Inter Prediction} \\ 
\cline{2-9}
& {Planar} &{DC} & {Diagonal} &{Horizontal} & {All modes} &  {Square} & {Rectangular} & {All modes} \\ 
\hline
$\mathcal{T}_{\text{GL-GBT}}$ & ${-6.27}$ & ${-5.04}$ & ${-8.74}$ & ${-6.53}$ & ${-6.70}$ &  ${-3.95}$ & ${-3.25}$ & ${-3.89}$ \\ 
\hline
$\mathcal{T}_{\text{EA-GBT}}$ & ${-0.51}$ & ${-0.47}$ & ${-0.69}$ & ${-0.66}$ & ${-0.54}$ &  ${-1.01}$ & ${-0.73}$ & ${-0.93}$ \\ 
\hline
$\mathcal{T}_{\text{GL-GBT+EA-GBT}}$ & ${-6.58}$ & ${-5.34}$ & ${-9.07}$ & ${-6.89}$ & ${-7.01}$ &  ${-4.80}$ & ${-3.65}$ & ${-4.73}$ \\ 
\hline
\end{tabular}
}
\label{table:results_effect_of_eagbt}
\end{table*}

\begin{table*}[!ht]
\centering
\caption{{Percentage of transforms selected among transform coded blocks by the RDOT scheme with different transform sets. GL-GBT corresponds to using GL-GBNT for intra blocks and GL-GBST for inter blocks.}}
\resizebox{0.6\textwidth}{!}
{
\begin{tabular}{|c|C{1.2cm}|C{1.2cm}|C{1.2cm}|C{1.2cm}|C{1.2cm}|C{1.2cm}|}
\hline
\multirow{2}{*}{Transform Set} & \multicolumn{3}{c|}{Intra Prediction} & \multicolumn{3}{c|}{Inter Prediction} \\ 
\cline{2-7}
& {DCT} &{GL-GBT} & {EA-GBT} &  {DCT} &{GL-GBT} & {EA-GBT} \\ 
\hline
$\mathcal{T}_{\text{GL-GBT}}$         & $42\%$ & $58\%$ & $-$ &  $73\%$ & $27\%$ & $-$\\ 
\hline
$\mathcal{T}_{\text{EA-GBT}}$         & $89\%$ & $-$ & $11\%$ &  $87\%$ & $-$ & $13\%$ \\ 
\hline
$\mathcal{T}_{\text{GL-GBT+EA-GBT}}$  & $39\%$ & $57\%$ & $4\%$ &  $66\%$ & $24\%$ & $10\%$ \\ 
\hline
\end{tabular}
}
\label{table:percentage}
\end{table*}

\begin{figure*}[htbp!]
\centering
\subfloat[All-intra coded frame (I-frame) \label{fig:sample_frame_intra}]{
{\includegraphics[clip,width=.45\textwidth]{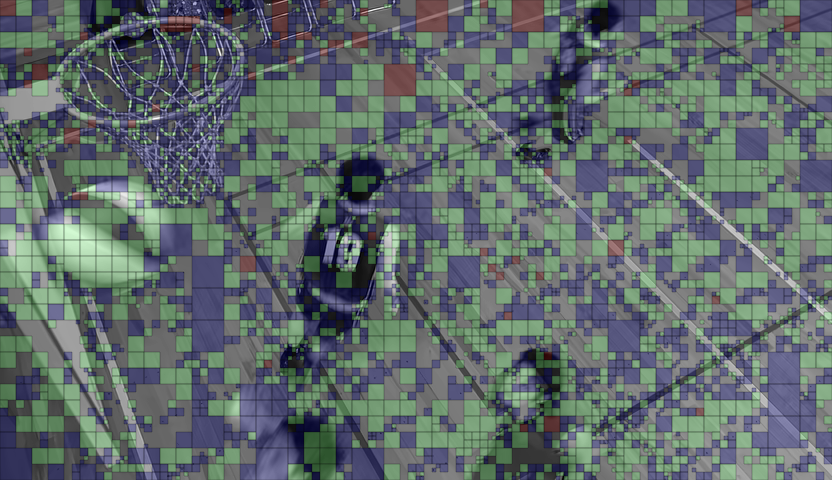}}}
\subfloat[Both intra and inter coded frame (B-frame)\label{fig:sample_frame_inter}]{
{\includegraphics[clip,width=.45\textwidth]{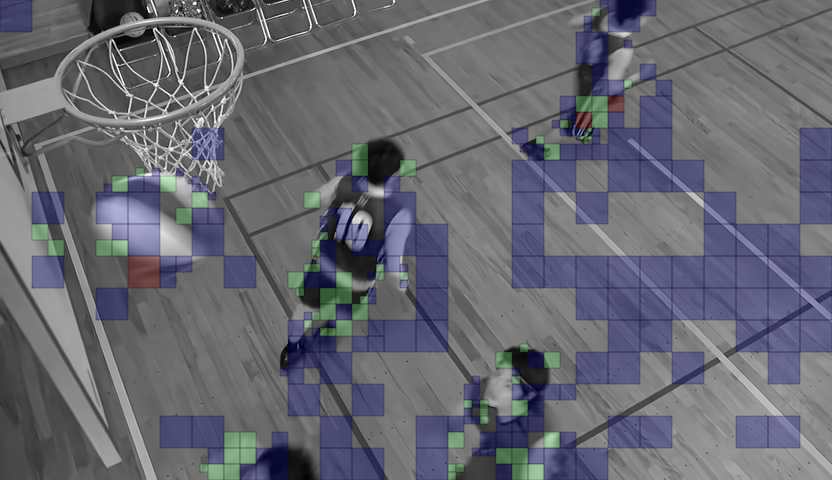}}}
\caption{{Illustration of selected transforms from the set $\mathcal{T}_{\text{GL-GBT+EA-GBT}} = \{ \text{DCT},\text{GL-GBT},\text{EA-GBT}\}$ by the RDOT scheme for each transform coded block. DCT, GL-GBT and EA-GBT are highlighted in purple, green and red, respectively. Most of the blocks in (b) are not transform coded (i.e., transform step is skipped) due to precise inter prediction.}}
\label{fig:sample_frames}
\end{figure*}


\ifCLASSOPTIONcaptionsoff
  \newpage
\fi




%
\bibliographystyle{IEEEtran}
\bibliography{refs.bib}

%
%
%








\end{document}

%% file: cg_sharp2.tex
\begingroup
  \makeatletter
  \providecommand\color[2][]{%
    \GenericError{(gnuplot) \space\space\space\@spaces}{%
      Package color not loaded in conjunction with
      terminal option `colourtext'%
    }{See the gnuplot documentation for explanation.%
    }{Either use 'blacktext' in gnuplot or load the package
      color.sty in LaTeX.}%
    \renewcommand\color[2][]{}%
  }%
  \providecommand\includegraphics[2][]{%
    \GenericError{(gnuplot) \space\space\space\@spaces}{%
      Package graphicx or graphics not loaded%
    }{See the gnuplot documentation for explanation.%
    }{The gnuplot epslatex terminal needs graphicx.sty or graphics.sty.}%
    \renewcommand\includegraphics[2][]{}%
  }%
  \providecommand\rotatebox[2]{#2}%
  \@ifundefined{ifGPcolor}{%
    \newif\ifGPcolor
    \GPcolortrue
  }{}%
  \@ifundefined{ifGPblacktext}{%
    \newif\ifGPblacktext
    \GPblacktextfalse
  }{}%
  \let\gplgaddtomacro\g@addto@macro
  \gdef\gplbacktext{}%
  \gdef\gplfronttext{}%
  \makeatother
  \ifGPblacktext
    \def\colorrgb#1{}%
    \def\colorgray#1{}%
  \else
    \ifGPcolor
      \def\colorrgb#1{\color[rgb]{#1}}%
      \def\colorgray#1{\color[gray]{#1}}%
      \expandafter\def\csname LTw\endcsname{\color{white}}%
      \expandafter\def\csname LTb\endcsname{\color{black}}%
      \expandafter\def\csname LTa\endcsname{\color{black}}%
      \expandafter\def\csname LT0\endcsname{\color[rgb]{1,0,0}}%
      \expandafter\def\csname LT1\endcsname{\color[rgb]{0,1,0}}%
      \expandafter\def\csname LT2\endcsname{\color[rgb]{0,0,1}}%
      \expandafter\def\csname LT3\endcsname{\color[rgb]{1,0,1}}%
      \expandafter\def\csname LT4\endcsname{\color[rgb]{0,1,1}}%
      \expandafter\def\csname LT5\endcsname{\color[rgb]{1,1,0}}%
      \expandafter\def\csname LT6\endcsname{\color[rgb]{0,0,0}}%
      \expandafter\def\csname LT7\endcsname{\color[rgb]{1,0.3,0}}%
      \expandafter\def\csname LT8\endcsname{\color[rgb]{0.5,0.5,0.5}}%
    \else
      \def\colorrgb#1{\color{black}}%
      \def\colorgray#1{\color[gray]{#1}}%
      \expandafter\def\csname LTw\endcsname{\color{white}}%
      \expandafter\def\csname LTb\endcsname{\color{black}}%
      \expandafter\def\csname LTa\endcsname{\color{black}}%
      \expandafter\def\csname LT0\endcsname{\color{black}}%
      \expandafter\def\csname LT1\endcsname{\color{black}}%
      \expandafter\def\csname LT2\endcsname{\color{black}}%
      \expandafter\def\csname LT3\endcsname{\color{black}}%
      \expandafter\def\csname LT4\endcsname{\color{black}}%
      \expandafter\def\csname LT5\endcsname{\color{black}}%
      \expandafter\def\csname LT6\endcsname{\color{black}}%
      \expandafter\def\csname LT7\endcsname{\color{black}}%
      \expandafter\def\csname LT8\endcsname{\color{black}}%
    \fi
  \fi
  \setlength{\unitlength}{0.0500bp}%
  \begin{picture}(6802.00,4534.00)%
    \gplgaddtomacro\gplbacktext{%
      \csname LTb\endcsname%
      \put(462,594){\makebox(0,0)[r]{\strut{}-8}}%
      \csname LTb\endcsname%
      \put(462,1210){\makebox(0,0)[r]{\strut{}-6}}%
      \csname LTb\endcsname%
      \put(462,1826){\makebox(0,0)[r]{\strut{}-4}}%
      \csname LTb\endcsname%
      \put(462,2443){\makebox(0,0)[r]{\strut{}-2}}%
      \csname LTb\endcsname%
      \put(462,3059){\makebox(0,0)[r]{\strut{} 0}}%
      \csname LTb\endcsname%
      \put(462,3675){\makebox(0,0)[r]{\strut{} 2}}%
      \csname LTb\endcsname%
      \put(594,374){\makebox(0,0){\strut{} 0}}%
      \csname LTb\endcsname%
      \put(1320,374){\makebox(0,0){\strut{} 0.1}}%
      \csname LTb\endcsname%
      \put(2047,374){\makebox(0,0){\strut{} 0.2}}%
      \csname LTb\endcsname%
      \put(2773,374){\makebox(0,0){\strut{} 0.3}}%
      \csname LTb\endcsname%
      \put(3500,374){\makebox(0,0){\strut{} 0.4}}%
      \csname LTb\endcsname%
      \put(4226,374){\makebox(0,0){\strut{} 0.5}}%
      \csname LTb\endcsname%
      \put(4952,374){\makebox(0,0){\strut{} 0.6}}%
      \csname LTb\endcsname%
      \put(5679,374){\makebox(0,0){\strut{} 0.7}}%
      \csname LTb\endcsname%
      \put(6405,374){\makebox(0,0){\strut{} 0.8}}%
      \put(1320,4093){\makebox(0,0){\strut{}10}}%
      \put(2047,4093){\makebox(0,0){\strut{}5}}%
      \put(2773,4093){\makebox(0,0){\strut{}3.33}}%
      \put(3500,4093){\makebox(0,0){\strut{}2.5}}%
      \put(4226,4093){\makebox(0,0){\strut{}2}}%
      \put(4952,4093){\makebox(0,0){\strut{}1.66}}%
      \put(5679,4093){\makebox(0,0){\strut{}1.43}}%
      \put(6405,4093){\makebox(0,0){\strut{}1.25}}%
      \put(187,2288){\rotatebox{-270}{\makebox(0,0){\strut{}Coding gain ($\mathsf{cg}$)}}}%
      \put(3499,154){\makebox(0,0){\strut{}$w_e/w_c = 1/s_{\text{edge}}$}}%
      \put(3499,4422){\makebox(0,0){\strut{}$s_{\text{edge}}$}}%
    }%
    \gplgaddtomacro\gplfronttext{%
      \csname LTb\endcsname%
      \put(5418,1647){\makebox(0,0)[r]{\strut{}$N=4$}}%
      \csname LTb\endcsname%
      \put(5418,1427){\makebox(0,0)[r]{\strut{}$N=8$}}%
      \csname LTb\endcsname%
      \put(5418,1207){\makebox(0,0)[r]{\strut{}$N=16$}}%
      \csname LTb\endcsname%
      \put(5418,987){\makebox(0,0)[r]{\strut{}$N=32$}}%
      \csname LTb\endcsname%
      \put(5418,767){\makebox(0,0)[r]{\strut{}$N=64$}}%
    }%
    \gplbacktext
    \put(0,0){\includegraphics{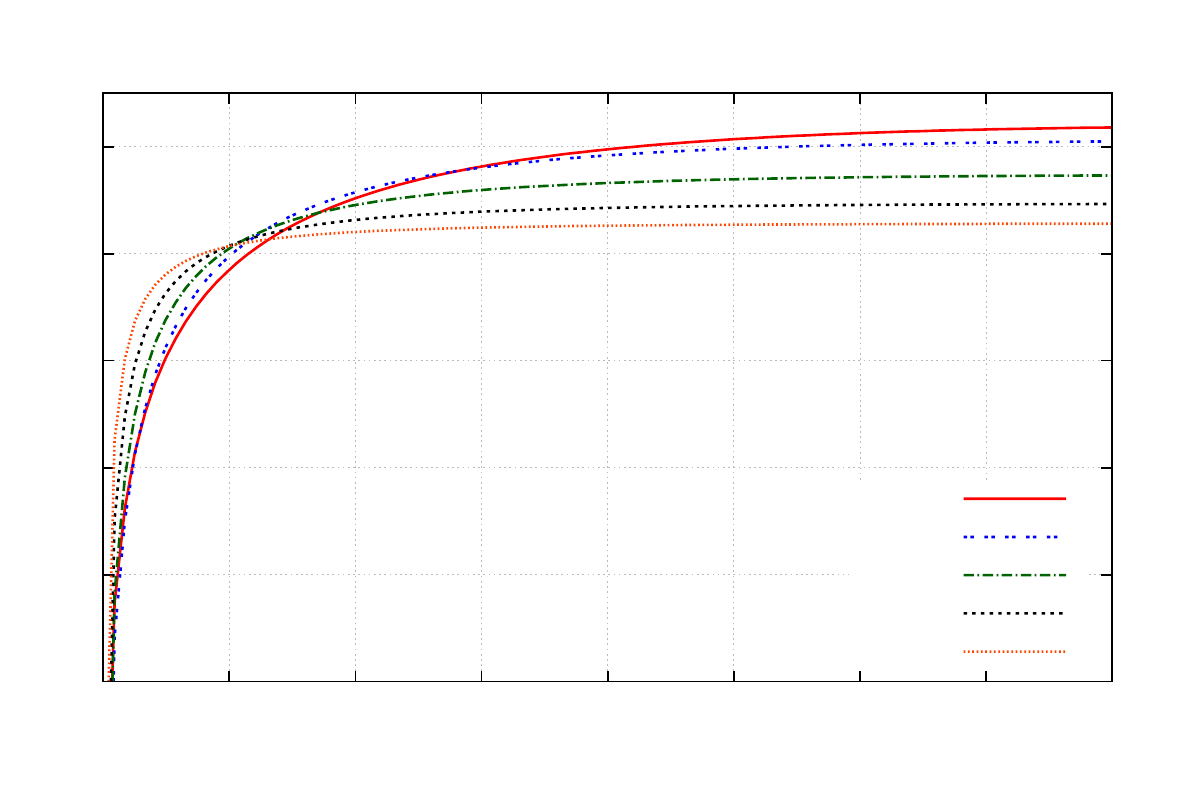}}%
    \gplfronttext
  \end{picture}%
\endgroup

%% file: cg_waterfill.tex
\begingroup
  \makeatletter
  \providecommand\color[2][]{%
    \GenericError{(gnuplot) \space\space\space\@spaces}{%
      Package color not loaded in conjunction with
      terminal option `colourtext'%
    }{See the gnuplot documentation for explanation.%
    }{Either use 'blacktext' in gnuplot or load the package
      color.sty in LaTeX.}%
    \renewcommand\color[2][]{}%
  }%
  \providecommand\includegraphics[2][]{%
    \GenericError{(gnuplot) \space\space\space\@spaces}{%
      Package graphicx or graphics not loaded%
    }{See the gnuplot documentation for explanation.%
    }{The gnuplot epslatex terminal needs graphicx.sty or graphics.sty.}%
    \renewcommand\includegraphics[2][]{}%
  }%
  \providecommand\rotatebox[2]{#2}%
  \@ifundefined{ifGPcolor}{%
    \newif\ifGPcolor
    \GPcolortrue
  }{}%
  \@ifundefined{ifGPblacktext}{%
    \newif\ifGPblacktext
    \GPblacktextfalse
  }{}%
  \let\gplgaddtomacro\g@addto@macro
  \gdef\gplbacktext{}%
  \gdef\gplfronttext{}%
  \makeatother
  \ifGPblacktext
    \def\colorrgb#1{}%
    \def\colorgray#1{}%
  \else
    \ifGPcolor
      \def\colorrgb#1{\color[rgb]{#1}}%
      \def\colorgray#1{\color[gray]{#1}}%
      \expandafter\def\csname LTw\endcsname{\color{white}}%
      \expandafter\def\csname LTb\endcsname{\color{black}}%
      \expandafter\def\csname LTa\endcsname{\color{black}}%
      \expandafter\def\csname LT0\endcsname{\color[rgb]{1,0,0}}%
      \expandafter\def\csname LT1\endcsname{\color[rgb]{0,1,0}}%
      \expandafter\def\csname LT2\endcsname{\color[rgb]{0,0,1}}%
      \expandafter\def\csname LT3\endcsname{\color[rgb]{1,0,1}}%
      \expandafter\def\csname LT4\endcsname{\color[rgb]{0,1,1}}%
      \expandafter\def\csname LT5\endcsname{\color[rgb]{1,1,0}}%
      \expandafter\def\csname LT6\endcsname{\color[rgb]{0,0,0}}%
      \expandafter\def\csname LT7\endcsname{\color[rgb]{1,0.3,0}}%
      \expandafter\def\csname LT8\endcsname{\color[rgb]{0.5,0.5,0.5}}%
    \else
      \def\colorrgb#1{\color{black}}%
      \def\colorgray#1{\color[gray]{#1}}%
      \expandafter\def\csname LTw\endcsname{\color{white}}%
      \expandafter\def\csname LTb\endcsname{\color{black}}%
      \expandafter\def\csname LTa\endcsname{\color{black}}%
      \expandafter\def\csname LT0\endcsname{\color{black}}%
      \expandafter\def\csname LT1\endcsname{\color{black}}%
      \expandafter\def\csname LT2\endcsname{\color{black}}%
      \expandafter\def\csname LT3\endcsname{\color{black}}%
      \expandafter\def\csname LT4\endcsname{\color{black}}%
      \expandafter\def\csname LT5\endcsname{\color{black}}%
      \expandafter\def\csname LT6\endcsname{\color{black}}%
      \expandafter\def\csname LT7\endcsname{\color{black}}%
      \expandafter\def\csname LT8\endcsname{\color{black}}%
    \fi
  \fi
  \setlength{\unitlength}{0.0500bp}%
  \begin{picture}(6802.00,4534.00)%
    \gplgaddtomacro\gplbacktext{%
      \csname LTb\endcsname%
      \put(594,594){\makebox(0,0)[r]{\strut{}-10}}%
      \csname LTb\endcsname%
      \put(594,1053){\makebox(0,0)[r]{\strut{}-8}}%
      \csname LTb\endcsname%
      \put(594,1513){\makebox(0,0)[r]{\strut{}-6}}%
      \csname LTb\endcsname%
      \put(594,1972){\makebox(0,0)[r]{\strut{}-4}}%
      \csname LTb\endcsname%
      \put(594,2432){\makebox(0,0)[r]{\strut{}-2}}%
      \csname LTb\endcsname%
      \put(594,2891){\makebox(0,0)[r]{\strut{} 0}}%
      \csname LTb\endcsname%
      \put(594,3350){\makebox(0,0)[r]{\strut{} 2}}%
      \csname LTb\endcsname%
      \put(594,3810){\makebox(0,0)[r]{\strut{} 4}}%
      \csname LTb\endcsname%
      \put(594,4269){\makebox(0,0)[r]{\strut{} 6}}%
      \csname LTb\endcsname%
      \put(726,374){\makebox(0,0){\strut{} 0.5}}%
      \csname LTb\endcsname%
      \put(2146,374){\makebox(0,0){\strut{} 0.75}}%
      \csname LTb\endcsname%
      \put(3566,374){\makebox(0,0){\strut{} 1}}%
      \csname LTb\endcsname%
      \put(4985,374){\makebox(0,0){\strut{} 1.25}}%
      \csname LTb\endcsname%
      \put(6405,374){\makebox(0,0){\strut{} 1.5}}%
      \put(187,2431){\rotatebox{-270}{\makebox(0,0){\strut{}Coding gain ($\mathsf{cg}$)}}}%
      \put(3565,154){\makebox(0,0){\strut{}$R/N$}}%
    }%
    \gplgaddtomacro\gplfronttext{%
      \csname LTb\endcsname%
      \put(5418,4096){\makebox(0,0)[r]{\strut{}$s_{\text{edge}}=10$}}%
      \csname LTb\endcsname%
      \put(5418,3876){\makebox(0,0)[r]{\strut{}$s_{\text{edge}}=20$}}%
      \csname LTb\endcsname%
      \put(5418,3656){\makebox(0,0)[r]{\strut{}$s_{\text{edge}}=40$}}%
      \csname LTb\endcsname%
      \put(5418,3436){\makebox(0,0)[r]{\strut{}$s_{\text{edge}}=100$}}%
      \csname LTb\endcsname%
      \put(5418,3216){\makebox(0,0)[r]{\strut{}$s_{\text{edge}}=200$}}%
    }%
    \gplbacktext
    \put(0,0){\includegraphics{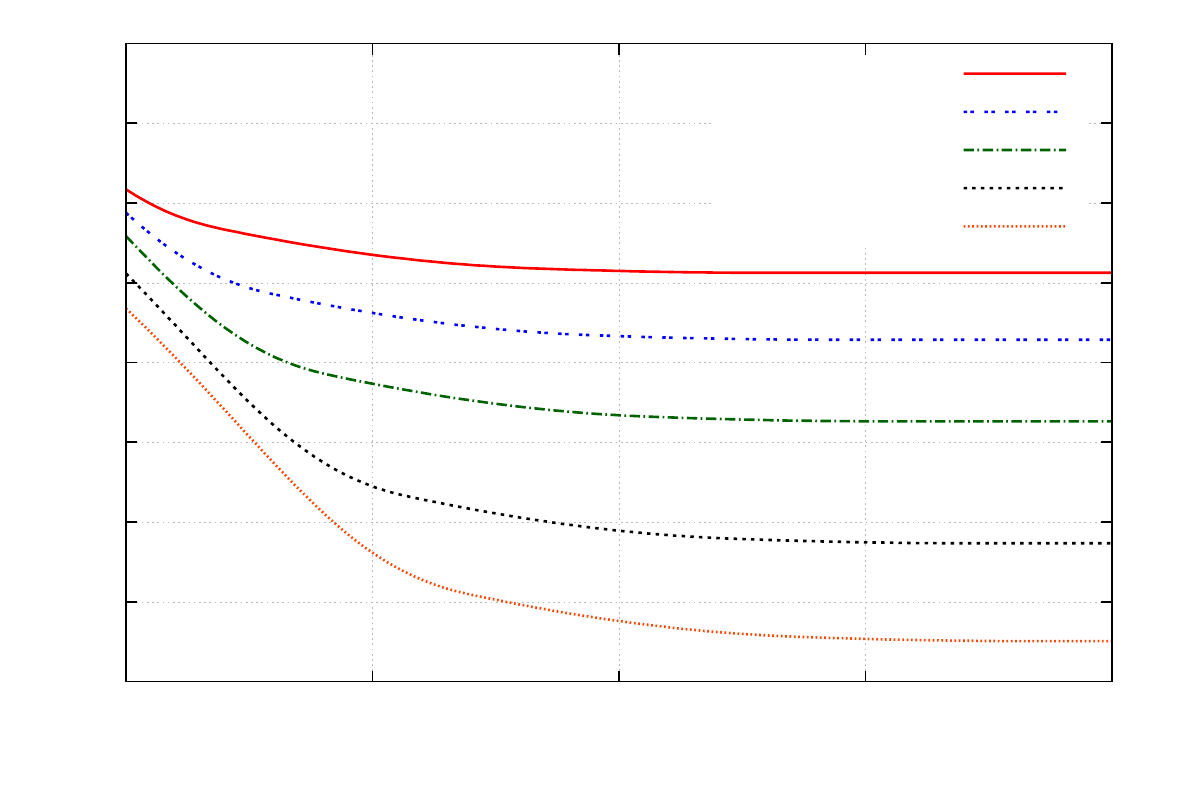}}%
    \gplfronttext
  \end{picture}%
\endgroup

%% file: Generalization.tex
\begingroup
  \makeatletter
  \providecommand\color[2][]{%
    \GenericError{(gnuplot) \space\space\space\@spaces}{%
      Package color not loaded in conjunction with
      terminal option `colourtext'%
    }{See the gnuplot documentation for explanation.%
    }{Either use 'blacktext' in gnuplot or load the package
      color.sty in LaTeX.}%
    \renewcommand\color[2][]{}%
  }%
  \providecommand\includegraphics[2][]{%
    \GenericError{(gnuplot) \space\space\space\@spaces}{%
      Package graphicx or graphics not loaded%
    }{See the gnuplot documentation for explanation.%
    }{The gnuplot epslatex terminal needs graphicx.sty or graphics.sty.}%
    \renewcommand\includegraphics[2][]{}%
  }%
  \providecommand\rotatebox[2]{#2}%
  \@ifundefined{ifGPcolor}{%
    \newif\ifGPcolor
    \GPcolortrue
  }{}%
  \@ifundefined{ifGPblacktext}{%
    \newif\ifGPblacktext
    \GPblacktextfalse
  }{}%
  \let\gplgaddtomacro\g@addto@macro
  \gdef\gplbacktext{}%
  \gdef\gplfronttext{}%
  \makeatother
  \ifGPblacktext
    \def\colorrgb#1{}%
    \def\colorgray#1{}%
  \else
    \ifGPcolor
      \def\colorrgb#1{\color[rgb]{#1}}%
      \def\colorgray#1{\color[gray]{#1}}%
      \expandafter\def\csname LTw\endcsname{\color{white}}%
      \expandafter\def\csname LTb\endcsname{\color{black}}%
      \expandafter\def\csname LTa\endcsname{\color{black}}%
      \expandafter\def\csname LT0\endcsname{\color[rgb]{1,0,0}}%
      \expandafter\def\csname LT1\endcsname{\color[rgb]{0,1,0}}%
      \expandafter\def\csname LT2\endcsname{\color[rgb]{0,0,1}}%
      \expandafter\def\csname LT3\endcsname{\color[rgb]{1,0,1}}%
      \expandafter\def\csname LT4\endcsname{\color[rgb]{0,1,1}}%
      \expandafter\def\csname LT5\endcsname{\color[rgb]{1,1,0}}%
      \expandafter\def\csname LT6\endcsname{\color[rgb]{0,0,0}}%
      \expandafter\def\csname LT7\endcsname{\color[rgb]{1,0.3,0}}%
      \expandafter\def\csname LT8\endcsname{\color[rgb]{0.5,0.5,0.5}}%
    \else
      \def\colorrgb#1{\color{black}}%
      \def\colorgray#1{\color[gray]{#1}}%
      \expandafter\def\csname LTw\endcsname{\color{white}}%
      \expandafter\def\csname LTb\endcsname{\color{black}}%
      \expandafter\def\csname LTa\endcsname{\color{black}}%
      \expandafter\def\csname LT0\endcsname{\color{black}}%
      \expandafter\def\csname LT1\endcsname{\color{black}}%
      \expandafter\def\csname LT2\endcsname{\color{black}}%
      \expandafter\def\csname LT3\endcsname{\color{black}}%
      \expandafter\def\csname LT4\endcsname{\color{black}}%
      \expandafter\def\csname LT5\endcsname{\color{black}}%
      \expandafter\def\csname LT6\endcsname{\color{black}}%
      \expandafter\def\csname LT7\endcsname{\color{black}}%
      \expandafter\def\csname LT8\endcsname{\color{black}}%
    \fi
  \fi
  \setlength{\unitlength}{0.0500bp}%
  \begin{picture}(6802.00,4534.00)%
    \gplgaddtomacro\gplbacktext{%
      \csname LTb\endcsname%
      \put(726,594){\makebox(0,0)[r]{\strut{}-7}}%
      \csname LTb\endcsname%
      \put(726,1053){\makebox(0,0)[r]{\strut{}-6.5}}%
      \csname LTb\endcsname%
      \put(726,1513){\makebox(0,0)[r]{\strut{}-6}}%
      \csname LTb\endcsname%
      \put(726,1972){\makebox(0,0)[r]{\strut{}-5.5}}%
      \csname LTb\endcsname%
      \put(726,2432){\makebox(0,0)[r]{\strut{}-5}}%
      \csname LTb\endcsname%
      \put(726,2891){\makebox(0,0)[r]{\strut{}-4.5}}%
      \csname LTb\endcsname%
      \put(726,3350){\makebox(0,0)[r]{\strut{}-4}}%
      \csname LTb\endcsname%
      \put(726,3810){\makebox(0,0)[r]{\strut{}-3.5}}%
      \csname LTb\endcsname%
      \put(726,4269){\makebox(0,0)[r]{\strut{}-3}}%
      \csname LTb\endcsname%
      \put(858,374){\makebox(0,0){\strut{}20}}%
      \csname LTb\endcsname%
      \put(2245,374){\makebox(0,0){\strut{}40}}%
      \csname LTb\endcsname%
      \put(3632,374){\makebox(0,0){\strut{}60}}%
      \csname LTb\endcsname%
      \put(5018,374){\makebox(0,0){\strut{}80}}%
      \csname LTb\endcsname%
      \put(6405,374){\makebox(0,0){\strut{}100}}%
      \put(319,2431){\rotatebox{-270}{\makebox(0,0){\strut{}Average BD-rate (\%)}}}%
      \put(3631,154){\makebox(0,0){\strut{}Percentage of retained training data (\%)}}%
    }%
    \gplgaddtomacro\gplfronttext{%
      \csname LTb\endcsname%
      \put(5418,4096){\makebox(0,0)[r]{\strut{}KLT}}%
      \csname LTb\endcsname%
      \put(5418,3876){\makebox(0,0)[r]{\strut{}GL-GBST}}%
      \csname LTb\endcsname%
      \put(5418,3656){\makebox(0,0)[r]{\strut{}GL-GBNT}}%
    }%
    \gplbacktext
    \put(0,0){\includegraphics{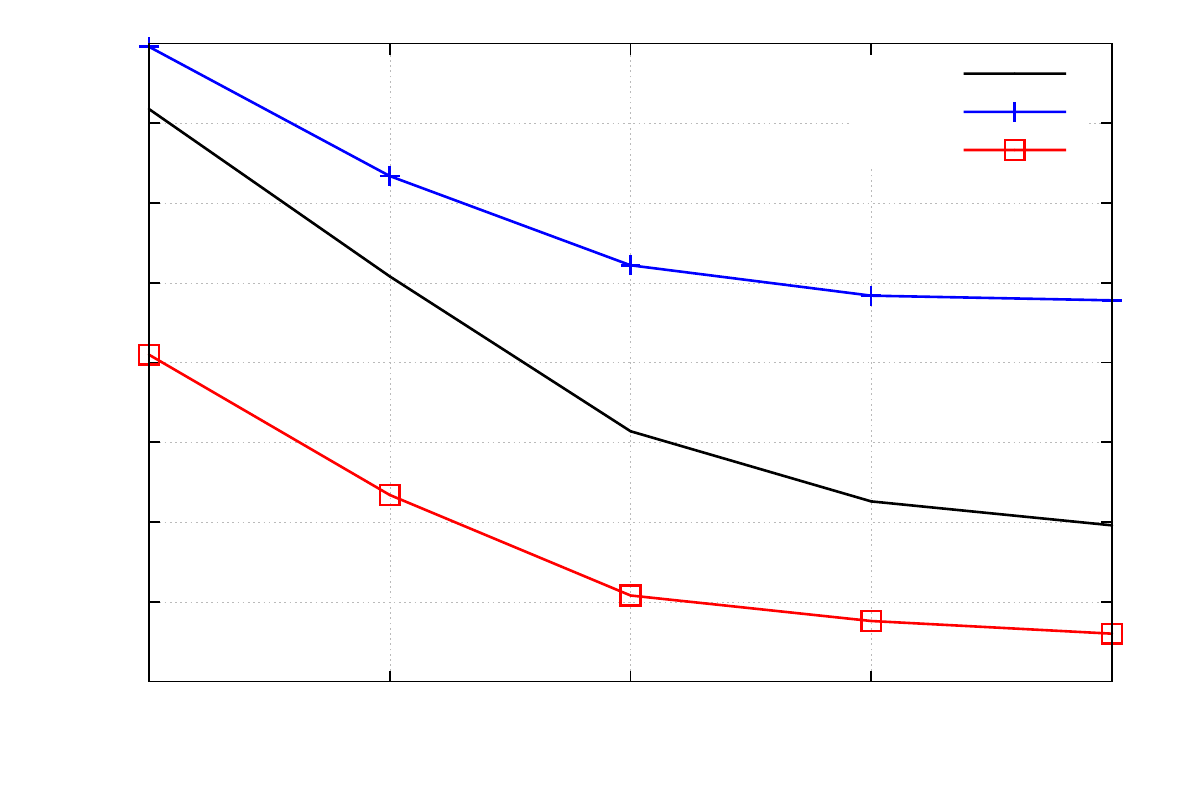}}%
    \gplfronttext
  \end{picture}%
\endgroup